\documentclass[journal]{IEEEtran}
\ifCLASSINFOpdf
\else
\fi

\usepackage{graphicx, amsmath,amsfonts,cite,amssymb,xcolor, enumerate, mathrsfs, amsthm,url, mathtools} 
\usepackage[normalem]{ulem}
\usepackage[hidelinks]{hyperref}
\hypersetup{
    colorlinks=true,       
    linkcolor=blue,          
    citecolor=green,        
    urlcolor=blue           
}
\usepackage{booktabs} 
\usepackage{array} 
\usepackage{paralist} 
\usepackage{verbatim} 
\usepackage{subcaption}
\usepackage{epstopdf}

\usepackage{bbm}
\usepackage{pgfgantt}
\usepackage{algpseudocode}
\usepackage{sidecap}

\usepackage[font=small,labelfont=bf,labelsep=space]{caption}
\newcommand{\un}{\underline}
\newcommand{\be}{\begin{equation}}
\newcommand{\ee}{\end{equation}}
\newcommand{\ben}{\begin{equation*}}
\newcommand{\een}{\end{equation*}}
\newcommand{\mc}{\mathcal}
\newcommand{\mbf}{\mathbf}

\newtheorem{lem}{Lemma}
\newtheorem{applem}{Lemma}[section]

\newtheorem{thm}{Theorem}

\newcommand{\e}{\epsilon}

\newcommand{\abs}[1]{\left \lvert #1\right \rvert}
\newcommand{\norm}[1]{\left \lVert #1\right \rVert}
\newcommand{\expec}{\mathbb{E}}

\newcommand{\mscrs}{\mathscr{S}}

\newcommand{\midd}{\textsf{mid}\,}
\newcommand{\edge}{\textsf{edge}\,}
\newcommand{\pure}{\mathsf{pure}}
\newcommand{\tot}{\mathsf{total}}

\graphicspath{{img/}}

\begin{document}
%
\title{Analysis of Approximate Message Passing\\
with Non-Separable Denoisers and\\
 Markov Random Field Priors}

%

\author{Yanting~Ma,~\IEEEmembership{Member,~IEEE,}
       Cynthia~Rush,~\IEEEmembership{Member,~IEEE,}
        and~Dror~Baron,~\IEEEmembership{Senior Member,~IEEE}
\thanks{Y. Ma is with Mitsubishi Electric Research Laboratories (MERL), Cambridge, MA 02138, USA (e-mail: yma@merl.com).}
\thanks{C. Rush is with the Department of Statistics, Columbia University, New York, NY 10027, USA (e-mail: cynthia.rush@columbia.edu).}
\thanks{D. Baron is with the Department of Electrical and Computer Engineering, North Carolina State University, Raleigh, NC 27695, USA (e-mail: barondror@ncsu.edu).}%
\thanks{This work was supported by the National Science Foundation (NSF) under grants CCF-1217749 and ECCS-1611112.}%
\thanks{Portions of the work appeared at the IEEE International Symposium on Information Theory (ISIT), Aachen, Germany, June 2017 \cite{MaRushBaron17}.}%
\thanks{The work was completed while Y. Ma was with North Carolina State University.}
}

\maketitle

\begin{abstract}
Approximate message passing (AMP) is a class of low-complexity, scalable algorithms for solving high-dimensional linear regression tasks where one wishes to recover an unknown signal from noisy, linear measurements.  AMP is an iterative algorithm that performs estimation by updating an estimate of the unknown signal at each iteration and the performance of AMP (quantified, for example, by the mean squared error of its estimates) depends on the choice of a ``denoiser" function that is used to produce these signal estimates at each iteration.

An attractive feature of AMP is that its performance can be tracked by a scalar recursion referred to as state evolution.   Previous theoretical analysis of the accuracy of the state evolution predictions has been limited to the use of only separable denoisers or block-separable denoisers, a class of denoisers that underperform when sophisticated dependencies exist between signal entries.  Since signals with entrywise dependencies are common in image/video-processing applications, in this work we study the high-dimensional linear regression task when the dependence structure of the input signal is modeled by a Markov random field prior distribution.
We provide a rigorous analysis of the performance of AMP, demonstrating the accuracy of the state evolution predictions, when a class of non-separable sliding-window denoisers is applied.  Moreover, we provide numerical examples where AMP with sliding-window denoisers can successfully capture local dependencies in images.
\end{abstract}


%


\section{Introduction} \label{sec:intro}

In this work, we study the problem of estimating an unknown signal $\beta:=(\beta_i)_{i\in\Gamma}$ from noisy, linear measurements as in the following model:
\be
y=A\mc{V}(\beta)+w,
\label{eq:model1}
\ee
where for some integer $p\in\{1,2,3\}$, $\Gamma\subset\mathbb{Z}^p$ is an index set with cardinality $|\Gamma|$, $y \in \mathbb{R}^n$ is the output, $A \in \mathbb{R}^{n \times |\Gamma|}$ is a known measurement matrix, $w \in \mathbb{R}^n$ is zero-mean noise with finite variance $\sigma^2$, and $\mc{V}$ (script $V$ stands for ``vectorization") is an invertible operator that rearranges elements of an array into a vector, hence $\mc{V}(\beta)$ is a length-$|\Gamma|$ vector.  We assume that the ratio of the dimensions of the measurement matrix is a constant value, $\delta := n/|\Gamma|$, with $\delta \in (0, \infty)$. 

Approximate message passing (AMP) \cite{DonMalMont09, MontChap11, BayMont11, krz12, Rangan11} is a class of low-complexity, scalable algorithms studied to solve the high-dimensional regression task of \eqref{eq:model1}.  The performance of AMP depends on a sequence of functions $\{\eta_t\}_{t\geq 0}$ used to generate a sequence of estimates $\{\beta^t\}_{t\geq 0}$ from effective observations computed in every iteration of the algorithm. A nice property of AMP is that under some technical conditions these observations can be approximated as the input signal $\beta$ plus independent and identically distributed (i.i.d.)\ Gaussian noise. 
For this reason, the functions $\{\eta_t\}_{t\geq 0}$ are referred to as ``denoisers."

Previous analysis of the performance of AMP only considers denoisers $\{\eta_t\}_{t\geq 0}$ that act coordinate-wise when applied to a vector; such denoisers are referred to as \emph{separable}. If the unknown signal $\beta$ has a prior distribution with
i.i.d.\ entries, restricting consideration to only separable denoisers causes no loss in performance.
However, in many real-world applications, the unknown signal $\beta$ contains dependencies between entries, and therefore a coordinate-wise independence structure does not approximate the prior for $\beta$ well. 
Instead of using a separable denoiser, {\em non-separable} denoisers can improve reconstruction quality for signals with such dependencies among entries. For example, when the signals are images \cite{Tan15,Metzler16} or sound clips \cite{Ma16}, non-separable
denoisers outperform reconstruction techniques based on over-simplified i.i.d.\ models.
In such cases, a more appropriate model might be a finite memory model, well-approximated with a Markov random field (MRF) prior.
In this paper, we extend the previous performance guarantees for AMP to a class of non-separable sliding-window denoisers when the unknown signal is a realization of an MRF. Sliding-window schemes have been studied for denoising signals with dependencies among entries by, for example, Sivaramakrishnan and Weissman \cite{SW_iDUDE2006,Sivaramakrishnan2008}.
MRFs are appropriate models for many types of images, especially texture images, which have an inherently random component \cite{dubes1989,cross1983}.

When the measurement matrix $A$ has i.i.d.\ Gaussian entries and the empirical distribution function\footnote{For a vector $(x_1,\ldots,x_N)\in\mathbb{R}^N$, the empirical distribution function $F_N:\mathbb{R}\to [0,1]$ is defined as $F_N(x):=\frac{1}{N}\sum_{i=1}^N \mathbb{I}_{\{x_i\leq x\}}$, where $\mathbb{I}$ denotes the indicator function. The empirical distribution function $F_N$ is said to converge to some distribution function $F:\mathbb{R}\to [0,1]$ if for all $x\in\mathbb{R}$ such that $F(x)$ is continuous, we have $\lim_{N\to\infty} F_N(x)=F(x)$.} of the unknown signal $\beta$ converges to some distribution function on $\mathbb{R}$, Bayati and Montanari \cite{BayMont11} proved 
that at each iteration
the performance of AMP can be accurately predicted by a simple, scalar iteration referred to as \emph{state evolution} in the large system limit ($n,|\Gamma| \to \infty$ such that $(n/|\Gamma|) \rightarrow \delta$ 
is a constant).  For example, if $\beta^t$ is the estimate produced by AMP at iteration $t$, the result by Bayati and Montanari \cite{BayMont11} implies that the normalized squared error, $\frac{1}{|\Gamma|}\norm{\beta^t - \beta}^2$, and other performance measures converge to deterministic values predicted by state evolution, which is a deterministic recursion calculated using the prior distribution of $\beta$.\footnote{Throughout the paper, $\|v\|^2$ denotes the sum of squares of all the entries in $v$, where $v$ could be, for example, in $\mathbb{R}^n$, $\mathbb{R}^{n\times n}$, or $\mathbb{R}^{n\times n\times n}$.}
Rush and Venkataramanan \cite{RushV18} provided a concentration version of the asymptotic result when the prior distribution of $\beta$ is i.i.d.\ sub-Gaussian.  The result in Rush and Venkataramanan \cite{RushV18} implies that
the probability of $\e$-deviation between various performance measures and their limiting constant values decay exponentially in $|\Gamma|$. 

Extensions of AMP performance guarantees beyond separable denoisers have been considered in special cases  
\cite{JavMonState13,RushGVISIT15} for
certain classes of block-separable denoisers that allow dependencies within blocks of the signal $\beta$ with independence across blocks.  A preliminary version \cite{MaRushBaron17} of this work has analyzed the performance of AMP with sliding-window denoisers applied to the setting where the unknown signal has a Markov chain prior. In this paper, we generalize the previous result with the applications of compressive imaging \cite{Tan15,Metzler16} and compressive hyperspectral imaging \cite{Tan.etal2015} in mind.  We consider 2D/3D MRF priors for the input signal $\beta$, and provide performance guarantees for AMP with 2D/3D sliding-window denoisers under some technical conditions. 

While we were concluding this manuscript, we became aware of recent work of Berthier \emph{et al.}\ \cite{Berthier17}. The authors prove that the loss of the estimates generated by AMP (for a class of loss functions) with general non-separable denoisers converges to the state evolution predictions asymptotically. Our work differs from \cite{Berthier17} in the following three aspects: (\emph{i}) our work provides finite sample analysis, whereas the result in \cite{Berthier17} is asymptotic; (\emph{ii}) we adjust the state evolution sequence for the specific class of non-separable sliding-window denoisers to account for the ``edge" issue that occurs in the finite sample regime (this point will become clear in later sections); (\emph{iii}) we consider the setting where the unknown signal is a realization of an MRF and the expectation in the definition of the state evolution sequence is with respect to (w.r.t.) the signal $\beta$, the matrix $A$, and the noise $w$, whereas in \cite{Berthier17}, the signal $\beta$ is deterministic and unknown, hence the expectation is only w.r.t. the matrix $A$ and the noise $w$.

\subsection{Sliding-Window Denoisers and AMP Algorithm}
\label{subsec:amp_algo}

\textbf{Notation:}
Before introducing the algorithm, we provide some notation that is used to define the sliding window in the sliding-window denoiser. Without loss of generality, we let the index set $\Gamma\subset\mathbb{Z}^p$, on which the input signal $\beta$ in \eqref{eq:model1} is defined, be
\begin{equation}
\Gamma=\begin{cases}
[N], &\text{if }p=1,\\
[N]\times [N], &\text{if }p=2,\\
[N]\times [N]\times [N], &\text{if }p=3,
\end{cases}
\label{eq:def_Gamma}
\end{equation}
where for an integer $N$, the notation $[N]$ represents the set of integers $\{1,\ldots,N\}$, hence, $|\Gamma|=N^p$. Similarly, let $\Lambda$ be a $p$-dimensional cube in $\mathbb{Z}^p$ with length $(2k+1)$ in each dimension, namely, 
\begin{equation}
\Lambda:=\begin{cases}
[2k+1], &\text{if }p=1,\\
[2k+1]\times [2k+1], &\text{if }p=2,\\
[2k+1]\times [2k+1]\times [2k+1], &\text{if }p=3,
\end{cases}
\label{eq:def_Lambda}
\end{equation}
where $2k+1 \leq N$. We call $k$ the half-window size.

\textbf{AMP with sliding-window denoisers:} 
The AMP algorithm for estimating $\beta$ from $y$ and $A$ in \eqref{eq:model1} generates a sequence of estimates $\{\beta^t\}_{t\geq 0}$, where $\beta^t\in\mathbb{R}^{\Gamma}$, $t$ is the iteration index, and  the initialization $\beta^0:=0$ is an all-zero array with the same dimension as the input signal $\beta$.
For $t\geq 0$, the algorithm proceeds as follows:
\begin{align}
z^t &= y-A\mc{V}(\beta^t)\nonumber\\
&\quad+\frac{z^{t-1}}{n}\sum_{i\in\Gamma} \eta_{t-1}'\left(\left[\mc{V}^{-1}(A^*z^{t-1})+\beta^{t-1}\right]_{\Lambda_i}\right),\label{eq:amp1}\\
\beta_i^{t+1} &=\eta_t\left([\mc{V}^{-1}(A^*z^t)+\beta^t]_{\Lambda_i}\right), \qquad \text{ for all } i\in\Gamma, \label{eq:amp2}
\end{align}
where the function $\{\eta_t\}_{t\geq 0}: \mathbb{R}^{\Lambda} \rightarrow \mathbb{R}$ is a sequence of denoisers, $\eta_{t-1}'$ is the partial derivative w.r.t. the center coordinate of the argument, $A^*$ is the transpose of $A$, and $\Lambda_i\subset\mathbb{Z}^p$ for each $i\in \Gamma\subset\mathbb{Z}^p$ is the $p$-dimensional cube $\Lambda$ translated to be centered at location $i$. The translated $p$-dimensional cubes $\{\Lambda_i\}_{i\in\Gamma}$ are referred to as ``sliding windows," which will be used to subset elements of a $p$-dimensional array. The effective observation at iteration $t$ is $\mc{V}^{-1}(A^*z^t)+\beta^t \in \mathbb{R}^{\Gamma}$, which can be approximated as the true signal $\beta$ plus i.i.d.\ Gaussian noise (in a sense that will be made clear in the statement of our main result, Theorem \ref{thm:main_amp_perf}). Note that the sliding-windows  $\{\Lambda_i\}_{i\in\Gamma}$ and the sliding-window denoiser $\eta_t$ are defined on multidimensional signals, hence we use the inverse of the vectorization operator, $\mc{V}^{-1}$, to rearrange elements of vectors into arrays before applying the sliding-window denoiser $\eta_t$. 
It should also be noted that the denoiser $\eta_t$ may only process part of the signal elements in $\Lambda$. 
For example, in the 2D case, if $\Lambda$ is defined as a $3 \times 3$ window, then $\eta_t$ may only process the center and the four adjacent pixel values in the window (see Figure \ref{fig:partial_window}) and ignore the four corners.
To simplify notation, we will write $\eta_t:\mathbb{R}^\Lambda\to\mathbb{R}$ throughout the paper, and interpret this notation to mean that any processing of neighboring signal values is allowed, including the possibility of ignoring some of their values.
\begin{figure}
\centering
\includegraphics[width=0.2\textwidth]{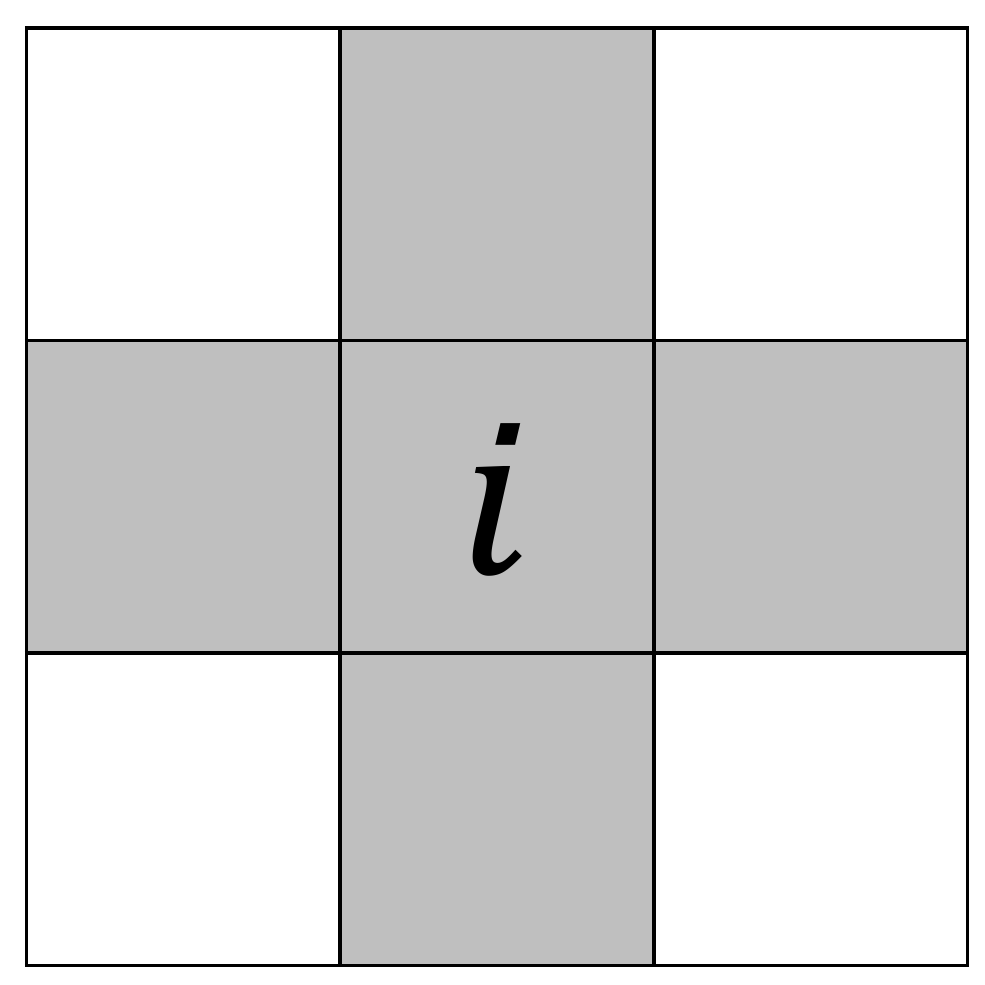}
\caption{For $\Lambda$ of size $3\times 3$, the denoiser $\eta_t:\mathbb{R}^\Lambda\to\mathbb{R}$ may only process the pixels in gray (the center and the four adjacent pixels).}
\label{fig:partial_window}
\end{figure}

\textbf{Edge cases:}
Notice that when the center coordinate $i$ is near the edge of $\Gamma$, some of the elements in $\Lambda_i$ may fall outside $\Gamma$, meaning that $\Lambda_i\cap\Gamma^c\neq\emptyset$ where $\Gamma^c$ is the complement of $\Gamma$ w.r.t. $\mathbb{Z}^p$. In the definition of the AMP algorithm with sliding-window denoisers and the subsequent analysis, these ``edge cases" must be handled carefully.  The following definitions provide a framework for the special treatment of the edge cases.  

Based on whether $\Lambda_i$ has elements outside $\Gamma$, we partition the index set $\Gamma$ into two sets $\Gamma^{\midd}$ and $\Gamma^{\edge}$ defined as:
\begin{equation}
\begin{split}
\Gamma^{\midd}&:=\{ i\in\Gamma \, \vert \, \Lambda_i\cap\Gamma^c =\emptyset\},\\
\Gamma^{\edge}&:=\{i\in\Gamma \, \vert \, \Lambda_i\cap\Gamma^c\neq\emptyset\}.
\end{split}
\label{eq:def_edge_index_set}
\end{equation}
That is, for $i\in\Gamma^{\midd}$, all elements in $\Lambda_i$ lie inside $\Gamma$, whereas for $i\in\Gamma^{\edge}$, some of the elements in $\Lambda_i$ fall outside $\Gamma$. The size of each set will depend on the half-window size $k$ and the dimension $p$.

For any $v\in\mathbb{R}^{\Gamma}$,  let $v_{\Lambda_i}$ be a subset of the elements of $v$ with indices in $\Lambda_i$ and for any $j \in \Lambda$, let $[\Lambda_i]_j$ be the $j^{th}$ index of $\Lambda_i$ so that $v_{[\Lambda_i]_j}$ returns a single element of $v_{\Lambda_i}$. Notice that for $i\in\Gamma^{\midd}$, all entries of $v_{\Lambda_i}$ are well-defined. However, for $i\in\Gamma^{\edge}$, the subset $v_{\Lambda_i}$ has undefined entries, namely, for all $j \in \Lambda$ such that $[\Lambda_i]_j\in\Gamma^c$, the entry $v_{[\Lambda_i]_j}$ is undefined. We now define the value of those ``missing" entries to be the average of the entries of $v_{\Lambda_i}$ having indices in $\Gamma$. Formally, for all $ j \in \Lambda$ such that $[\Lambda_i]_j \in\Gamma^c$, define
\begin{equation}
v_{[\Lambda_i]_j}:=\frac{1}{|\Lambda_i\cap\Gamma|}\sum_{ \ell \,  \in \, \Lambda_i\cap\Gamma} v_{\ell}.
\label{eq:def_missing_entry}
\end{equation}
It may improve signal recovery quality to use other schemes for ``missing" entries,
like interpolation. We leave the study of these improved schemes for future work, but
the effect of improved processing around edges should become minor as $N$ increases.
Notice that $v_{\Lambda_i}$ for all $i\in\Gamma$ are now defined using only the entries in the original $v\in\mathbb{R}^{\Gamma}$. 
It will be useful to emphasize this point in the proof of our main result, so we define a set of functions $\{\mc{T}_i\}_{i\in\Gamma}$ with $\mc{T}_i:\mathbb{R}^{\Lambda_i\cap\Gamma}\to\mathbb{R}^{\Lambda}$ as
\begin{equation}
\mc{T}_i(v_{\Lambda_i\cap\Gamma}):= v_{\Lambda_i},\qquad\text{ for all } \qquad i\in\Gamma,
\label{eq:def_T}
\end{equation}
where $v_{\Lambda_i}$ follows our definition above.
That is, $\mc{T}_i$ is identity for $i\in\Gamma^{\midd}$, whereas for $i\in\Gamma^{\edge}$, $\mc{T}_i$ extends a smaller array $v_{\Lambda_i\cap\Gamma}$ to a larger one $v_{\Lambda_i}$ with the extended entries defined by \eqref{eq:def_missing_entry}.

\textbf{Examples for defining ``missing" entries:}
To illustrate the notations defined above, we present an example for the $p=1$ case (hence $v\in\mathbb{R}^N$ is a vector). As defined above in \eqref{eq:def_Gamma} and \eqref{eq:def_Lambda}, we have $\Gamma=\{1,\ldots,N\}$, $\Lambda=\{1, \ldots, 2k+1\}$, and $\Lambda_i=(i-k,\ldots,i-1,i,i+1,\ldots, i+k)$ for each $i \in [N]$. Moreover, $\Gamma^{\midd}=\{k+1, k+2, \ldots, N-k\}$ and $\Gamma^{\edge}=\{1, 2, \ldots, k\} \cup \{N-k+1, N-k+2, \ldots, N\}$ as defined in \eqref{eq:def_edge_index_set}.
Therefore, for $i \in \Gamma^{\midd}$,
\ben
\begin{split}
v_{\Lambda_i}&=\mc{T}_i(v_{i-k}, v_{i-k+1}, \ldots, v_{i+k})\\
& := (v_{i-k}, v_{i-k+1}, \ldots, v_{i+k}) \in \mathbb{R}^{2k+1}.
\end{split}
\een
For $i \in \Gamma^{\edge}$, the vector $v_{\Lambda_i}$ is still length-$(2k+1)$, and we set the values of the non-positive indices, i.e., $1-k, 2-k, \ldots, -1, 0$, or indices above $N$, i.e., $N+1, N+2, \ldots, N+k$, to be the average of values in the vector $v_{\Lambda_i}$ with indices in $\Lambda_i\cap [N]$.  For example, let $i=3$ and $k=5$ giving $\Lambda_3=(-2, -1, 0, 1 ,\ldots,8)$ so that for $j \in \{1, 2, 3\}$ we have $[\Lambda_3]_j \in \Gamma^c$.  Following \eqref{eq:def_missing_entry}, define
\begin{equation*}
\bar{v} = \frac{1}{8}\sum_{j=1}^8 v_j, \quad \text{ and set } \quad v_{[\Lambda_3]_1} = v_{[\Lambda_3]_2} = v_{[\Lambda_3]_3} = \bar{v},
\end{equation*}
and so $v_{\Lambda_3} =\mc{T}_3(v_1, \ldots, v_8) := (\bar{v}, \bar{v}, \bar{v}, v_1, \ldots, v_{8})\in \mathbb{R}^{11}$.  An example for the $p=2$ case (hence $v\in\mathbb{R}^{N\times N}$ is a matrix), is shown in Figure \ref{fig:def_missing_entry}.

\begin{figure}
\centering
\includegraphics[width=0.3\textwidth]{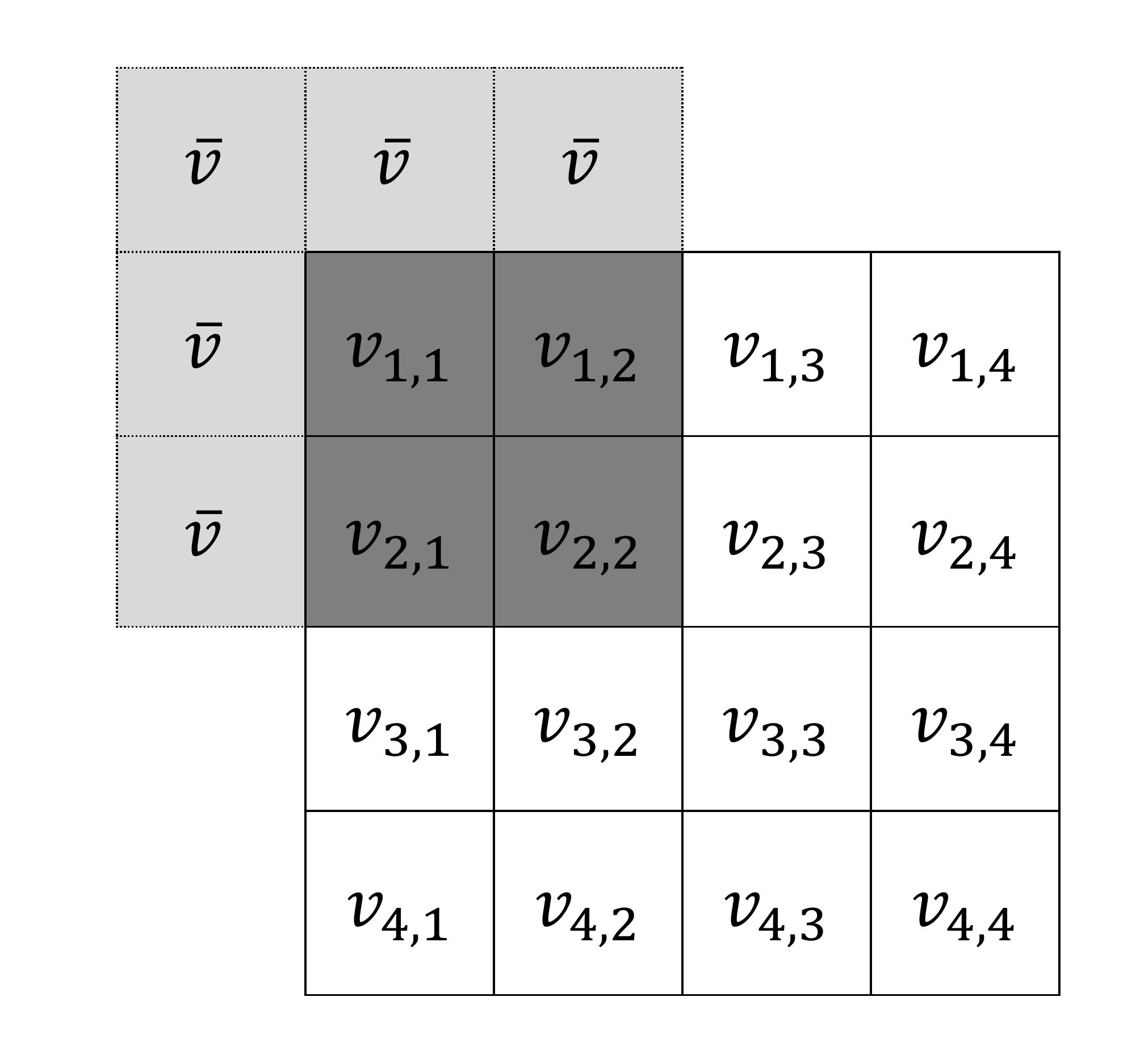}
\caption{Illustration of the definition of ``missing" entries in a window. The matrix $v\in\mathbb{R}^{4\times 4}$. The half-window size is $k=1$, thus $\Lambda=[3]\times[3]$. For the window $\Lambda_{(1,1)}$ centered at $(1,1)$, the ``existing" entries in the window are $v_{1,1},v_{1,2},v_{2,1},v_{2,2}$ shown in dark gray. Five entries, which are in light gray, are missing, hence we define their value to be the average of the existing ones, $\bar{v}:=\frac{1}{4}(v_{1,1}+v_{1,2}+v_{2,1}+v_{2,2})$.}
\label{fig:def_missing_entry}
\end{figure}

\subsection{Contributions and Outline}

Our main result proves concentration for (order-2) pseudo-Lipschitz (PL(2)) loss functions\footnote{A function $f: \mathbb{R}^{m} \to \mathbb{R}$ is (order-2) \emph{pseudo-Lipschitz} if there exists a constant $L >0$ such that for all $x, y \in \mathbb{R}^{m}$, $|f(x)-f(y)|\leq L(1+\norm{x}+\norm{y})\norm{x-y}$.} acting on the AMP estimate given in \eqref{eq:amp2} at any iteration $t$ of the algorithm to constant values predicted by the state evolution equations that will be introduced in the following.  
This work covers the case where the unknown signal $\beta$ has an MRF prior on $\mathbb{Z}^p$. For example, when $p=2$, $\beta$ can be thought of as an image, whereas when $p=3$, $\beta$ can be thought of as a hyperspectral image cube.  Moreover we use numerical examples to demonstrate the effectiveness of AMP with sliding-window denoisers when used to reconstruct images from noisy linear measurements.

The rest of the paper is organized as follows. Section \ref{sec:mainresult} provides model assumptions, state evolution formulas, the main performance guarantee, and numerical examples illustrating the effectiveness of the algorithm for compressive image reconstruction. Our main performance guarantee (Theorem \ref{thm:main_amp_perf}) is a concentration result for PL loss functions acting on the AMP outputs from \eqref{eq:amp1}-\eqref{eq:amp2}
to the state evolution predictions. Section \ref{sec:amp_proof} provides the proof of Theorem \ref{thm:main_amp_perf}. The proof is based on a technical lemma, Lemma \ref{lem:main_lem}, and the proof of Lemma \ref{lem:main_lem} is provided in Section \ref{sec:main_lem_proof}.

\section{Main Results} \label{sec:mainresult}

\subsection{Definitions and Assumptions}
\label{subsec:def_assumption}
First we include some definitions relating to MRFs that will be used to state our assumptions on the unknown signal $\beta$.
These definitions can be found in standard textbooks such as \cite{Georgii98}; we include them here for convenience.

\textbf{Definitions}: 
Let $(\Omega,\mathcal{F},P)$ be a probability space. A random field is a collection of random variables $X=\{X_i\}_{i\in\Gamma}$ defined on $(\Omega,\mathcal{F},P)$ having spatial dependencies, where $X_i:\Omega\to E$ for some measurable state space $(E,\mathcal{E})$ and $\Gamma \subset \mathbb{Z}^p$ is a non-empty,  finite subset of the infinite lattice $\mathbb{Z}^p$. Note that $i\in\Gamma\subset\mathbb{Z}^p$, hence $i=(i_1,\ldots,i_p)$. One can consider $\Gamma$ as a collection of spatial locations. Denote the $q^{th}$-order neighborhood of location $i\in\Gamma$ by $\mathcal{N}^q_i$, that is, $\mathcal{N}^q_i\subset\Gamma$ is a collection of location indices at a distance less than or equal to $q$ from $i$ but not including $i$. Formally,
\ben
\mathcal{N}^q_i = \{\, j \in \Gamma \setminus \{i\} \, \lvert \, \norm{i-j}^2 \leq q \, \}.
\een
Following these definitions, $X$ is said to be a $q^{th}$-order MRF if, for all $ i\in\Gamma$ and for all measurable subsets $B \in \mathcal{E}$, we have
\begin{equation*}
P(X_i\in B \, | \, X_j, \, j\in\Gamma\setminus \{i\}) = P(X_i\in B \, | \, X_j, \, j\in\mathcal{N}^q_i),
\end{equation*}
and for all $B \in \mathcal{E}^{\Gamma}$ we have $P(X \in B) > 0$.  The positivity condition ensures that the joint distribution of an MRF is a Gibbs distribution by the Hammersley-Clifford theorem \cite{HammersleyClifford1971}.

Let $\mu$ denote the distribution measure of $X$, namely for all $B\in\mathcal{E}^\Gamma$, we have
$P(X \in B) = \mu (B)$,
and let $\mu_\Lambda$ be the distribution measure of $X_\Lambda:=\{X_i\}_{i\in\Lambda}$ for $\Lambda\subset\Gamma$. For any $i \in \Gamma$, define the set $i+\Lambda:=\{i+j \, \lvert \, j\in\Lambda\}$. Then the random field is said to be \emph{stationary} if for all $i \in \Gamma$ such that $i + \Lambda \subset\Gamma$, it is true that $u_{\Lambda}=u_{i+\Lambda}$.

Next we introduce the \emph{Dobrushin uniqueness condition},  under which the random field admits a unique stationary distribution. Define the \emph{Dobrushin interdependence matrix} $(C_{i,j})_{i,j\in\Gamma}$ for the measure $\mu$ of the random field $X$ to be
\begin{equation}
C_{i,j} := \sup_{\substack{\xi ,\, \xi' \, \in \, E^\Gamma\\ \xi_{j^c}\, =\,\xi'_{j^c}}} \|\mu_i(\cdot \vert \xi) - \mu_i(\cdot | \xi')\|_{\mathsf{tv}}.
\label{eq:Dob_matrix_def}
\end{equation}
In the above, the index set $j^c:=\Gamma\setminus \{j\}$ and the total variation distance $\|\cdot\|_{\mathsf{tv}}$ between two probability measures $\rho_1$ and $\rho_2$ on $(E,\mc{E})$ is defined as
\begin{equation*}
\|\rho_1(\cdot) - \rho_2(\cdot)\|_{\mathsf{tv}}: = \max_{B\in\mc{E}} \abs{\rho_1(B) - \rho_2(B)}.
\end{equation*}
Note that if $E$ is countable, then
\begin{equation}
\|\rho_1(\cdot) - \rho_2(\cdot)\|_{\mathsf{tv}} = \frac{1}{2} \sum_{x\in E} \abs{\rho_1(x) - \rho_2(x)}.
\label{eq:tv_countable}
\end{equation}
The measure $\mu$ is said to satisfy the Dobrushin uniqueness condition if
\begin{equation*}
c:=\sup_{i \, \in \, \Gamma}\sum_{j \, \in \, \Gamma} C_{i,j} <1.
\end{equation*}
The \emph{Dobrushin contraction coefficient}, $c$, is a quantity that estimates the magnitude of change of the single site conditional expectations, as they appear in \eqref{eq:Dob_matrix_def}, when the field values at the other sites vary. Similarly, we define the \emph{transposed Dobrushin contraction condition} as
\begin{equation*}
c^*:= \sup_{j \,\in \,\Gamma}\sum_{i\,\in\,\Gamma} C_{i,j}<1.
\end{equation*}

\textbf{Assumptions}: We can now state our assumptions on the signal $\beta$, the matrix $A$, and the noise $w$ in the linear system \eqref{eq:model1}, as well as the denoiser function $\eta_t$ used in the algorithm \eqref{eq:amp1} and \eqref{eq:amp2}.

\emph{Signal:} Let $E\subset\mathbb{R}$ be a bounded state space (countable or uncountable). Let $\beta=(\beta_i)_{i\in\Gamma}$ be a stationary MRF with Gibbs distribution measure $\mu$ on $E^\Gamma$, where $\Gamma\subset\mathbb{Z}^p$ is a finite and nonempty rectangular lattice. We assume that $\mu$ satisfies the Dobrushin uniqueness condition and the transposed Dobrushin uniqueness condition. These two conditions together are needed for the results in Lemma \ref{lem:kulske} and Lemma \ref{lem:PL_MRF_conc}, which demonstrate concentration of sums of pseudo-Lipschitz functions when the input to the functions are MRFs with distribution measure $\mu$.
Roughly, the conditions ensure that the dependencies between the terms in the sums are sufficiently weak for the desired concentration to hold.
The class of finite state space stationary MRFs, which is widely used for image analysis \cite{Li2009}, is one example that satisfies our assumption.

\emph{Denoiser functions:} The denoiser functions $\eta_t:\mathbb{R}^{\Lambda}\rightarrow\mathbb{R}$ used in \eqref{eq:amp2} are assumed to be Lipschitz\footnote{A function $f: \mathbb{R}^{m} \to \mathbb{R}$ is \emph{Lipschitz} if there exists a constant $L >0$ such that for all $x, y \in \mathbb{R}^{m}$,
$\abs{f(x) - f(y)} \leq L\norm{x-y}$.} for each $t>0$ and are, therefore, also weakly differentiable with bounded (weak) partial derivatives. We further assume that the partial derivative w.r.t.\ the center coordinate of $\Lambda$, which is denoted by $\eta_t':\mathbb{R}^\Lambda\to\mathbb{R}$, is itself differentiable with bounded partial derivatives.
Note that this implies $\eta_t'$ is Lipschitz. (It is possible to weaken this condition to allow $\eta_t'$ to have a finite number of discontinuities, if needed, as in \cite{RushV18}.)

\emph{Matrix:} The entries of the matrix $A$ are i.i.d.\ $\sim\mc{N}(0,1/n)$.

\emph{Noise:} The entries of the measurement noise vector $w$ are i.i.d.\ according to some sub-Gaussian distribution $p_w$ with mean 0 and finite variance $\sigma^2$. The sub-Gaussian assumption implies \cite{BLMConc} that for all $\e \in (0,1)$ and for some constants $K,\kappa>0$,
\begin{equation*}
P\Big(\Big \lvert \frac{1}{n}\|w\|^2 - \sigma^2 \Big \lvert \geq \e\Big)\leq Ke^{-\kappa n \e^2}.
\end{equation*}

\subsection{Performance Guarantee}

As noted in Section \ref{sec:intro}, the behavior of the AMP algorithm is predicted by a deterministic scalar recursion referred to as state evolution, which we now introduce. More specifically, the state evolution sequences $\{\tau_t^2\}_{t\geq 0}$ and $\{\sigma_t^2\}_{t\geq 0}$ defined below in \eqref{eq:taut_sigmat_def} will be used in Theorem \ref{thm:main_amp_perf} to characterize the estimation error of the estimates produced by AMP.  Let the joint distribution $\mu$ define the (stationary) prior distribution for the unknown signal $\beta$ in \eqref{eq:model1}. 
Following our assumption of stationarity, $\beta_i\sim\mu_1$ for all $i\in\Gamma$ and $\beta_{\Lambda_i}\sim\mu_\Lambda$ for all $i\in\Gamma^{\midd}$ with $\Gamma^{\midd}$ defined in \eqref{eq:def_edge_index_set}, where $\mu_1$ and $\mu_\Lambda$ denote the one-dimensional marginal and $\Lambda$-dimensional marginal of $\mu$, respectively. Define $\sigma_\beta^2 = \mathbb{E}[\beta_1^2] > 0$, and 
$\sigma_0^2 = \sigma_\beta^2/\delta$. 
Iteratively define the state evolution sequences $\{\tau_t^2\}_{t \geq 0}$ and $\{\sigma_t^2\}_{t \geq 1}$ as follows:
\begin{equation}
\begin{split}
\tau_t^2 &= \sigma^2 + \sigma_t^2,\\
\sigma_t^2 &=\frac{1}{\delta|\Gamma|} \sum_{i\in\Gamma}\mathbb{E} \left[ \left( \eta_{t-1}( [\beta + \tau_{t-1} Z]_{\Lambda_i} ) - \beta_i \right)^2 \right],
\end{split}
\label{eq:taut_sigmat_def}
\end{equation}
where $\eta_t:\mathbb{R}^\Lambda\to\mathbb{R}$ is the sliding-window denoiser and $Z \in \mathbb{R}^{\Gamma}$ has i.i.d.\ standard normal entries, independent of $\beta$, which implies that $Z_{\Lambda_i}$ is independent of $\beta_{\Lambda_i}$ and $\beta_i$ for all $i\in\Gamma$. 
Let $\beta'\in E^{\Lambda}\sim\mu_\Lambda$ and define $Z'\in\mathbb{R}^\Lambda$ with entries that are i.i.d.\ $\mc{N}(0,1)$.  We notice that for all $i\in\Gamma^{\midd}$, we have $\beta_{\Lambda_i}\overset{d}{=} \beta'$ and $Z_{\Lambda_i}\overset{d}{=}Z'$. 
Therefore, for all $i\in\Gamma^{\midd}$, the expectations in \eqref{eq:taut_sigmat_def} satisfy 
$\mathbb{E} \left[ \left( \eta_{t-1}( [\beta + \tau_{t-1} Z]_{\Lambda_i} ) - \beta_i \right)^2 \right]=\mathbb{E} \left[ \left( \eta_{t-1}( \beta' + \tau_{t-1} Z' ) - \beta_c' \right)^2 \right]$,
where $\beta_c'$ is the center coordinate of $\beta'$. For $i\in\Gamma^{\edge}$ with $\Gamma^{\edge}$ defined in \eqref{eq:def_edge_index_set}, it is not necessarily true that $\beta_{\Lambda_i}\overset{d}{=}\beta'$ since, by definition \eqref{eq:def_missing_entry}, some entries of $\beta_{\Lambda_i}$ are defined as the average of other entries. 

The explicit expression for the definition of $\sigma_t^2$ in \eqref{eq:taut_sigmat_def} is different when considering $\Gamma\subset\mathbb{Z}^p$ for different $p$ values, as the size of the set $\Gamma^{\edge}$ depends on the dimension. In the following, we provide explicit expressions for $\sigma_t^2$ for the cases $p=1,2$, but in the proof we will use the general expression given in \eqref{eq:taut_sigmat_def} for brevity. We emphasize that the definition of the state evolution sequence in \eqref{eq:taut_sigmat_def} only uses the marginal distribution $\mu_{\Lambda}$ (or $\beta'\in E^\Lambda$) instead of the joint distribution $\mu$ (or $\beta\in E^{\Gamma}$), as demonstrated in the two examples below in \eqref{eq:SE1D} and \eqref{eq:SE2D}.

\textbf{Examples for explicit expressions for $\sigma_t^2$:} Let $\beta_c'$ be the center coordinate of $\beta'\in E^\Lambda$ and $\Lambda_{c}$ the window $\Lambda\subset\mathbb{Z}^p$ translated with center $c\in\mathbb{Z}^p$.  Recall that $\Lambda$ is the $p$-dimensional cube with length $(2k+1)$ in each of the $p$ dimensions.  Then we have $\beta'=\beta'_{\Lambda_c}$ and when we consider shifts $\beta'_{\Lambda_{c + \ell}}$ for $\ell \in \{-k, -k + 1, \ldots, k-1, k\}$ we, analogous to the definition in \eqref{eq:def_missing_entry}, define ``missing" entries to be replaced by the average of the existing entries.  (Note that since $\beta'$ is exactly of size $\Lambda$, thus for any $\ell \neq 0$, there will be ``missing" entries.)  For example, when $p=1$,
\begin{equation*}
\begin{split}
\beta'_{\Lambda_{c}} &= (\, \beta'_1 \,,\, \beta'_2 \,, \, \ldots \,,\, \beta'_{2k+1} \, ),\\
\beta'_{\Lambda_{c-2}}&= (\,\text{avg} \,,\, \text{avg} \,,\, \beta'_1 \,,\, \beta'_2 \,,\, \ldots \,,\, \beta'_{2k-1} \,),
\end{split}
\end{equation*}
where $\text{avg} = \frac{1}{2k-1} \sum_{i=1}^{2k-1} \beta'_i$.  Generalizing, we have $\beta'_{\Lambda_{c+\ell}} ~\in~\mathbb{R}^{2k+1}$ with 
\ben
\beta'_{\Lambda_{c+\ell}} \!=\! \begin{cases}
&\hspace*{-0.15in}\Big(\frac{1}{2k+1+\ell} \sum_{i=1}^{2k+1+\ell} \beta'_i, \ldots ,\\
&\hspace*{-0.15in}\frac{1}{2k+1+\ell} \sum_{i=1}^{2k+1+\ell} \beta'_i ,  \beta'_1 , \beta'_2 , \ldots, \beta'_{2k+1 + \ell}\Big)  \text{ if } \ell < 0, \\
&\hspace*{-0.15in}\Big(\beta'_1 \,,\, \beta'_2\,,\, \ldots \,,\, \beta'_{2k+1} \Big) \text{ if } \ell = 0, \\
&\hspace*{-0.15in}\Big(\beta'_{1+\ell} , \beta'_{2+\ell} , \ldots , \beta'_{2k+1} , \\
&\hspace*{-0.15in}\frac{1}{2k+1-\ell} \sum_{i=1 + \ell}^{2k+1} \beta'_i , \ldots ,  \frac{1}{2k+1-\ell} \sum_{i=1+\ell}^{2k+1} \beta'_i \Big)  \text{ if } \ell > 0.
\end{cases}
\een
The same idea can be extended when $p>1$.

For the case $p=1$, we note that $\Gamma^{\midd} = \{k+1, k+2, \ldots, N-k-1\}$ and $\Gamma^{\edge}= \{1, 2, \ldots, k\} \cup \{N-k, N-k+1, \ldots, N\}$, hence $|\Gamma^{\midd}|=N-2k$ and $|\Gamma^{\edge}|=2k$.  Therefore, we have
\begin{equation}
\begin{split}
&\sigma_t^2 =\frac{(N-2k)}{\delta N}  \mathbb{E} \left[ \left( \eta_{t-1}( \beta' + \tau_{t-1} Z' ) - \beta_c' \right)^2 \right]\\
&+ \frac{1}{\delta N}  \sum_{\ell \in \mc{K}_0} \mathbb{E} \left[ \left(\eta_{t-1}( [\beta' + \tau_{t-1} Z']_{\Lambda_{c+\ell}} ) - \beta_{c+\ell}' \right)^2\right],
\end{split}
\label{eq:SE1D}
\end{equation}
where $\mc{K}_0 = \{-k, \ldots, -1\} \cup \{1, \ldots, k\}$.
In the above the first term corresponds to the $N-2k$ middle indices, while the second term sums over $2k$ terms corresponding to all the possible edge cases.

For the case $p=2$, we note that $\Gamma^{\midd} = \{(i,j)  \, \lvert \, k+1 \leq i, j \leq N-k+1\}$, hence $|\Gamma^{\midd}|=(N-2k)^2$.  Here we note $\ell = (\ell_1, \ell_2) \in \{-k , -k+1, \ldots, k-1, k\}\times \{-k , -k+1, \ldots, k-1, k\}$.  Therefore,

\begin{align}
&\sigma_t^2 =  \frac{(N-2k)^2}{\delta N^2}  \mathbb{E} \left[ \left( \eta_{t-1}( \beta' + \tau_{t-1} Z' ) - \beta_c' \right)^2 \right]\nonumber\\
&+  \frac{1}{\delta N^2} \sum_{\ell_1, \ell_2 \in \mc{K}_0 }\mathbb{E} \left[ \left(\eta_{t-1}( [\beta' + \tau_{t-1} Z']_{\Lambda_{c+\ell}} ) - \beta_{c+\ell}' \right)^2\right] \nonumber\\
&+  \frac{(N-2k)}{\delta N^2} \sum_{\substack{\ell_1 \in \mc{K}_0 \\ \ell_2 = 0} }\mathbb{E} \left[ \left(\eta_{t-1}( [\beta' + \tau_{t-1} Z']_{\Lambda_{c+\ell}} ) - \beta_{c+\ell}' \right)^2\right] \nonumber\\
&+  \frac{(N-2k)}{\delta N^2}\sum_{\substack{\ell_2 \in \mc{K}_0 \\ \ell_1 = 0} }\mathbb{E} \left[ \left(\eta_{t-1}( [\beta' + \tau_{t-1} Z']_{\Lambda_{c+\ell}} ) - \beta_{c+l}' \right)^2\right],
\label{eq:SE2D}
\end{align}
where we notice that there are $(2k)^2$ terms in the second summand, $2k$ terms in the third and fourth summands, and $\frac{(N-2k)^2}{N^2}+\frac{(2k)^2}{N^2}+\frac{2k(N-2k)}{N^2}+\frac{2k(N-2k)}{N^2}=1$.  Again, in the above the first term sums over all the middle indices.  In this case, the second term corresponds to the corner edge cases, while the third and fourth terms correspond to the edge cases in one dimension only. We note that $\sigma_t^2$ is a function of $N$, but do not explicitly represent this relationship to simplify the notation. Moreover, for fixed $k$, the terms $\frac{(2k)^2}{N^2}$, $\frac{2k(N-2k)}{N^2}$, and $\frac{2k(N-2k)}{N^2}$ vanish as $N$ goes to infinity. Therefore, we have $\lim_{N\to\infty}\sigma_t^2(N)=\frac{1}{\delta}\mathbb{E} [ ( \eta_{t-1}( \beta' + \tau_{t-1} Z' ) - \beta_c' )^2 ]$.

Similar to \cite{RushV18}, our performance guarantee, Theorem \ref{thm:main_amp_perf}, is a concentration inequality for PL(2) loss functions at any fixed iteration $t<T^*$, where $T^*$ is the first iteration when either $(\sigma_t^\perp)^2$ or $(\tau_t^\perp)^2$ defined in \eqref{eq:sigperp_defs} is smaller than a predefined quantity $\hat{\epsilon}$. The precise definition of $(\sigma_t^\perp)^2$ and $(\tau_t^\perp)^2$ is deferred to Section \ref{subsec:concvals}. For now, we can understand $(\sigma_t^\perp)^2$ (respectively, $(\tau_t^\perp)^2$) as a number that quantifies (in a probabilistic sense) how close an estimate $\beta^t$ (respectively, a residual $z^t$) is to the subspace spanned by the previous estimates $\{\beta^s\}_{s<t}$ (respectively, the previous residuals $\{z^s\}_{s<t}$). In the special case where $\{\eta_t\}_{t\geq 0}$ are Bayes-optimal conditional expectation denoisers, it can be shown that small $(\sigma_t^\perp)^2$ implies that the difference between $\sigma_t^2$ and $\sigma_{t-1}^2$ is small \cite{RushV18}.

\begin{thm}\label{thm:main_amp_perf}
Under the assumptions stated in Section \ref{subsec:def_assumption}, and for fixed half window-size $k>0$, then for any (order-$2$) pseudo-Lipschitz function $\phi: \mathbb{R}^{2} \rightarrow \mathbb{R}$, $\e \in (0,1)$, and  $0\leq t< T^*$,

\begin{align}
&P\Big(\Big \lvert \frac{1}{|\Gamma|}\sum_{i\in\Gamma}\Big(\phi(\beta^{t+1}_{i},\beta_{i}) - \mathbb{E}[\phi(\eta_{t}([\beta +\tau_t Z]_{\Lambda_i}),\beta_i)]\Big) \Big \lvert \geq \e\Big)\nonumber\\
&\hspace*{2in} \leq K_{k,t} e^{-\kappa_{k,t} n \e^2},
\label{eq:PL_conc_result}
\end{align}
where $\beta \in E^{\Gamma} \sim \mu$, $Z \in\mathbb{R}^{\Gamma}$  has i.i.d.\ standard normal entries and is independent of $\beta$, and the deterministic quantity $\tau_t$ is defined in \eqref{eq:taut_sigmat_def}. The constants $K_{k,t}, \kappa_{k,t} > 0$ do not depend on $n$ or $\e$, but do depend on $k$ and $t$.  Their values are not explicitly specified.
\end{thm}
\begin{proof}
See Section \ref{sec:amp_proof}.
\end{proof}

\textbf{\emph{Remarks}}:

(1) The probability in \eqref{eq:PL_conc_result} is w.r.t.\ the product measure on the space of the matrix $A$, signal $\beta$, and noise $w$.

(2) By choosing the following PL(2) loss function, $\phi(a,b) = (a - b)^2$, Theorem \ref{thm:main_amp_perf} gives the following concentration result for the mean squared error of the estimates.  For all $t \geq 0$,
\ben
P\Big(\Big \lvert \frac{1 }{|\Gamma|}\|\beta^{t+1}- \beta\|^2  - \delta \sigma_{t+1}^2\Big\lvert \geq \e \Big)\leq K_{k,t} e^{-\kappa_{k,t} n \e^2},
\een
with $\sigma_{t+1}^2$ defined in \eqref{eq:taut_sigmat_def}. 


\subsection{Numerical Examples}
\label{subsec:numerical}

Before moving to the proof of Theorem \ref{thm:main_amp_perf}, we first demonstrate the effectiveness of the AMP algorithm with sliding-window denoisers when used to reconstruct an image $\beta_0$ from its linear measurements acquired according to \eqref{eq:model1}. We verify that state evolution accurately tracks the normalized estimation error of AMP, as is guaranteed by Theorem \ref{thm:main_amp_perf}. While we use squared error as the error metric in our examples, which corresponds to the case where the PL(2) loss function $\phi$ in Theorem \ref{thm:main_amp_perf} is defined as $\phi(a,b):=(a-b)^2$, we remind the reader that Theorem \ref{thm:main_amp_perf} also supports other PL(2) loss functions. Moreover, we apply AMP with sliding-window denoisers to reconstruct texture images, which are known to be well-modeled by MRFs in many cases \cite{dubes1989,cross1983}.

\subsubsection{Verification of state evolution} \label{subsubseq:verification}
We consider a class of stationary MRFs on $\mathbb{Z}^2$ whose neighborhood is defined as the eight-nearest neighbors, meaning this is a $2^{nd}$-order MRF per the definition in Section \ref{subsec:def_assumption}. The joint distribution of such an MRF on any finite $M\times N$ rectangular lattice in $\mathbb{Z}^2$ has the following expression \cite{Champagnat1998}:
\begin{equation}
\begin{split}
&\mu(x)=P(\beta=x)\\
&= \frac{\prod_{m=1}^{M-1}\prod_{n=1}^{N-1} \left[\begin{matrix}
x_{m,n} & x_{m,n+1}\\
x_{m+1,n} & x_{m+1,n+1}
\end{matrix}\right]\prod_{m=2}^{M-1}\prod_{n=2}^{N-1}\left[\begin{matrix} x_{m,n}\end{matrix}\right]}{\prod_{m=2}^{M-1}\prod_{n=1}^{N-1}
\left[\begin{matrix}
x_{m,n} & x_{m,n+1}
\end{matrix}\right]
\prod_{m=1}^{M-1}\prod_{n=2}^{N-1}
\left[\begin{matrix}
x_{m,n}\\ 
x_{m+1,n}
\end{matrix}\right]},
\end{split}
\label{eq:mu}
\end{equation}
where we follow the notation in \cite{Champagnat1998} for the generic measure $\left[\begin{matrix}
x_{m,n} & x_{m,n+1}\\
x_{m+1,n} & x_{m+1,n+1}
\end{matrix}\right]$ defined as
\begin{equation*}
\begin{split}
&\left[\begin{matrix}
x_{m,n} & x_{m,n+1}\\
x_{m+1,n} & x_{m+1,n+1}
\end{matrix}\right]\\
&:=P\Big(\beta_{m,n}=x_{m,n},\beta_{m,n+1}=x_{m,n+1},\\
&\qquad\qquad\qquad\beta_{m+1,n}=x_{m+1,n},\beta_{m+1,n+1}=x_{m+1,n+1}\Big),
\end{split}
\end{equation*}
and the conditional distribution of the element in the box given the element(s) not in the box:
\begin{equation*}
\begin{split}
&\left[\begin{matrix}
x_{m,n} & x_{m,n+1}\\
x_{m+1,n} & \fbox{$x_{m+1,n+1}$}
\end{matrix}\right]\\
&:=P\Big(\beta_{m+1,n+1}=x_{m+1,n+1}\,\Big\vert\\
&\qquad\beta_{m,n}=x_{m,n},\beta_{m+1,n}=x_{m+1,n},\beta_{m,n+1}=x_{m,n+1}\Big). 
\end{split}
\end{equation*} 
The generic measure needs to satisfy some consistency conditions to ensure the Markovian property and stationarity of the MRF on a finite grid; details can be found in \cite{Champagnat1998}.
For convenience, in simulations we use a $\Pi_+$ Binary MRF as defined in \cite[Definition 7]{Champagnat1998}, for which the generic measure is conveniently parameterized by four parameters, namely,
\begin{equation}
[1\quad\framebox[1.5\width]{0}]=p,\quad [0\quad\framebox[1.5\width]{1}]=q,\quad \left[\begin{matrix}
0 & 0\\
\framebox[1.5\width]{1} & 0
\end{matrix}\right]=r,\quad \left[\begin{matrix}
1 & 1\\
1 & \framebox[1.5\width]{0}
\end{matrix}\right]=s.
\label{eq:MRF_para}
\end{equation} 
In the simulations, we set $\{p=0.4,q=0.5,r=0.01,s=0.4\}$.
Using \eqref{eq:Dob_matrix_def} and \eqref{eq:tv_countable}, it can be checked that the distribution measure of this MRF satisfies the Dobrushin uniqueness condition.

\begin{figure*}[t!]
    \centering
    \begin{subfigure}[t]{0.25\textwidth}
        \centering
        \includegraphics[height=1\textwidth]{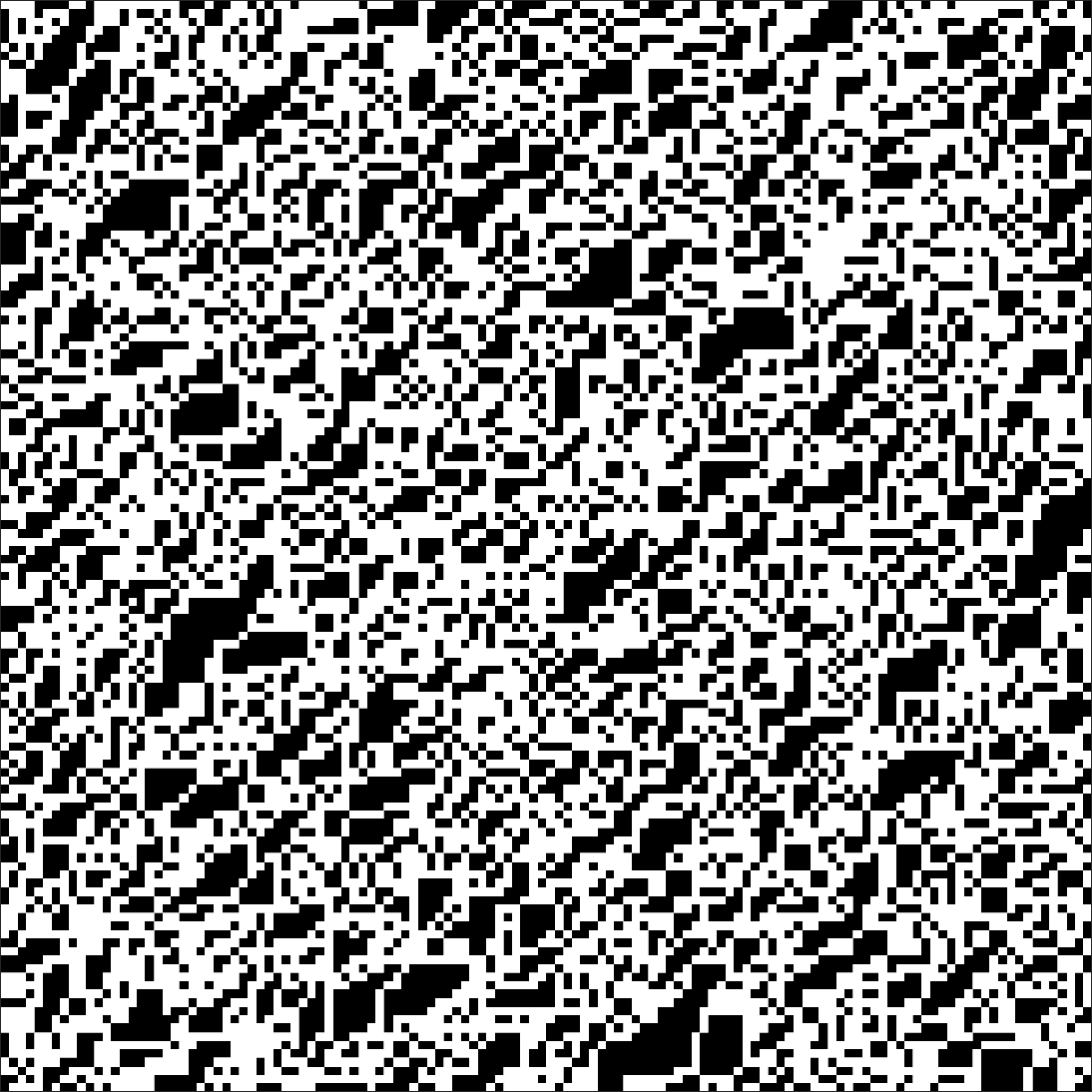}
    \end{subfigure}%
    ~ 
    \begin{subfigure}[t]{0.25\textwidth}
        \centering
        \includegraphics[height=1\textwidth]{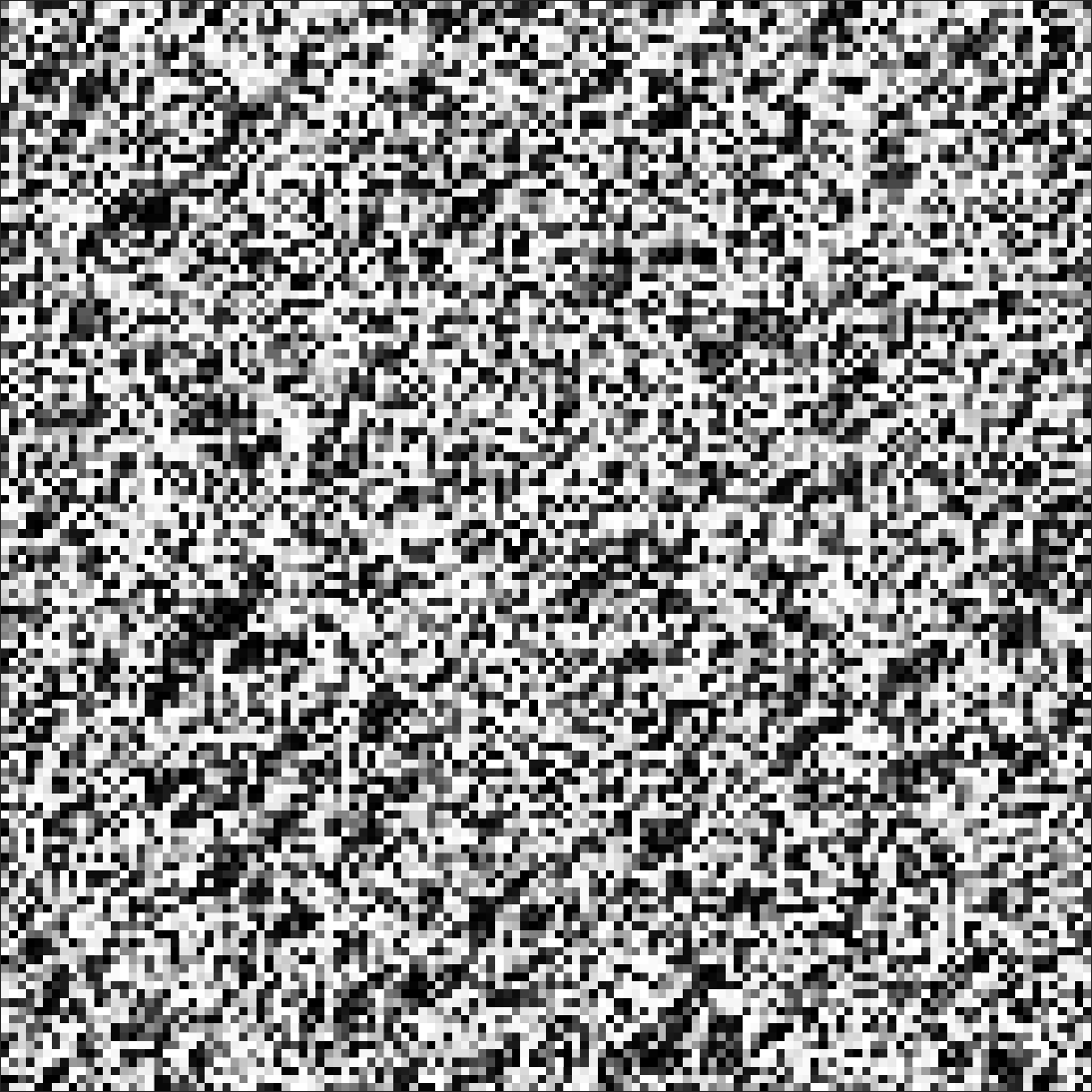}
    \end{subfigure}
    ~
    \begin{subfigure}[t]{0.25\textwidth}
        \centering
        \includegraphics[height=1\textwidth]{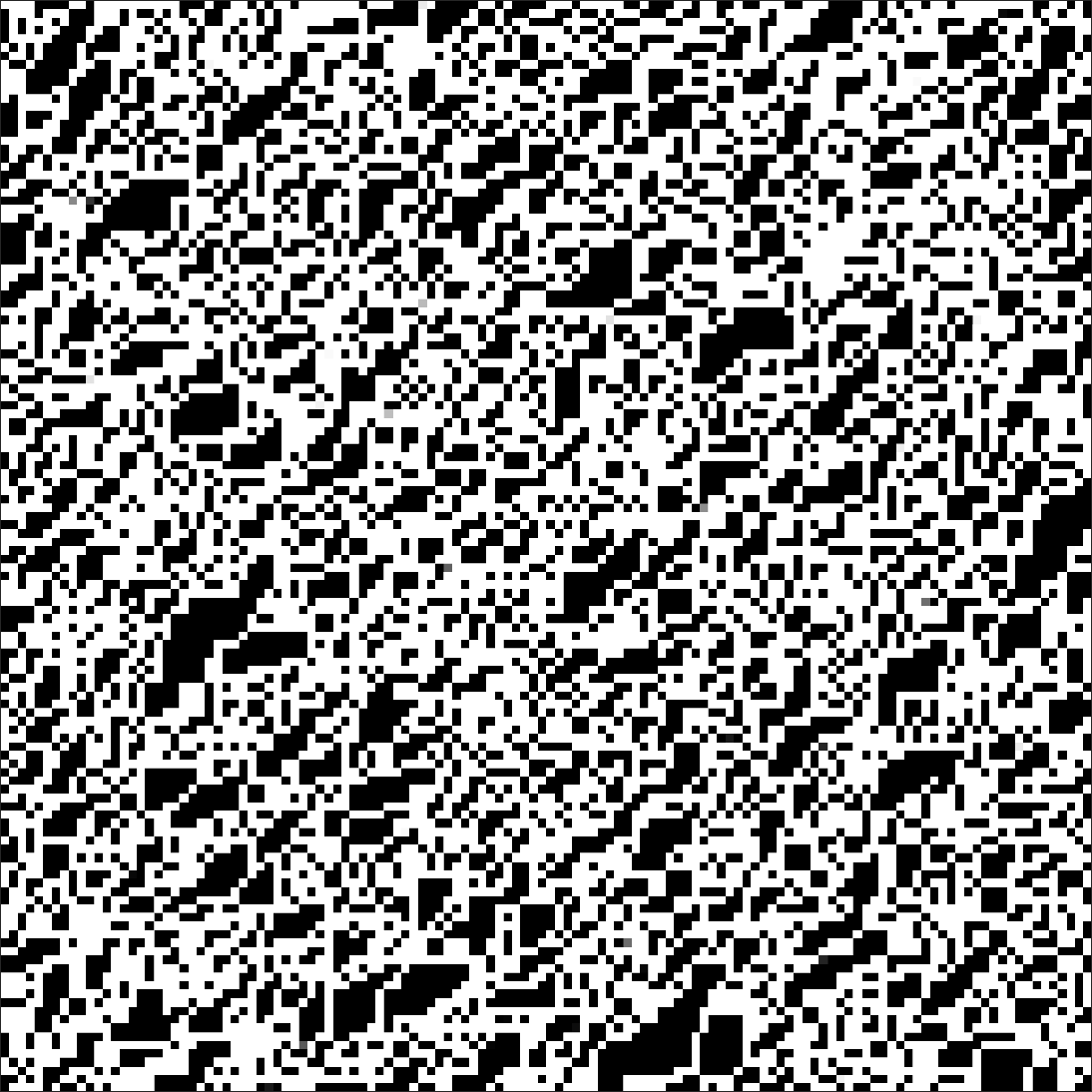}
    \end{subfigure}   
    \caption{From left to right: ground-truth image generated by the MRF described in Section \ref{subsubseq:verification}, image reconstructed by AMP with a separable Bayesian denoiser (computed from the incorrect assumption that the signal is generated from an i.i.d.\ Bernoulli distribution), and image reconstructed by AMP with a Bayesian sliding-window denoiser with $k=1$, hence $\Lambda=[3]\times [3]$. ($\Gamma = [128]\times [128]$, $\delta=0.5$, $\text{SNR}=17$ dB.)}
\label{fig:visual_comparison}
\end{figure*}

\begin{figure}[t!]
     \centering
     \includegraphics[width=0.45\textwidth]{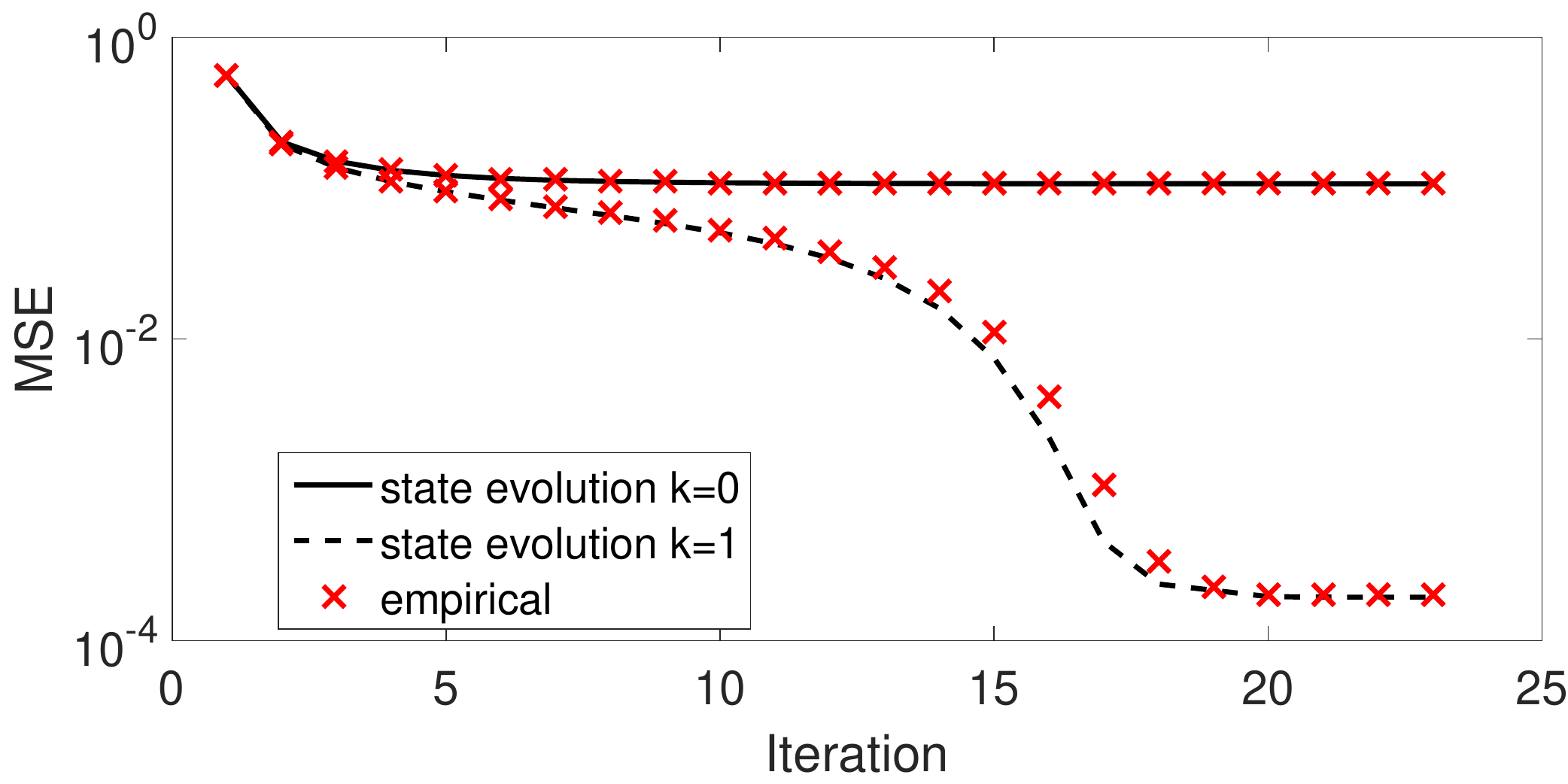}  
     \caption{Numerical verification that the empirical MSE achieved by AMP with sliding-window denoisers is tracked by state evolution. The empirical MSE is averaged over 50 realizations of the MRF (as described in Section \ref{subsubseq:verification}), measurement matrix, and measurement noise. ($\Gamma = [128]\times [128]$, $\delta=0.5$, $\text{SNR}=17$ dB.)}   
     \label{fig:se_track_empirical} 
\end{figure} 

As mentioned previously, an attractive property of AMP, which is formally stated in Theorem \ref{thm:main_amp_perf}, is the following: for large $n$ and $|\Gamma|$ and for $i\in\Gamma$, the observation vector $[A^*z^t + \beta^t]_{\Lambda_i}$ used as an input to the estimation function in \eqref{eq:amp2} is approximately distributed as $\beta' + \tau_t Z'$, where $\beta' \sim \mu_\Lambda$, $Z'$ has i.i.d.\ standard normal entries, independent of $\beta'$, and $\tau_t$ is defined in \eqref{eq:taut_sigmat_def}. With this property in mind, a natural choice of denoiser functions $\{\eta_t\}_{t\geq 0}$ are those that calculate the conditional expectation of the signal given the value of the input argument, which we refer to as Bayesian sliding-window denoisers.  Let $V_t=\beta' + \tau_t Z' $ and $v \in \mathbb{R}^{\Lambda}$, then for each $t \geq 0$ we define
\begin{align}
\eta_t(v) &:= \mathbb{E}\left[\, \beta_c' \, \left\vert \, V_t =  v \right. \,\right]\nonumber\\
&=P\left(\beta'_c=1|V_t=v\right)\nonumber=\frac{P\left( V_t=v, \beta'_c = 1\right)}{P(V_t=v)}\nonumber\\
& = \frac{\sum_{x_{\Lambda\setminus c} \in \{0,1\}^{\Lambda\setminus c}, x_c = 1} f_{V_t|\beta'}(v|x) \mu(x)}{\sum_{x\in\{0,1\}^\Lambda} f_{V_t|\beta'}(v|x) \mu(x)},\label{eq:bayes_denoiser}
\end{align}
where $x_c$ denotes the center coordinate of $x\in\mathbb{R}^\Lambda$, $x_{\Lambda\setminus c}$ denotes all coordinates in $x$ except the center, $f_{V_t|\beta'}(v|x)=\prod_{i\in\Lambda} \frac{1}{\sqrt{2\pi}\tau_t}\exp\left(- \frac{(v_i - \beta'_i)^2}{2\tau_t^2}\right)$ since coordinates of $Z'$ are i.i.d.\ normal, and $\mu(x)$ is computed according to \eqref{eq:mu} with $M=N=2k+1$ by using \eqref{eq:MRF_para} and the property of $\Pi_+$ Binary MRF given in \cite[Definition 7]{Champagnat1998}.  Figure \ref{fig:se_track_empirical} shows that the MSE achieved by AMP with the non-separable sliding-window denoiser defined above is tracked by state evolution at every iteration.
 
Notice that when $k=0$, the denoisers $\{\eta_t\}_{t \geq 0}$ are separable and since the empirical distribution of $\beta_0$ converges to the stationary probability distribution $\mu_1$ on 
$E\subset\mathbb{R}$, the state evolution analysis for AMP with separable denoisers ($k=0$) was justified by Bayati and Montanari \cite{BayMont11}. However, it can be seen in Figures \ref{fig:visual_comparison} and \ref{fig:se_track_empirical} that the MSE achieved by the separable denoiser ($k=0$) is significantly higher (worse) than that achieved by the non-separable denoisers ($k=1$).

\begin{figure*}[t!]
\centering
\includegraphics[width=0.8\textwidth]{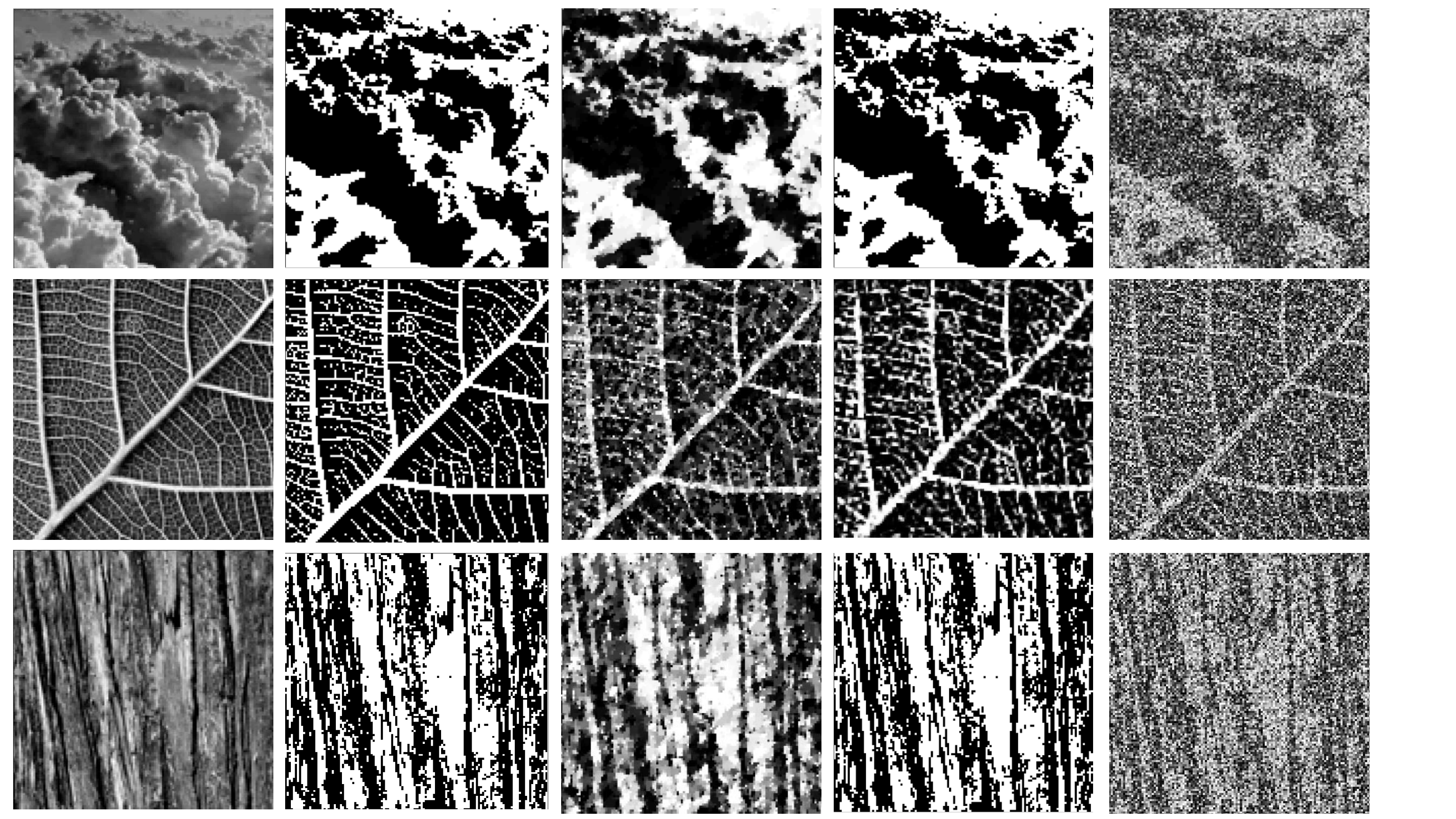}
\caption{Reconstruction of texture images using AMP with different denoisers. From left to right: original gray level images, binary ground-truth images, images reconstructed by AMP with a total variation denoiser \cite{Beck2009FISTA}, non-separable Bayesian sliding-window denoiser (MRF prior, $k=1$), and separable Bayesian denoiser (Bernoulli prior), respectively. From top to bottom: images of a cloud, a leaf, and wood, respectively. ($\Gamma = [128]\times [128]$, $\delta=0.3$, $\text{SNR}=20$ dB.)}
\label{fig:texture}
\end{figure*}

\subsubsection{Texture Image Reconstruction} \label{subsubseq:texture}
We now use the Bayesian sliding-window denoiser defined in \eqref{eq:bayes_denoiser} to reconstruct binary texture images shown in Figure \ref{fig:texture}. The MRF prior is the same type as described in Section \ref{subsubseq:verification}, namely the $\Pi_+$ Binary MRF, but we set the parameters $\{p=0.18, q=0.16, r=0.034, s=0.01\}$. Note that while it is possible to learn an MRF model for each of the images using well-established MRF learning algorithms, we do not include this procedure in our simulations since the study of texture image modeling is beyond the scope of this paper.  Moreover, the reconstruction results obtained using the simple MRF defined above are sufficiently satisfactory, despite the fact that the prior may be inaccurate. In Figure \ref{fig:texture}, we begin with natural images of a cloud, a leaf, and wood ($1^{st}$ column) and then use thresholding to generate binary test images ($2^{nd}$ column). In addition to presenting the reconstructed images obtained by the Bayesian sliding-window denoisers with $k=1$ ($4^{th}$ column) and $k=0$ ($5^{th}$ column), respectively, we also present those obtained by AMP with a total variation denoiser \cite{Beck2009FISTA} as a baseline approach ($3^{th}$ column).


\section{Proof of Theorem \ref{thm:main_amp_perf}} \label{sec:amp_proof}

The proof of Theorem \ref{thm:main_amp_perf} follows 
the work of Rush and Venkataramanan~\cite{RushV18},
with modifications for the dependent structure of the unknown vector $\beta$ in \eqref{eq:model1}.  
For this reason, we use much of the same notation.
We prove Theorem \ref{thm:main_amp_perf} using a technical lemma, Lemma \ref{lem:main_lem}, which corresponds to \cite[Lemma 6]{RushV18}. 
Before stating the lemma, we cover some preliminary results and establish notation to be used in its proof. 


\subsection{Proof Notation} \label{subsec:defs}
As in the previous work by Bayati and Montanari \cite{BayMont11}, as well as the work by Rush and Venkataramanan \cite{RushV18}, the technical lemma is proved for a more general recursion, with AMP being a specific example of the general recursion as shown below. The connection between AMP and the general recursion will be explained in \eqref{eq:hqbm_def_AMP} and \eqref{eq:fg_def}.

Fix the half-window size $0\leq k\leq(N-1)/2$, an integer. Let $\{f_t\}_{t\geq 0}:\mathbb{R}^{\Lambda\times\Lambda}\to\mathbb{R}$ and $\{g_t\}_{t\geq 0}:\mathbb{R}^2\to\mathbb{R}$ be sequences of Lipschitz functions. Specifically, the arguments of $f_t$ are two variables in  $\mathbb{R}^\Lambda$, for example, for $x,y\in\mathbb{R}^\Lambda$, we write $f_t(x,y)$ and call $x$ the first argument of $f_t$. Given noise $w \in \mathbb{R}^n$ and unknown signal $\beta \in E^{\Gamma}$, define vectors $h^{t+1}, q^{t+1} \in \mathbb{R}^{|\Gamma|}$ and $b^t, m^t \in \mathbb{R}^n$, as well as arrays $\hat{h}^{t+1}, \hat{q}^{t+1}\in\mathbb{R}^{\Gamma}$ (for which $h^{t+1}$ and $q^{t+1}$ are the vectorized versions) for $t \geq 0$ recursively as follows.  Starting with initial condition $\hat{q}^0 \in \mathbb{R}^{\Gamma}$:
\begin{equation}
\begin{split}
&h^{t+1}:=A^*m^t-\xi_t q^t, \quad q^t:=\mc{V}(\hat{q}^t),\\
&\hat{h}^{t+1}:=\mc{V}^{-1}(h^{t+1}), \quad \hat{q}_i^t:=f_t([\hat{h}^t]_{\Lambda_i},\, \beta_{\Lambda_i}), \text{ for all } i\in\Gamma, \\
&b^t:=Aq^t-\lambda_tm^{t-1},\quad m_i^t:=g_t(b_i^t,w_i), \text{ for all } i\in [n], 
\label{eq:hqbm_def}
\end{split}
\end{equation}
with the scalars $\xi_t,\lambda_t$ defined as
\begin{equation}
\begin{split}
\xi_t &:= \frac{1}{n}\sum_{i=1}^n g_t'(b^t_i,w_i),\quad \lambda_t :=\frac{1}{n}\sum_{i\in\Gamma}f_t'([\hat{h}^t]_{\Lambda_i},\beta_{\Lambda_i}),
\end{split}
\label{eq:xi_lamb_def}
\end{equation}
where the derivative of $g_t$ is w.r.t.\ the first argument, and the derivative of $f_t$ is w.r.t.\ the center coordinate of the first argument. In the context of AMP, as made explicit in \eqref{eq:hqbm_def_AMP}, the terms $\hat{h}^{t+1}$ and $\hat{q}^t$ measure the error in the observation $\mc{V}^{-1}(A^*z^t) + \beta^t$ and the estimate $\beta^t$ at time $t$, respectively, (the error w.r.t.\ the true $\beta$).  The term $m^t$ measures the residual at time $t$ and the term $b^t$ is the difference between the noise and residual at time $t$.

Recall that the unknown vector $\beta\in E^{\Gamma}$ is assumed to have a stationary MRF prior with joint distribution measure $\mu$. 
Let $\beta \in E^{\Gamma} \sim\mu$ and $\mathbf{0}\in\mathbb{R}^\Lambda$ be an all-zero array. Define
\begin{align}
\sigma_{\beta}^2 &:= \mathbb{E}[\beta_1^2], \label{eq:sigmabetadef} \\
\sigma_0^2 &:= \frac{1}{\delta |\Gamma|}\sum_{i\in\Gamma}\mathbb{E}[f_0^2(\mathbf{0},\beta_{\Lambda_i})] > 0. \label{eq:sigma0}
\end{align}
Further, for all $i\in\Gamma$ let
\be
\hat{q}_i^{0}  := f_0 (\mathbf{0}, \beta_{\Lambda_i}), \qquad \text{ and } \qquad q^{0} := \mc{V}(\hat{q}^{0}),
\label{eq:q0def}
\ee
and assume that there exist constants $K, \kappa > 0$ such that
\be
P\Big(\Big\lvert \frac{1}{n} \norm{q^0}^2 - \sigma_0^2 \Big \lvert \geq \e \Big) \leq K e^{-\kappa n \e^2}.
\label{eq:qassumption}
\ee
Define the state evolution scalars $\{\tau_t^2\}_{t \geq 0}$ and $\{\sigma_t^2\}_{t \geq 1}$ for the general recursion as follows, 
\be
\begin{split}
\tau_t^2 &:= \mathbb{E}[(g_{t}(\sigma_{t} Z, W))^2],\\
\sigma_t^2 &:= \frac{1}{\delta |\Gamma| } \sum_{i\in\Gamma}\mathbb{E}[(f_{t}(\tau_{t-1} \mbf{Z}_{\Lambda_i}, \beta_{\Lambda_i}))^2],
\end{split}
\label{eq:sigmatdef}
\ee
where random variables $W \sim p_w$ and $Z \sim \mc{N}(0,1)$ are independent and random arrays 
$\beta\in E^\Gamma \sim \mu$ and $\mbf{Z}\in\mathbb{R}^\Gamma$ with i.i.d.\ $\mc{N}(0,1)$ entries are also independent. We assume that both $\sigma_0^2$ and $\tau_0^2$ are strictly positive. The technical lemma will show that $\hat{h}^{t+1}$ can be approximated as i.i.d.\ $\mc{N}(0, \tau_t^2)$ in functions of interest for the problem, namely when used as an input to PL functions, and $b^t$ can be approximated as i.i.d.\ $\mc{N}(0, \sigma_t^2)$ in PL functions.  Moreover, it will be shown that the probability of the deviations of the quantities  $\frac{1}{n}\|m^t\|^2$ and $\frac{1}{n}\|\hat{q}^t\|^2$  from $\tau_t^2$ and $\sigma_t^2$, respectively, decay exponentially in $n$.

We note that the AMP algorithm introduced in \eqref{eq:amp1} and \eqref{eq:amp2} is a special case of the general recursion of \eqref{eq:hqbm_def} and \eqref{eq:xi_lamb_def}.  Indeed, define the following vectors  recursively for $t\geq 0$, starting with $\beta^0=0$ and $z^0=y$, 
\begin{equation}
\begin{split}
\hat{h}^{t+1} = \beta - (\mc{V}^{-1}(A^*z^t) + \beta^t), \qquad &  \hat{q}^t  = \beta^t - \beta, \\
b^t = w-z^t,\qquad & m^t  =-z^t.
\end{split}
\label{eq:hqbm_def_AMP}
\end{equation}
It can be verified that these vectors satisfy  \eqref{eq:hqbm_def} and \eqref{eq:xi_lamb_def} using Lipschitz functions
\be
f_t(a, \beta_{\Lambda_i}) = \eta_{t-1}(\beta_{\Lambda_i} - a) - \beta_i  \text{ and } g_t(b, w_i) = b-w_i,
\label{eq:fg_def}
\ee
where $ a \in\mathbb{R}^{\Lambda}$ and $b\in\mathbb{R}$.  Using the choice of $f_t, g_t$ given in \eqref{eq:fg_def} also yields the expressions for $\sigma_t^2, \tau_t^2$  given in \eqref{eq:taut_sigmat_def}.   In the remaining analysis, the general recursion given in \eqref{eq:hqbm_def} and \eqref{eq:xi_lamb_def} is used.  Note that in AMP, $q^0=-\beta$ and $\sigma_0^2 = \sigma_\beta^2/\delta$, hence, assumption \eqref{eq:qassumption} for AMP requires
\begin{equation}
P\Big(\Big \lvert \frac{1}{|\Gamma|} \norm{\beta}^2 - \sigma_\beta^2 \Big \lvert \geq \delta \e \Big) \leq K e^{-\kappa n \delta^2\e^2}.
\label{eq:beta0_assumption2}
\end{equation}
Under our assumptions for $\beta$ as stated in Section \ref{subsec:def_assumption}, we see that \eqref{eq:beta0_assumption2} is satisfied using Lemma \ref{lem:PL_MRF_conc} (Appendix \ref{app:conc_dependent}), since the function $f(x)=x^2$ is pseudo-Lipschitz.
Finally, note that if we assume $\sigma_\beta^2>0$ and $\delta<\infty$, then the condition of strict positivity of $\sigma_0^2$ and $\tau_0^2$ defined in \eqref{eq:sigmatdef} is satisfied.


Let $[c_1 \mid c_2 \mid \ldots \mid c_k]$ denote a matrix with columns $c_1, \ldots, c_k$. For $t\geq 1$, define matrices
\begin{equation}
\begin{split}
&M_t  := [m^0 \mid \ldots \mid m^{t-1} ],\quad Q_t  :=  [q^0 \mid \ldots \mid q^{t-1} ], \\
&B_t := [b^0 | \ldots | b^{t-1}],\quad H_t := [h^1 | \ldots | h^{t}]. 
\end{split}
\label{eq:XYMQt}
\end{equation}
Moreover, $M_0$, $Q_0$, $B_0$, $H_0$ are defined to be the all-zero vector. 
 

The values $m^t_{\|}$ and $q^t_{\|}$ are projections of $m^t$ and $q^t$ onto the column space of $M_t$ and $Q_t$, with $m^t_{\perp} := m^t - m^t_{\|},$ and $q^t_{\perp} := q^t - q^t_{\|}$ being the projections onto the orthogonal complements of $M_t$ and $Q_t$.  Finally, define the vectors
 \be 
 \alpha^t := (\alpha^t_0, \ldots, \alpha^t_{t-1})^*, \qquad  \gamma^t :=  (\gamma^t_0, \ldots, \gamma^t_{t-1})^*, 
 \label{eq:vec_alph_gam_conc}
 \ee 
 to be the coefficient vectors of the parallel projections, i.e.,
 \be
 m^t_{\| } := \sum_{i=0}^{t-1} \alpha^t_i m^i, \qquad  q^t_{\|} := \sum_{i=0}^{t-1} \gamma^t_i q^i.
 \label{eq:mtqt_par}
 \ee
The technical lemma, Lemma \ref{lem:main_lem},  shows that for large $n$, the entries of the vectors $\alpha^t$ and ${\gamma}^t$ concentrate to constant values, which are defined in the following section.  


\subsection{Concentrating Constants} \label{subsec:concvals}

Recall that $\beta \in E^{\Gamma}$ is the unknown vector to be recovered and $w \in \mathbb{R}^n$  is the measurement noise.  
In this section we introduce the concentrating values for inner products of pairs of the vectors $\{h^t, m^t, q^t, b^t\}$ that are used in Lemma \ref{lem:main_lem}.

Let $\{\breve{Z}_t \}_{t \geq 0}$ be a sequence of zero-mean jointly Gaussian random variables taking values in $\mathbb{R}$,
and let  $\{ \tilde{\mathbf{Z}}_t \}_{t \geq 0}$ be a sequence of zero-mean jointly Gaussian random arrays taking values in $\mathbb{R}^{\Gamma}$. The covariance of the two random sequences is defined recursively as follows. For $r,t \geq 0$, $i,j\in\Gamma$,
\be 
\expec[\breve{Z}_r \breve{Z}_t] = \frac{\tilde{E}_{r,t}}{\sigma_r \sigma_t}  , \quad \expec \Big[ [\tilde{\mathbf{Z}}_{r}]_i [\tilde{\mathbf{Z}}_{t}]_j \Big] = \begin{cases}\frac{\breve{E}_{r,t}}{\tau_r \tau_t}, &\text{if } i = j\\
0,&\text{if } i\neq j
\end{cases}, 
\label{eq:tildeZcov}
\ee
where
\be
\begin{split}
&\breve{E}_{r,t} := \mathbb{E}\Big[g_{r}(\sigma_{r} \breve{Z}_{r}, W) g_{t}(\sigma_{t} \breve{Z}_{t}, W) \Big],\\
&\tilde{E}_{r,t} := \sum_{i\in\Gamma}\frac{\mathbb{E} [f_{r}(\tau_{r-1} [\tilde{\mathbf{Z}}_{r-1}]_{\Lambda_i}, \beta_{\Lambda_i}) f_{t}(\tau_{t-1} [\tilde{\mathbf{Z}}_{t-1}]_{\Lambda_i}, \beta_{\Lambda_i}) ]}{\delta |\Gamma|} . 
\end{split}
\label{eq:Edef}
\ee
Note that both terms of the above \eqref{eq:Edef} are scalar values and we take $f_0(\cdot, \beta_{\Lambda_i}) := f_0(\mathbf{0},\beta_{\Lambda_i})$, the initial condition.  Moreover, $\tilde{E}_{t,t} = \sigma_t^2$ and $\breve{E}_{t,t} = \tau_t^2$, as can be seen from \eqref{eq:sigmatdef}, thus for all $i\in\Gamma$, we have $\expec [[\tilde{\mathbf{Z}}_{t}]_i^2] = \expec[\breve{Z}^2_t] =1$. Therefore, $\tilde{\mathbf{Z}}_{t}$ has i.i.d.\ $\mc{N}(0,1)$ entries.

Next, we define matrices $\tilde{C}^t, \breve{C}^t \in \mathbb{R}^{t \times t}$ and vectors $\tilde{E}_t,\breve{E}_t\in\mathbb{R}^t$ whose entries are $\{\tilde{E}_{r,t}\}_{r,t\geq 0}$ and $\{\breve{E}_{r,t}\}_{r,t\geq 0}$ defined in \eqref{eq:Edef}: for $0\leq i,j \leq t-1$,
\begin{align}
\tilde{C}^{t}_{i+1,j+1} := \tilde{E}_{i,j}, &\quad \breve{C}^{t}_{i+1,j+1} := \breve{E}_{i,j},\label{eq:Ct_def} \\
\tilde{E}_t := (\tilde{E}_{0,t} \ldots, \tilde{E}_{t-1,t})^*, &\quad \breve{E}_t:=(\breve{E}_{0,t} \ldots, \breve{E}_{t-1,t})^*. \label{eq:Et_def}
\end{align}
Lemma \ref{lem:Ct_invert} below shows that $\tilde{C}^{t}$ and $\breve{C}^{t}$ are invertible. Therefore, we can define the 
concentrating values for $\gamma^t$ and $\alpha^t$ defined in \eqref{eq:vec_alph_gam_conc} as
\be
\hat{\gamma}^{t} := (\tilde{C}^t)^{-1}\tilde{E}_t  \quad \text{ and } \quad \hat{\alpha}^{t} := (\breve{C}^t)^{-1}\breve{E}_t,
\label{eq:hatalph_hatgam_def}
\ee
as well as
the values of $(\sigma_{t}^{\perp})^2$ and  $(\tau^{\perp}_{t})^2$ for $t > 0$:
\be
\begin{split}
& (\sigma_{t}^{\perp})^2 := \sigma_t^2 - (\hat{\gamma}^{t})^* \tilde{E}_{t}= \tilde{E}_{t,t} - \tilde{E}^*_{t} (\tilde{C}^{t})^{-1}  \tilde{E}_{t}, \\
& (\tau^{\perp}_{t})^2 := \tau_{t}^2 - (\hat{\alpha}^{t})^* \breve{E}_t =  \breve{E}_{t,t} -   \breve{E}^*_{t} (\breve{C}^{t})^{-1}  \breve{E}_{t}.
\label{eq:sigperp_defs}
\end{split}
\ee
For $t=0$, we let $(\sigma^{\perp}_0)^2 := \sigma_0^2$ and $(\tau^{\perp}_0)^2 := \tau_0^2$.
Finally, define the concentrating values for $\lambda_{t+1}$ and $\xi_{t}$ defined in \eqref{eq:xi_lamb_def} as
\begin{equation}
\hat{\xi}_{t} := \mathbb{E} [g'_{t}(\sigma_{t} \breve{Z}_{t}, W)],\quad \hat{\lambda}_{t+1} := \frac{1}{\delta |\Gamma|}\sum_{i\in\Gamma} \mathbb{E} [f'_{t}(\tau_{t} [\tilde{\mathbf{Z}}_{t}]_{\Lambda_i}, \beta_{\Lambda_i})].\label{eq:hatlambda_hatxi}
\end{equation}

\begin{lem} \label{lem:Ct_invert}
If $(\sigma_k^{\perp})^2$ and $(\tau_k^{\perp})^2$ are bounded below by some positive constants for $k \leq t$, then the matrices $\tilde{C}^{k+1}$ and $\breve{C}^{k+1}$ defined in \eqref{eq:Ct_def} are invertible for $k \leq t$. \end{lem}

\begin{proof}
The proof follows directly as that of \cite[Lemma 1]{RushV18} and therefore is not restated here.  To see that this is the case, note that the proof of \cite[Lemma 1]{RushV18} relies only on the relationship between $(\sigma_k^{\perp})^2$ (resp.\ $(\tau_k^{\perp})^2$) and $\tilde{C}^{k}$ (resp.\ $\breve{C}^{k}$) as defined in \cite[(4.19)]{RushV18}, which is the same as \eqref{eq:sigperp_defs}, and not the actual values taken by these objects. Therefore, the proof for \cite[Lemma 1]{RushV18} applies here.

%
\end{proof}

\subsection{Conditional Distribution Lemma} \label{sec:cond_dist_lemma}

As mentioned previously, the proof of Theorem \ref{thm:main_amp_perf} relies on a technical lemma, Lemma \ref{lem:main_lem}, stated in Section \ref{subsec:main_lem_state} and proved in Section \ref{sec:main_lem_proof}.  Lemma \ref{lem:main_lem} uses the conditional distribution of the vectors $h^{t+1}$ and $b^t$ 
given the matrices in \eqref{eq:XYMQt} as well as $\beta, w$.  
Two forms of the conditional distribution of $h^{t+1}$ will be provided in Lemmas \ref{lem:hb_cond} and \ref{lem:h_cond_pure}, which correspond to \cite[Lemma 4]{RushV18} and \cite[Lemma 5]{RushV18}, respectively. Lemma \ref{lem:h_cond_pure} explicitly shows that the conditional distribution of $h^{t+1}$ can be represented as the sum of a standard Gaussian vector and a deviation term, where the explicit expression of the deviation term is provided in Lemma \ref{lem:hb_cond}. 
Then Lemma \ref{lem:main_lem} shows that the deviation term is small, meaning that its normalized Euclidean norm concentrates on zero, and also provides concentration results for various inner products involving the other terms in recursion \eqref{eq:hqbm_def}, namely $\{ h^{t+1}, q^t, b ^t, m^t \}$.

The following notation is used. Considering two random variables $X, Y$ and a sigma-algebra $\mscrs$, we denote the relationship that the conditional distribution of $X$ given $\mscrs$ equals the distribution of $Y$ by $X |_\mscrs \stackrel{d}{=} Y$.   We represent a $t \times t$ identity matrix as $\mathsf{I}_{t}$, dropping the $t$ subscript when it is clear from the context. For a matrix $A$ with full column rank, $\mathsf{P}^{\parallel}_{A} := A(A^*A)^{-1}A^*$ is the orthogonal projection matrix onto the column space of $A$, and $\mathsf{P}^\perp_{A}:=\mathsf{I}- \mathsf{P}^{\parallel}_{A}$.  Define $\mathscr{S}_{t_1, t_2}$ to be the sigma-algebra generated by the terms
\[ \{b^0, ..., b^{t_1 -1}, m^0, ..., m^{t_1 - 1}, h^1, ..., h^{t_2}, q^0, ..., q^{t_2}, \beta, w\}. \]

\begin{lem}
For the vector $h^{t+1}$ and $b^t$ 
defined in \eqref{eq:hqbm_def}, the following conditional distribution holds for $t \geq 1$:
\begin{align}
&h^{1} \lvert_{\mscrs_{1, 0}} \stackrel{d}{=} \tau_0 Z_0 + \Delta_{1,0}, \nonumber\\
&h^{t+1} \lvert_{\mscrs_{t+1, t}} \stackrel{d}{=} \sum_{r=0}^{t-1} \hat{\alpha}^t_r h^{r+1} + \tau^{\perp}_t Z_t + \Delta_{t+1,t}, \label{eq:Ha_dist} \\
&b^{0} \lvert_{\mscrs_{0, 0}} \stackrel{d}{=} \sigma_0 Z'_0 + \Delta_{0,0}, \quad b^{t} \lvert_{\mscrs_{t, t}}\stackrel{d}{=} \sum_{r=0}^{t-1} \hat{\gamma}^{t}_r b^r + \sigma^{\perp}_t Z'_t + \Delta_{t,t}, \label{eq:Ba_dist}
\end{align}
where $Z_0, Z_t \in \mathbb{R}^{|\Gamma|}$ and $Z'_0, Z'_t \in \mathbb{R}^n$ 
are i.i.d.\ standard Gaussian random vectors that are independent of the corresponding conditioning sigma algebras. The term $\hat{\gamma}^{t}_i$ and 
$\hat{\alpha}^t_{i}$ for $i = 0,...,t-1$ is defined in \eqref{eq:hatalph_hatgam_def} and the term $(\tau_{t}^{\perp})^2$ and $(\sigma_{t}^{\perp})^2$ 
in \eqref{eq:sigperp_defs}.  The deviation terms are 
\begin{align}
\Delta_{0,0} &= \Big(\frac{\norm{q^0}}{\sqrt{n}} - \sigma_0\Big)Z'_0, \label{eq:D00} \\
\Delta_{1,0} &=  \Big[ \Big(\frac{\norm{m^0}}{\sqrt{n}}  - \tau_0\Big)\mathsf{I}_N -\frac{\norm{m^0}}{\sqrt{n}} \mathsf{P}^{\parallel}_{q^0}  \Big] Z_0\nonumber\\
& + q^0  \Big(\frac{\norm{q^0}^2}{n}  \Big)^{-1}  \Big(\frac{(b^0)^*m_0}{n} - \xi_0 \frac{\norm{q^0}^2}{n} \Big), \label{eq:D10}
\end{align}
where $\mathsf{I}$ is the identity matrix and for any matrix $A$, $\mathsf{P}^{\parallel}_A$ is the orthogonal projection matrix onto the column space of $A$.  
For $t >0$, defining $\mathbf{M}_{t} := \frac{1}{n}M_{t}^* M_{t}$ and $\mathbf{Q}_{t+1} := \frac{1}{n}Q_{t+1}^* Q_{t+1}$,
\begin{align}
&\Delta_{t,t} =  \sum_{r=0}^{t-1} (\gamma^t_r - \hat{\gamma}^{t}_r) b^r + \Big[  \Big(\frac{\norm{q^t_{\perp}}}{\sqrt{n}} - \sigma_{t}^{\perp} \Big) \mathsf{I}_n  - \frac{\norm{q^t_{\perp}} }{\sqrt{n}} \mathsf{P}^{\parallel}_{M_t}\Big]Z'_t  \nonumber \\
& + M_t \mathbf{M}_{t}^{-1} \Big(\frac{H_t^* q^t_{\perp}}{n} - \frac{M_t^*}{n}\Big[\lambda_t m^{t-1} - \sum_{r=1}^{t-1} \lambda_{r} \gamma^t_{r} m^{r-1}\Big]\Big),\label{eq:Dtt} \\
&\Delta_{t+1,t} =   \sum_{r=0}^{t-1} (\alpha^t_r - \hat{\alpha}^t_r) h^{r+1}\nonumber\\
& +  \Big[ \Big(\frac{\norm{m^t_{\perp}}}{\sqrt{n}} - \tau_{t}^{\perp}  \Big)  \mathsf{I}_N  -\frac{\norm{m^t_{\perp}}}{\sqrt{n}} \mathsf{P}^{\parallel}_{Q_{t+1}} \Big]Z_t \nonumber \\
& + Q_{t+1}  \mathbf{Q}_{t+1}^{-1}  \Big(\frac{B^*_{t+1} m^t_{\perp}}{n}  - \frac{Q_{t+1}^*}{n} \Big[\xi_t q^t - \sum_{i=0}^{t-1} \xi_i \alpha^t_i q^i \Big] \Big).\label{eq:Dt1t}  
\end{align} 
\label{lem:hb_cond}
\end{lem}

\begin{proof}
 As in \cite{RushV18}, the key theoretical insight in the proof is to study the distribution of $A$ conditioned on the sigma algebra $\mathscr{S}_{t_1, t}$ where $t_1$ is either $t+1$ or $t$, meaning one treats $A$ as random, and considers the output of the AMP algorithm up until the current iteration as fixed and given in the sigma-algebra.  This is done by observing that conditioning on $\mathscr{S}_{t_1, t}$ is equivalent to conditioning on the linear constraints
\[A Q_{t_1} = Y_{t_1} \qquad \text{ and } \qquad A^* M_t = X_t,\]
(due to the relationships $b^s = Aq^s - \lambda_s m^{s-1}$ and $h^{r+1}=A^*m^{r} - \xi_{r} q^{r}$ for $0\leq s \leq t$ and $0\leq r \leq t-1$ given in \eqref{eq:hqbm_def}).  
Then it is straightforward to characterize the conditional distribution of a Gaussian matrix given linear constraints.

Since the sequences $\{\lambda_s; \xi_s; b^s; q^s; h^s; m^s\}_{0\leq s \leq t}$ are all in the conditioning $\mscrs_{t+1, t}$, this conditional distribution depends only on the relationship between the matrix $A$ and these fixed, given terms.  This relationship, namely that specified via $b^s = Aq^s - \lambda_s m^{s-1}$ and $h^{r+1}=A^*m^{r} - \xi_{r} q^{r}$, is the same here as in \cite[Lemma 3]{RushV18}, and so the proofs are identical.  We therefore do not repeat the details. Although $\{\lambda_s\}_{0\leq s \leq t}$ and  $\{q^s\}_{0\leq s \leq t}$ are obtained from $\{f_s\}_{0\leq s \leq t}$, which is separable in \cite{RushV18}, but non-separable in our case,  $\{\lambda_s\}_{0\leq s \leq t}$ and  $\{q^s\}_{0\leq s \leq t}$ are simply treated as fixed elements in the conditioning sigma-algebra $\mscrs_{t+1, t}$ and the fact that they are calculated via non-separable functions here does not change the proof.

Then one is able to specify the conditional distributions of $b^t$ and $h^{t+1}$ given $\mathscr{S}_{t, t}$ and $\mathscr{S}_{t+1, t}$, respectively, using the conditional distribution of $A$.  Again since the relationship between $b^t$ and $h^{t+1}$ and $A$ is the same here as in \cite{RushV18}, the details are identical to that provided in the proof of \cite[Lemma 4]{RushV18}  and are not repeated here.  
\end{proof}

Note that Lemma \ref{lem:hb_cond} holds only when 
$Q^*_{t+1} Q_{t+1}$ is invertible.  The following lemma provides an alternative representation of the conditional distribution of $h^{t+1}|_{\mscrs_{t+1,t}}$ for $t\geq 0$, and it explicitly shows that $h^{t+1}|_{\mscrs_{t+1,t}}$ is distributed as an i.i.d.\ Gaussian random vector with $\mc{N}(0,\tau_t^2)$ entries plus a deviation term.

\begin{lem}
For $t\geq 0$, let $Z_t\in\mathbb{R}^{|\Gamma|}$ be i.i.d.\ standard normal random vectors.
Let $h_{\pure}^1:=\tau_0 Z_0$. For $t\geq 1$, recursively define
\begin{equation}
h^{t+1}_\pure = \sum_{r=0}^{t-1} \hat{\alpha}^t_r h^{r+1}_\pure + \tau_t^\perp Z_t
\label{eq:h_pure_def}
\end{equation}
and a set of scalars $\{\mathsf{d}^t_i\}_{0\leq i\leq t}$ with $\mathsf{d}^0_0=1$,
\begin{equation}
\mathsf{d}^t_i = \sum_{r=i}^{t-1} \mathsf{d}^r_{i} \hat{\alpha}^t_r \quad \text{for }\quad 0\leq i \leq (t-1), \quad\text{ and }\quad \mathsf{d}^t_t = 1.
\label{eq:d_def}
\end{equation}

Let $\hat{h}^{t+1}_{\pure}=\mc{V}^{-1}(h^{t+1}_{\pure})\in\mathbb{R}^\Gamma$.  Then for all $t\geq 0$ we have
\begin{equation}
(\hat{h}^1_\pure,\ldots,\hat{h}^{t+1}_\pure) \overset{d}{=} (\tau_0 \tilde{\mbf{Z}}_0,\ldots, \tau_t \tilde{\mbf{Z}}_t),
\label{eq:h_pure_1}
\end{equation}
where $\{\tilde{\mbf{Z}}_t\}_{t\geq 0}$ are jointly Gaussian with correlation structure defined in \eqref{eq:tildeZcov}.
Moreover,
\begin{equation}
h^{t+1}|_{\mscrs_{t+1,t}} \overset{d}{=} h^{t+1}_\pure + \sum_{r=0}^t \mathsf{d}^t_r \Delta_{r+1,r}.
\label{eq:h_pure_2}
\end{equation}
\label{lem:h_cond_pure}
\end{lem}

\begin{proof}
First, we prove \eqref{eq:h_pure_1} by induction. For $t=1$, $\hat{h}^1_\pure = \tau_0\mc{V}^{-1}( Z_0)\overset{d}{=}\tau_0 \tilde{\mbf{Z}}_0$. As the inductive hypothesis, assume $(\hat{h}^1_\pure,\ldots,\hat{h}^{t}_\pure) \overset{d}{=} (\tau_0 \tilde{\mbf{Z}}_0,\ldots, \tau_{t-1} \tilde{\mbf{Z}}_{t-1})$. 
By \eqref{eq:h_pure_def}, term $\hat{h}^{t+1}_\pure$ is equal in distribution to $\sum_{r=0}^{t-1} \hat{\alpha}^t_r \tau_r \tilde{\mbf{Z}}_r + \tau_t^\perp \mathbf{Z}$, where $\mathbf{Z}\in \mathbb{R}^{\Gamma}$ is independent of $\tilde{\mbf{Z}}_r$ for all $r=0,\ldots,t-1$. In what follows, we show
\begin{equation*}
\begin{split}
&(\tau_0 \tilde{\mbf{Z}}_0,\ldots, \tau_{t-1} \tilde{\mbf{Z}}_{t-1}, \sum_{r=0}^{t-1} \hat{\alpha}^t_r \tau_r\tilde{\mbf{Z}}_r + \tau_t^\perp \mathbf{Z})\\
&\hspace*{1.5in} \overset{d}{=} (\tau_0 \tilde{\mbf{Z}}_0,\ldots, \tau_{t-1} \tilde{\mbf{Z}}_{t-1},\tau_t \tilde{\mbf{Z}}_t).
\end{split}
\end{equation*}
Note that $\tilde{\mathbf{Z}}_0, \ldots, \tilde{\mathbf{Z}}_{t-1}, \mathbf{Z}$ are all zero-mean Gaussian, and therefore so is the sum.  We now study the variance and covariance of $ \sum_{r=0}^{t-1} \hat{\alpha}_r^t\tau_r \tilde{\mathbf{Z}}_r + \tau_t^{\perp}\mathbf{Z}$ by demonstrating the following two results:
\begin{enumerate}
\item[(\emph{i})]   For all $i, j\in\Gamma$, 
\begin{equation*}
\begin{split}
&\mathbb{E}\Big[\Big(\sum_{r=0}^{t-1} \hat{\alpha}_r^t\tau_r [\tilde{\mathbf{Z}}_{r}]_i + \tau_t^{\perp}\mathbf{Z}_i\Big)\Big(\sum_{r=0}^{t-1} \hat{\alpha}_r^t\tau_r[\tilde{\mathbf{Z}}_{r}]_j + \tau_t^{\perp}\mathbf{Z}_j\Big)\Big]\\
& = \tau_t^2  \mathbb{E}\left[[\tilde{\mathbf{Z}}_t]_i, [\tilde{\mathbf{Z}}_t]_j\right] = \begin{cases}\tau_t^2 &\text{ if } i=j, \\
0 &\text{ otherwise.}
\end{cases}
\end{split}
\end{equation*}
\item[(\emph{ii})]  For $0 \leq s \leq (t-1)$ and all $i, j\in\Gamma$, 
\begin{equation*}
\begin{split}
&\mathbb{E}\Big[\tau_s [\tilde{\mathbf{Z}}_{s}]_i \Big(\sum_{r=0}^{t-1} \hat{\alpha}_r^t\tau_r[\tilde{\mathbf{Z}}_{r}]_j + \tau_t^{\perp}\mathbf{Z}_j\Big)\Big] \\
&= \tau_s  \tau_t  \mathbb{E}\left[[\tilde{\mathbf{Z}}_s]_i, [\tilde{\mathbf{Z}}_t]_j\right] = \begin{cases}\breve{E}_{s,t} &\text{ if } i=j, \\
0 &\text{ otherwise.}
\end{cases}
\end{split}
\end{equation*}
\end{enumerate}
First consider (\emph{i}). We note,
\begin{align*}
&\!\mathbb{E}\Big[\! \Big(\sum_{r=0}^{t-1} \hat{\alpha}_r^t\tau_r [\tilde{\mathbf{Z}}_{r}]_i\! +\! \tau_t^{\perp}\mathbf{Z}_i\Big)\! \! \Big(\sum_{r=0}^{t-1}\! \hat{\alpha}_{r}^t\tau_r [\tilde{\mathbf{Z}}_{r}]_j + \tau_t^{\perp}\mathbf{Z}_j\Big) \!\Big]\\
& \overset{(a)}{=}\!  \sum_{r=0}^{t-1}\sum_{s=0}^{t-1}\hat{\alpha}^t_r\hat{\alpha}^t_s \tau_r \tau_s \mathbb{E}\left[ [\tilde{\mathbf{Z}}_r]_i [\tilde{\mathbf{Z}}_s]_j\right]\! +\! (\tau_t^\perp)^2\mathbb{E}\left[\mathbf{Z}_i \mathbf{Z}_j\right]\\
& \overset{(b)}{=} \begin{cases}
\sum_{r=0}^{t-1}\sum_{s=0}^{t-1}\hat{\alpha}^t_r\hat{\alpha}^t_s\breve{E}_{r,s} + (\tau_t^\perp)^2\overset{(c)}{=}\tau_t^2,&\text{if } i=j, \\
0, &\text{otherwise}.
\end{cases}
\end{align*}
In the above, step $(a)$ follows from the fact that $\mathbf{Z}$ is independent of $\tilde{\mathbf{Z}}_0, \ldots, \tilde{\mathbf{Z}}_{t-1}$, step $(b)$ from the covariance definition \eqref{eq:tildeZcov} and the i.i.d.\ standard normal nature of elements of $\mathbf{Z}$, and step $(c)$ from 
\begin{equation*}
\begin{split}
\sum_{r=0}^{t-1}\sum_{\ell =0}^{t-1}\hat{\alpha}^t_{r}\hat{\alpha}^t_{\ell}\breve{E}_{r,\ell} &=(\hat{\alpha}^t)^*\breve{C}^t\hat{\alpha}^t\\
& = [\breve{E}_t^*(\breve{C}^t)^{-1}](\breve{C}^t)^{-1}[(\breve{C}^t)^{-1}\breve{E}_t]\\
&=\breve{E}_t^*(\breve{C}^t)^{-1}\breve{E}_t = \breve{E}_{t,t}-(\tau_t^{\perp})^2.
\end{split}
\end{equation*}
Next, consider (\emph{ii}). We see that
\begin{align*}
&\mathbb{E} \Big[ \tau_s [\tilde{\mathbf{Z}}_{s}]_i \Big(\sum_{r=0}^{t-1} \hat{\alpha}^t_{r} \tau_r [\tilde{\mathbf{Z}}_r]_j + \tau_t^{\perp} \mathbf{Z}_j \Big)\Big] \\
& \overset{(a)}{=}\sum_{r=0}^{t-1}\hat{\alpha}^t_{r}\tau_s\tau_r \mathbb{E}\big[[\tilde{\mathbf{Z}}_s]_i [\tilde{\mathbf{Z}}_r]_j\big] \overset{(b)}{=} 
\begin{cases}
\sum_{r=0}^{t-1} \breve{E}_{s,r} \hat{\alpha}^t_{r},&\text{if } i=j,\\
0,&\text{otherwise}.
\end{cases}
\end{align*}
In the above, step $(a)$ follows since $\mathbf{Z}$ is independent of $\tilde{\mathbf{Z}}_s$ and step $(b)$ from \eqref{eq:tildeZcov}.
Finally, notice that $\sum_{r=0}^{t-1} \breve{E}_{s,r}\hat{\alpha}^t_{r} = [\breve{C}^t\hat{\alpha}^t]_{s+1} = \breve{E}_{s,t},$ where the first equality holds since the sum equals the inner product of the $(s+1)^{th}$ row of $\breve{C}^t$ with $\hat{\alpha}^t$ and the second equality by definition of $\hat{\alpha}^t$ in \eqref{eq:hatalph_hatgam_def}.

Next, we prove \eqref{eq:h_pure_2}, also by induction. For $t=0$, by \eqref{eq:Ha_dist} we have $h^{t+1}|_{\mscrs_{t+1,t}} \overset{d}{=} \tau_0 Z_0 + \Delta_{1,0} \overset{d}{=} h_{\pure}^1  + \Delta_{1,0}$. Assume that $h^{r+1}|_{\mscrs_{t+1,t}} \overset{d}{=} h^{r+1}_\pure + \sum_{i=0}^r \mathsf{d}^r_i \Delta_{i+1,i}$ holds for $r=0,\ldots,t-1$ as the inductive hypothesis. Then,
\begin{align*}
&h^{t+1}|_{\mscrs_{t+1,t}} \overset{d}{=} \sum_{r=0}^{t-1} \hat{\alpha}^t_r h^{r+1} + \tau^{\perp}_t Z_t + \Delta_{t+1,t}\\
& \overset{d}{=} \sum_{r=0}^{t-1} \hat{\alpha}^t_r \Big( h^{r+1}_\pure + \sum_{i=0}^r \mathsf{d}^r_i \Delta_{i+1,i}\Big) + \tau^{\perp}_t Z_t + \Delta_{t+1,t}\\
& = \sum_{r=0}^{t-1} \hat{\alpha}^t_r h^{r+1}_\pure + \tau^{\perp}_t Z_t  + \sum_{r=0}^{t-1} \sum_{i=0}^r \hat{\alpha}^t_r\mathsf{d}^r_i \Delta_{i+1,i}+ \Delta_{t+1,t}\\
& = h^{t+1}_\pure + \sum_{i=0}^t \mathsf{d}^t_i \Delta_{i+1,i}.
\end{align*}
In the above, the first equality uses \eqref{eq:Ha_dist} and the second the inductive hypothesis. The last equality follows by noticing that $\sum_{r=0}^{t-1} \sum_{i=0}^r v_{r,i} = \sum_{i=0}^{t-1} \sum_{r=i}^{t-1} v_{r,i}$ for $(v_{i,r})_{0\leq i,r \leq t-1}$ and using \eqref{eq:d_def}.
\end{proof}




\subsection{Main Concentration Lemma}
\label{subsec:main_lem_state}


\begin{lem}
We use the shorthand $X_n \doteq c$ to denote the concentration inequality $P(\abs{X_n-c} \geq \epsilon) \leq K_{k,t} e^{-\kappa_{k,t} n \e^2}$, where $K_{k,t}, \kappa_{k,t}$ denote constants depending on the iteration index $t$ and the fixed half-window size $k$, but not on $n$ or $\e$. The following statements hold for $0 \leq t < T^*$ and $\e\in(0,1)$.

\begin{enumerate}[(a)]

\item For $\Delta_{t+1,t}$ defined in \eqref{eq:D10} and \eqref{eq:Dt1t}, 
\begin{align}
P\Big(\frac{1}{|\Gamma|}\norm{\Delta_{{t+1,t}}}^2 \geq \epsilon \Big) \leq K_{k,t} e^{-\kappa_{k,t} n \e}. \label{eq:Ha}
\end{align}

\item For (order-2) pseudo-Lipschitz functions $\phi_h: \mathbb{R}^{(t+2)|\Lambda|} \to \mathbb{R}$,
\be
\begin{split}
&\frac{1}{|\Gamma|}\sum_{i\in\Gamma} \phi_h\Big([\hat{h}^1]_{\Lambda_i}, \ldots, [\hat{h}^{t+1}]_{\Lambda_i}, \beta_{\Lambda_i}\Big)\\
&  \doteq 
\frac{1}{|\Gamma|}\sum_{i\in\Gamma}\expec\Big[\phi_h\Big(\tau_0 [\tilde{\mathbf{Z}}_0]_{\Lambda_i}, \ldots, \tau_t [\tilde{\mathbf{Z}}_t]_{\Lambda_i}, \beta_{\Lambda_i} \Big)\Big].  
\end{split}
\label{eq:Hb1}
\ee
The random vectors $\tilde{\mathbf{Z}}_{0}, \ldots, \tilde{\mathbf{Z}}_t \in \mathbb{R}^{\Gamma}$ are jointly Gaussian with zero mean entries, which are independent of the other entries in the same vector with covariance across iterations given by \eqref{eq:tildeZcov}, and are independent of $\beta \sim \mu$.

\item Recall that the operator $\mc{V}$ rearranges the elements of an array into a vector,
\begin{align}
\frac{(h^{t+1})^* q^0}{n} &\doteq 0, \quad \frac{(h^{t+1})^* \mc{V}(\beta)}{n} \doteq 0, \label{eq:Hc} \\
\frac{(b^t)^* w}{n} &\doteq 0.  \label{eq:Bc}
\end{align}


\item For all $0 \leq r \leq t$, 
\begin{align}
\frac{(h^{r+1})^* h^{t+1}}{|\Gamma|} &\doteq \breve{E}_{r,t}, \label{eq:Hd} \\
\frac{(b^r)^*b^t}{n} &\doteq \tilde{E}_{r,t}. \label{eq:Bd}
\end{align}


\item For all $0 \leq r \leq t$,
\begin{align}
\frac{(q^{0})^* q^{t+1}}{n} &\doteq \tilde{E}_{0,t+1}, \quad \frac{(q^{r+1})^*q^{t+1}}{n} \doteq \tilde{E}_{r+1,t+1},  \label{eq:He}  \\
\frac{(m^r)^* m^t}{n} &\doteq \breve{E}_{r,t}. \label{eq:Be}
\end{align}

\item For all $0 \leq r \leq t$,  
\begin{align}
 \frac{(h^{t+1})^*q^{r+1}}{n} &\doteq  \hat{\lambda}_{r+1} \breve{E}_{r,t}, \quad \lambda_t \doteq \hat{\lambda}_{t}, \nonumber\\
\frac{(h^{r+1})^*q^{t+1}}{n} &\doteq  \hat{\lambda}_{t+1} \breve{E}_{r,t}, \label{eq:Hf} \\
\frac{(b^r)^*m^t}{n} &\doteq \hat{\xi}_{t}\tilde{E}_{r,t} ,\quad\quad \xi_t \doteq \hat{\xi}_{t},\nonumber\\
 \frac{(b^t)^*m^r}{n}& \doteq \hat{\xi}_{r} \tilde{E}_{r,t}.\label{eq:Bf}
\end{align}


\item For $\textbf{Q}_{t+1} = \frac{1}{n} Q_{t+1}^* Q_{t+1}$ and $\textbf{M}_{t} = \frac{1}{n} M_{t}^* M_{t}$, when the inverses exist, for all $0 \leq  i,j  \leq t$ and $0 \leq  i',j' \leq t-1$:
\begin{align}
&\left[\mathbf{Q}_{t+1}^{-1} \right]_{i+1,j+1} \doteq  [(\tilde{C}^{t+1})^{-1}]_{i+1,j+1}, \quad \gamma^{t+1}_{i} \doteq \hat{\gamma}^{t+1}_{i}, \label{eq:Hg}\\ 
&\left[\mathbf{M}_t^{-1} \right]_{i'+1,j'+1} \doteq  [(\breve{C}^t)^{-1}]_{i'+1,j'+1}, \quad \alpha^{t}_{i'} \doteq \hat{\alpha}^{t}_{i'}, \,   t \geq 1,\label{eq:Bg}  
\end{align}
where $\hat{\gamma}^{t+1}_{i}$ and $\hat{\alpha}^{t}_{i'}$ are defined in \eqref{eq:hatalph_hatgam_def},

\item With $\sigma_{t+1}^{\perp}, \tau_{t}^{\perp}$ defined in \eqref{eq:sigperp_defs},
\begin{align} 
\frac{1}{n}\norm{q^{t+1}_{\perp}}^2 &\doteq (\sigma_{t+1}^{\perp})^2, \label{eq:Hh} \\
\frac{1}{n}\norm{m^t_{\perp}}^2 &\doteq (\tau_{t}^{\perp})^2. \label{eq:Bh}
\end{align}

\end{enumerate}
\label{lem:main_lem}
\end{lem}


\subsection{Proof of Theorem \ref{thm:main_amp_perf}}

\begin{proof}
Applying part (b) of Lemma \ref{lem:main_lem} to a PL(2) function $\phi_h: \mathbb{R}^{2|\Lambda|} \to \mathbb{R}$,
\ben
\begin{split}
&P\Big(\Big\lvert \frac{1}{|\Gamma|} \sum_{i\in\Gamma} \Big(\phi_h(\hat{h}^{t+1}_{\Lambda_i}, \beta_{\Lambda_i}) - \mathbb{E}[\phi_h(\tau_t Z_{\Lambda_i}, \beta_{\Lambda_i})] \Big)\Big \lvert \geq \e \Big) \\
&\hspace*{2.5in}\leq K_{k,t} e^{-\kappa_{k,t} n \e^2} ,
\end{split}
\een
where the random field $\beta \in E^{\Gamma} \sim \mu$ is independent of $Z \in \mathbb{R}^{\Gamma}$ having i.i.d.\ standard normal entries. Now for $i \in \Gamma$ let
\be
 \phi_h(\hat{h}^{t+1}_{\Lambda_i}, \beta_{\Lambda_i}) := \phi(\eta_{t}(\beta_{\Lambda_i} - \hat{h}^{t+1}_{\Lambda_i}), \beta_i),
 \label{eq:phih_map}
\ee
where $\phi: \mathbb{R}^{2} \to \mathbb{R}$ is the PL(2) function in the statement of the theorem. The function  $\phi_h(\hat{h}^{t+1}_{\Lambda_i}, \beta_{\Lambda_i})$ in \eqref{eq:phih_map} is PL(2) since $\phi$ is PL(2) and $\eta_t$ is Lipschitz. We therefore obtain
\be
\begin{split}
&P\Big( \Big\lvert \frac{1}{|\Gamma|} \sum_{i\in\Gamma} \Big(\phi(\eta_{t}(\beta_{\Lambda_i} - \hat{h}^{t+1}_{\Lambda_i}), \beta_i)\\
&\qquad  - \mathbb{E}[ \phi(\eta_{t}(\beta_{\Lambda_i} - \tau_t Z_{\Lambda_i}), \beta_i)]\Big)\Big \lvert \geq \e\Big) \leq K_{k,t} e^{-\kappa_{k,t} n \e^2}. \nonumber
\end{split}
\ee
The proof is completed by noting from \eqref{eq:amp2} and \eqref{eq:hqbm_def_AMP} that 
$\beta^{t+1}_i = \eta_{t}([\mc{V}^{-1}\left(A^* z^t\right) + \beta^t]_{\Lambda_i}) = \eta_{t}(\beta_{\Lambda_i} - \hat{h}^{t+1}_{\Lambda_i})$.
\end{proof}


\section{Proof of Lemma \ref{lem:main_lem}} \label{sec:main_lem_proof}

We make use of concentration results listed in Appendices \ref{app:conc_lemma}, \ref{app:other}, and \ref{app:conc_dependent}, where  
Appendix \ref{app:conc_dependent} contains concentration results for dependent variables that were needed to provide the new results in this paper. Note that the lemmas that are stated in Appendices are labeled by capital letters with numbers (e.g., Lemma \ref{sums}), whereas the lemmas that are stated in the body are labeled by numbers (e.g., Lemma \ref{lem:Ct_invert}).

The proof of  Lemma \ref{lem:main_lem} proceeds by induction on $t$.  We label as $\mathcal{H}^{t+1}$ the results \eqref{eq:Ha}, \eqref{eq:Hb1}, \eqref{eq:Hc}, \eqref{eq:Hd}, \eqref{eq:He}, \eqref{eq:Hf}, \eqref{eq:Hg}, \eqref{eq:Hh} and similarly as $\mathcal{B}^t$ the results 
\eqref{eq:Bc}, \eqref{eq:Bd}, \eqref{eq:Be}, \eqref{eq:Bf}, \eqref{eq:Bg}, \eqref{eq:Bh}.  The proof consists of four steps: \textbf{(1)} proving that $\mathcal{B}_0$ holds, \textbf{(2)} proving that $\mathcal{H}_1$ holds, \textbf{(3)} assuming that $\mathcal{B}_r$ and $\mathcal{H}_s$ hold for all $r < t $ and $s \leq t $, then proving that $\mathcal{B}_t$ holds, and \textbf{(4)} assuming that $\mathcal{B}_r$ and $\mathcal{H}_s$ hold for all $r \leq t $ and $s \leq t$, then proving that $\mathcal{H}_{t+1}$ holds.  

The proof of steps \textbf{(1)} and \textbf{(3)} -- the $\mathcal{B}$ steps -- follow as in \cite{RushV18}. To see that this is the case, notice that in the proof for \cite[Lemma 6]{RushV18}, the results given by $\mathcal{B}_t (b)$ - $\mathcal{B}_t(h)$ involve the sequence of functions $\{g_r\}_{0\leq r\leq t}$ and the sequence of vectors $\{b^r\}_{0\leq r\leq t}$. In our case, the definition of the (separable) functions $g_t(\cdot)$ in \eqref{eq:hqbm_def} is the same that in \cite[(4.1)]{RushV18} and the conditional distribution of $\{b^r\}_{0\leq r\leq t}$ given in Lemma \ref{lem:hb_cond} has the same expression as that in \cite[Lemma 6]{RushV18}. 
Therefore, the proof of $\mathcal{B}_t (b)$ - $\mathcal{B}_t(h)$  in \cite{RushV18} is directly applicable here. 

Now consider $\mathcal{B}_t (a)$. When $t=0$, it only involves $\|q^0\|$, which has the same assumption in our case and in \cite{RushV18}. When $t>0$, the proof uses induction hypothesis $\mathcal{B}_{0} (d)$-$\mathcal{B}_{t-1} (d)$, $\mathcal{B}_{t-1}(e)$, and $\mathcal{H}_{t}(g), \mathcal{H}_{t}(h)$. The statement of those hypotheses have the same form as in \cite{RushV18}. Although the concentration constants in those hypotheses have different definitions in this paper due to using non-separable denoisers, it does not change the proof for $\mathcal{B}_t (a)$, since the actual values of the concentration constants are not involved in the proof. Therefore, we do not repeat the proof of steps \textbf{(1)} and \textbf{(3)} here.
In what follows, we only show steps \textbf{(2)} and \textbf{(4)}.

For each step, in parts $(a)$--$(h)$ of the proof, we use $K$ and $\kappa$ to label universal constants,  meaning that they do not depend on $n$ or $\e$, but may depend on $t$ and $k$, in the concentration upper bounds.

\subsection{Step 2: Showing that $\mc{H}_1$ holds}
  
Throughout the proof we will make use of a function $\mc{S}:\mathbb{R}^{\Lambda}\to\mathbb{R}$ that selects the center coordinate of its argument. For example, for $v \in \mathbb{R}^{|\Gamma|}$, 
\be
\mc{S}(v_{\Lambda_i}) = v_i.
\label{eq:operator_S}
\ee
We will only use $\mc{S}$ in cases where such a ``center point" is well-defined. Notice that $\mc{S}$ is Lipschitz, since $|\mc{S}(x)-\mc{S}(x')|=|x_c-x_c'|\leq \|x-x'\|,$ for all $x,x'\in\mathbb{R}^{\Lambda}$, where $x_c$ (respectively, $x'_c$) is the center coordinate of $x$ (respectively, $x'$). Moreover, if a function $f:\mathbb{R}^{\Lambda\times\tilde{\Lambda}}\to\mathbb{R}$ is defined as $f(x,y):=\mc{S}(x)$ with arbitrary but fixed $\tilde{\Lambda}$, then $f$ is Lipschitz, because $|f(x,y)-f(x',y')|=|\mc{S}(x)-\mc{S}(x')| \leq \|x-x'\| \leq \|(x,y)-(x',y')\|$.  We are now ready to prove $\mc{H}_1 (a) - \mc{H}_1 (h)$.

\textbf{(a)} 
The definition of $\Delta_{1,0}$ is given in \eqref{eq:D10}. First notice that by Lemma \ref{fact:gauss_p0}, we have $\mathsf{P}_{q^0}^{\parallel} Z_0 \overset{d}{=} \frac{q^0}{\|q^0\|}\bar{Z}_0$, where $\bar{Z}_0$ is a standard normal random variable on $\mathbb{R}$. Using this fact and \eqref{eq:D10}, we have
\begin{equation}
\begin{split}
\Delta_{1,0} &\overset{d}{=} \Big(\frac{\|m^0\|}{\sqrt{n}} - \tau_0\Big) Z_0 - \frac{\|m^0\|}{\sqrt{n}}\Big(\frac{q^0}{\|q^0\|}\Big)\bar{Z}_0\\
& + q^0\Big(\frac{n}{\|q^0\|^2}\Big)\Big(\frac{(b^0)^*m^0}{n} - \frac{\xi_0 \|q^0\|^2}{n}\Big).
\end{split}
\label{eq:D10_H1a}
\end{equation}
By applying the triangle inequality to the norm of the RHS of \eqref{eq:D10_H1a} and then applying Lemma \ref{sums}, 
\begin{align}
&P\Big(\frac{\|\Delta_{1,0}\|^2}{|\Gamma|} \geq \e \Big) = P\Big(\frac{\|\Delta_{1,0}\|}{\sqrt{|\Gamma|}} \geq \sqrt{\e} \Big) \nonumber \\
& \leq P\Big(\Big\vert \frac{\|m^0\|}{\sqrt{n}}-\tau_0 \Big\vert \frac{\|Z_0\|}{\sqrt{|\Gamma|}} \geq \frac{\sqrt{\e}}{3} \Big) + P\Big( \frac{\|m^0\| |\bar{Z}_0|}{\sqrt{n|\Gamma|}} \geq \frac{\sqrt{\e}}{3} \Big)\nonumber\\
&+ P\Big(\Big\vert \frac{(b^0)^*m^0}{\sqrt{n}\|q^0\|}  - \frac{\xi_0 \|q^0\|}{\sqrt{n}}\Big\vert \geq \frac{\sqrt{\e}}{3\sqrt{\delta}}\Big).
\label{eq:3terms_H1a}
\end{align}
Label the three terms on the right-hand side (RHS) of \eqref{eq:3terms_H1a} as $T_1 - T_3$. We will show that each term is bounded by $Ke^{-\kappa n \e}$.

First consider $T_1$.
\begin{align}
T_1 &\leq P\Big(\Big\vert \frac{\|m^0\|}{\sqrt{n}}-\tau_0 \Big\vert \Big( \Big\vert \frac{\|Z_0\|}{\sqrt{|\Gamma|}} - 1\Big\vert + 1\Big) \geq \frac{\sqrt{\e}}{3}\Big)\nonumber\\ 
&\overset{(a)}{\leq}  P\Big(\Big\vert \frac{\|m^0\|}{\sqrt{n}}-\tau_0 \Big\vert \geq \frac{\sqrt{\e}}{6}\Big) + P\Big(\Big\vert \frac{\|Z_0\|}{\sqrt{|\Gamma|}} - 1\Big\vert \geq \frac{\sqrt{\e}}{3}\Big)\nonumber\\
&\overset{(b)}{\leq} Ke^{-\kappa n \e} + Ke^{-\kappa n \e},
\label{eq:T1_H1a}
\end{align}
where step $(b)$ follows by Lemma \ref{sqroots}, Lemma \ref{subexp}, and $\mathcal{B}_0(e)$.
To see that step $(a)$ in \eqref{eq:T1_H1a} holds, we notice that
\begin{equation}
\begin{split}
&\Big\{\Big\vert \frac{\|m^0\|}{\sqrt{n}}-\tau_0 \Big\vert < \frac{\sqrt{\e}}{6}\Big\} \cap \Big\{\Big\lvert \frac{\|Z_0\|}{\sqrt{|\Gamma|}} - 1\Big\vert < \frac{\sqrt{\e}}{3}\Big\}\\
& \subset \Big\{ \Big\vert \frac{\|m^0\|}{\sqrt{n}}-\tau_0\Big\vert \Big(\Big\vert \frac{\|Z_0\|}{\sqrt{|\Gamma|}} - 1\Big\vert + 1\Big) < \frac{\sqrt{\e}}{3} \Big\},
\end{split}
\label{eq:intersect_H1a}
\end{equation}
since if the two events on the left-hand side (LHS) of \eqref{eq:intersect_H1a} hold, then using that $\e<1$,
\begin{equation*}
\Big\vert \frac{\|m^0\|}{\sqrt{n}}-\tau_0 \Big\vert \Big( \Big\vert \frac{\|Z_0\|}{\sqrt{|\Gamma|}} - 1\Big\vert + 1\Big) < \frac{\sqrt{\e}}{6}\Big(\frac{\sqrt{\e}}{3}+1\Big) < \frac{\sqrt{\e}}{3}.
\end{equation*}
Taking the complement on both sides of \eqref{eq:intersect_H1a},
\begin{equation*}
\begin{split}
&\Big\{ \Big\vert \frac{\|m^0\|}{\sqrt{n}}-\tau_0 \Big\vert \Big(\Big\vert \frac{\|Z_0\|}{\sqrt{|\Gamma|}} - 1\Big\vert + 1\Big) \geq \frac{\sqrt{\e}}{3} \Big\}\\
& \subset \Big\{\Big\vert \frac{\|m^0\|}{\sqrt{n}}-\tau_0\Big\vert \geq \frac{\sqrt{\e}}{6}\Big\} \cup \Big\{\Big\vert \frac{\|Z_0\|}{\sqrt{|\Gamma|}} - 1\Big\vert \geq \frac{\sqrt{\e}}{3}\Big\}.
\end{split}
\end{equation*}
Then step $(a)$ in \eqref{eq:T1_H1a} follows by the union bound.

Next consider $T_2$.
\begin{equation*}
\begin{split}
&T_2 \leq P\Big( \Big(\Big\vert  \frac{\|m^0\|}{\sqrt{n}} - \tau_0 \Big\vert + \tau_0\Big)\frac{|\bar{Z}_0|}{\sqrt{\Gamma}} \geq \frac{\sqrt{\e}}{3}\Big)\\ 
&\overset{(a)}{\leq} P\Big(\Big\vert  \frac{\|m^0\|}{\sqrt{n}} - \tau_0 \Big\vert\geq \frac{\sqrt{\e}}{3}\Big) 
   + P\Big( \frac{|\bar{Z}_0|}{\sqrt{\Gamma}} \geq \frac{\sqrt{\e}}{6}\min(\tau_0^{-1},1)\Big)\\
&\overset{(b)}{\leq} Ke^{-\kappa n \e} + Ke^{-\kappa n \e},
\end{split}
\end{equation*}
where step $(a)$ follows by similar justification as that for step $(a)$ in \eqref{eq:T1_H1a} and step $(b)$ follows by Lemma \ref{sqroots}, Lemma \ref{lem:normalconc}, and $\mathcal{B}_0(e)$.

Finally consider $T_3$.
\begin{equation*}
\begin{split}
T_3 &\overset{(a)}{\leq} P\Big( \Big\vert \frac{(b^0)^* m^0}{n} \frac{\sqrt{n}}{\|q^0\|} - \hat{\xi}_0\sigma_0 \Big\vert \geq \frac{\sqrt{\e}}{6\sqrt{\delta}} \Big)\\
& + P\Big( \xi_0 \frac{\|q^0\|}{\sqrt{n}} -\hat{\xi}_0\sigma_0\geq \frac{\sqrt{\e}}{6\sqrt{\delta}} \Big) \overset{(b)}{\leq} Ke^{-\kappa n \e } + Ke^{-\kappa n \e},
\end{split}
\end{equation*}
where step $(a)$ follows by Lemma \ref{sums} and step $(b)$ follows by Lemma \ref{products}, $\mathcal{B}_0(f)$, and the assumption on $\|q^0\|$ given in \eqref{eq:qassumption}.

\textbf{(b)} 
Let $\hat{Z}_0:=\mc{V}^{-1}(Z_0) \in \mathbb{R}^{\Gamma}$ and $\hat{\Delta}_{1,0}:=\mc{V}^{-1}(\Delta_{1,0}) \in \mathbb{R}^{\Gamma}$ be the array versions of the vectors $Z_0$ and $\Delta_{1,0}$, respectively. For $t=0$, the left-hand side (LHS) of \eqref{eq:Hb1} can be bounded as
\begin{align}
&P\Big(\Big \lvert \frac{1}{|\Gamma|} \sum_{i\in\Gamma} \Big( \phi_h (\hat{h}^{1}_{\Lambda_i}, \beta_{\Lambda_i}) - \mathbb{E}[\phi_h(\tau_0 [\tilde{\mathbf{Z}}_0]_{\Lambda_i}, \beta_{\Lambda_i})]\Big) \Big \lvert \geq \epsilon \Big) \nonumber\\
&\overset{(a)}{=} P\Big(\Big \lvert \frac{1}{|\Gamma|} \sum_{i\in\Gamma} \Big( \phi_h([\tau_0 \hat{Z}_{0} + \hat{\Delta}_{1,0}]_{\Lambda_i}, \beta_{\Lambda_i})\nonumber\\
& \hspace*{1.3in} - \mathbb{E}[\phi_h(\tau_0 [\tilde{\mathbf{Z}}_0]_{\Lambda_i}, \beta_{\Lambda_i})] \Big) \Big \lvert \geq \epsilon \Big)  \nonumber\\
&\overset{(b)}{\leq} P\Big(\Big \lvert \frac{1}{|\Gamma|} \sum_{i\in\Gamma} \Big(\mathbb{E}_{\tilde{\mathbf{Z}}_0}[\phi_h(\tau_0 [\tilde{\mathbf{Z}}_0]_{\Lambda_i},  \beta_{\Lambda_i})]\nonumber\\
& \hspace*{1.3in}- \mathbb{E}_{\tilde{\mathbf{Z}}_0, \beta}[\phi_h(\tau_0 [\tilde{\mathbf{Z}}_0]_{\Lambda_i}, \beta_{\Lambda_i})]\Big) \Big \lvert \geq \frac{\epsilon}{3} \Big)\nonumber \\
&  + P\Big(\Big \lvert \frac{1}{|\Gamma|} \sum_{i\in\Gamma} \Big(\phi_h(\tau_0 [\hat{Z}_{0}]_{\Lambda_i},  \beta_{\Lambda_i})\nonumber\\
&\hspace*{1.3in} - \mathbb{E}_{\tilde{\mathbf{Z}}_0}[\phi_h(\tau_0 [\tilde{\mathbf{Z}}_0]_{\Lambda_i},  \beta_{\Lambda_i})]\Big) \Big \lvert \geq \frac{\epsilon}{3} \Big) \nonumber\\
& + P\Big(\Big \lvert \frac{1}{|\Gamma|} \sum_{i\in\Gamma} \Big(\phi_h([\tau_0 \hat{Z}_{0} + \hat{\Delta}_{1,0}]_{\Lambda_i},  \beta_{\Lambda_i})\nonumber\\
&\hspace*{1.3in} - \phi_h(\tau_0 [\hat{Z}_{0}]_{\Lambda_i}, \beta_{\Lambda_i})\Big) \Big \lvert \geq \frac{\epsilon}{3} \Big).
\label{eq:phih_0}
\end{align}
Step $(a)$ follows from the conditional distribution of $h^1$ given in Lemma \ref{lem:hb_cond} \eqref{eq:Ha_dist} and since $\hat{h}^1=\mc{V}^{-1}(h^1)$ and $\tau_0\hat{Z}_0+\hat{\Delta}_{1,0}=\mc{V}^{-1}(\tau_0 Z_0 + \Delta_{1,0})$. Step $(b)$ follows from Lemma \ref{sums}. Label the terms on the RHS of \eqref{eq:phih_0} as $T_1 -T_3$. We show that each of these terms  is bounded above by $K e^{-\kappa n \e^2}.$ 

First, consider $T_1$. Recall the definition of the functions $\mc{T}_i$ for $i\in\Gamma$ in \eqref{eq:def_T}, which extends an array in $\mathbb{R}^{\Lambda_i\cap\Gamma}$ to an array in $\mathbb{R}^{\Lambda}$ by defining the extended entries to be the average of the entries in the original array. For arbitrary but fixed $s\in\mathbb{R}^{\Lambda}$, the function $\tilde{\phi}_{h,i}:\mathbb{R}^{\Lambda_i\cap\Gamma\times\Lambda}\to\mathbb{R}$ defined as $\tilde{\phi}_{h,i}(v,s):=\phi_h(\mc{T}_i(v),s)$ is PL(2) by Lemma \ref{lem:PLwithAvg}. Then it follows from Lemma \ref{lem:expZ} that the function $\phi_{1,i}:\mathbb{R}^{\Lambda}\to\mathbb{R}$ 
defined as $\phi_{1,i}(s):=\mathbb{E}_{\tilde{\mbf{Z}}_0}[\tilde{\phi}_{h,i}(\tau_0[\tilde{\mbf{Z}}_0]_{\Lambda_i\cap\Gamma},s)]$ is PL(2), since $[\tilde{\mbf{Z}}_0]_{\Lambda_i\cap\Gamma}$ is an array of i.i.d.\ standard norm random variables for all $i\in\Gamma$. Notice that $\mathbb{E}_{\tilde{\mbf{Z}}_0}[\tilde{\phi}_{h,i}(\tau_0[\tilde{\mbf{Z}}_0]_{\Lambda_i\cap\Gamma},s)]=\mathbb{E}_{\tilde{\mbf{Z}}_0}[\phi_{h}( \tau_0[\tilde{\mbf{Z}}_0]_{\Lambda_i},s)]$ by the definition of $\tilde{\phi}_{h,i}$ and $\mc{T}_i$. Therefore,
\begin{equation*}
\begin{split}
T_1 &\overset{(a)}{=}P\Big(\Big \lvert \frac{1}{|\Gamma|} \sum_{i\in\Gamma} \Big[\phi_{1,i}(\mc{T}_i(\beta_{\Lambda_i\cap\Gamma})) - \mathbb{E}[\phi_{1,i}(\mc{T}_i(\beta_{\Lambda_i\cap\Gamma}))]\Big] \Big \lvert \geq \e \Big)\\
&\overset{(b)}{\leq} Ke^{-\kappa n \e^2},
\end{split}
\end{equation*}
where in step $(a)$ we use the definition of $\mc{T}_i$ in \eqref{eq:def_T} and step $(b)$ follows from Lemma \ref{lem:PL_MRF_conc} by noticing from Lemma \ref{lem:PLwithAvg} that the function $\phi_{1,i}(\mc{T}_i(\cdot))$ is PL(2) for all $i\in\Gamma$.

Next, consider $T_2$. We use iterated expectation to condition on the value of $\beta$.  Then $T_2$ an be expressed as an expectation as follows,
\begin{align*}
T_2 &= \mathbb{E}_{\beta} \Big[P\Big(\Big \lvert \frac{1}{|\Gamma|} \sum_{i\in\Gamma}(\phi_h(\tau_0 [\hat{Z}_{0}]_{\Lambda_i},  \beta_{\Lambda_i})\\
&\qquad\qquad\qquad - \mathbb{E}_{\tilde{\mathbf{Z}}_0}[\phi_h(\tau_0 [\tilde{\mathbf{Z}}_0]_{\Lambda_i},  \beta_{\Lambda_i})]) \Big \lvert \geq \frac{\epsilon}{3} \,\, \Big \vert  \,\, \beta  \Big)\Big].
\end{align*}
Define the function $f:E^{\Gamma}\to [0,1]$ as 
\begin{align*}
f(a)&:=P\Big( \Big \lvert \frac{1}{|\Gamma|} \sum_{i\in\Gamma} \Big(\phi_h(\tau_0 [\hat{Z}_{0}]_{\Lambda_i},  \beta_{\Lambda_i})\\
& \qquad\qquad- \mathbb{E}_{\tilde{\mathbf{Z}}_0}[\phi_h(\tau_0 [\tilde{\mathbf{Z}}_0]_{\Lambda_i},  \beta_{\Lambda_i})]\Big) \Big \lvert \geq \frac{\epsilon}{3} \,\, \Big \vert \,\, \beta=a  \Big)\\
&=P\Big(\Big \lvert \frac{1}{|\Gamma|} \sum_{i\in\Gamma} \Big(\phi_h(\tau_0 [\hat{Z}_{0}]_{\Lambda_i},  a_{\Lambda_i})\\
&\qquad\qquad - \mathbb{E}_{\tilde{\mathbf{Z}}_0}[\phi_h(\tau_0 [\tilde{\mathbf{Z}}_0]_{\Lambda_i},  a_{\Lambda_i})]\Big) \Big \lvert \geq \frac{\epsilon}{3}  \Big).
\end{align*}
For any fixed $a\in E^{\Gamma}$, define a function $\phi_{2,i}:\mathbb{R}^{\Lambda}\to\mathbb{R}$ as $\phi_{2,i}(s):=\phi_h(s,a_{\Lambda_i})$ for each $i\in\Gamma$ and note that it is PL(2) with PL constant upper-bounded by $L(1+2\sqrt{|\Lambda|}M)$, where $L$ is the PL constant for $\phi_h$ and $M$ is such that $|x|\leq M$ for all $x\in E$, since by the pseudo-Lipschitz property of $\phi_h$ and the triangle inequality,
\begin{align*}
&\abs{\phi_{2,i}(x) - \phi_{2,i}(x)} = \abs{\phi_h(x,a_{\Lambda_i})-\phi_h(y,a_{\Lambda_i})} \\ 
&\qquad\leq L\Big(1+2\|a_{\Lambda_i}\| + \|x\|+\|y\|\|\Big)\|x-y\| \\
&\qquad\leq L(1+2\sqrt{|\Lambda|}M)(1+\|x\|+\|y\|)\|x-y\|.
\end{align*}
Using $L(1+2\sqrt{|\Lambda|}M)$ as the PL constant for $\phi_{2,i}$ for all $i\in\Gamma$, then 
\begin{equation*}
\begin{split}
f(a) &= P\Big( \Big\lvert \frac{1}{|\Gamma|} \sum_{i\in\Gamma} \Big(\phi_h(\tau_0 \mc{T}_i([\hat{Z}_{0}]_{\Lambda_i\cap\Gamma}),  a_{\Lambda_i})\\
& - \mathbb{E}_{\tilde{\mathbf{Z}}_0}[\phi_h(\tau_0 \mc{T}_i([\tilde{\mathbf{Z}}_0]_{\Lambda_i\cap\Gamma}),  a_{\Lambda_i})] \Big) \Big \lvert \geq \frac{\epsilon}{3}  \Big)\leq K e^{-\kappa |\Gamma| \e^2},
\end{split}
\end{equation*} 
where the last inequality follows from Lemma \ref{lem:PL_overlap_gauss_conc_ext} by noticing that $\mathbb{E}_{\tilde{\mathbf{Z}}_0}[\phi_h(\tau_0 \mc{T}_i([\tilde{\mathbf{Z}}_0]_{\Lambda_i\cap\Gamma}),  a_{\Lambda_i})]=\mathbb{E}_{\hat{Z}_0}[\phi_h(\tau_0 \mc{T}_i([\hat{Z}_0]_{\Lambda_i\cap\Gamma}),  a_{\Lambda_i})]$. Therefore, $T_2 = \mathbb{E}[f(\beta)]\leq K e^{-\kappa |\Gamma| \e^2}$, since $K$ and $\kappa$ don't depend on $\beta$ (as it doesn't show up in the pseudo-Lipschitz constant $L(1+2\sqrt{|\Lambda|}M)$). 

Finally, consider $T_3$, the third term on the RHS of \eqref{eq:phih_0}.
\begin{align}
T_3 &\overset{(a)}{\leq} P\Big(\frac{1}{|\Gamma|} \sum_{i\in\Gamma} L (1 + ||[\tau_0 \hat{Z}_{0} + \hat{\Delta}_{1,0}]_{\Lambda_i}||\nonumber\\
&\hspace*{1in} + ||\tau_0 [\hat{Z}_{0}]_{\Lambda_i}||)||[\hat{\Delta}_{1,0}]_{\Lambda_i}|| \geq \frac{\epsilon}{3} \Big) \nonumber \\
&\overset{(b)}{\leq} P\Big( \frac{1}{\sqrt{|\Gamma|}}||\hat{\Delta}_{1,0}||  \Big(1 + \sqrt{\frac{2d}{|\Gamma|}} ||\hat{\Delta}_{1,0}||\nonumber\\
&\hspace*{1in} + 2\tau_0  \sqrt{\frac{2d}{|\Gamma|}} ||\hat{Z}_{0}||\Big)\geq \frac{\epsilon}{3L\sqrt{6d}} \Big).  \label{eq:B1func1eq1}
\end{align}
Step $(a)$ follows from the fact that $\phi_h$ is PL(2). 
Step $(b)$ uses $||[\tau_0 \hat{Z}_{0} + \hat{\Delta}_{1,0}]_{\Lambda_i}|| \leq ||\tau_0 [\hat{Z}_{0}]_{\Lambda_i}|| + ||[\hat{\Delta}_{1,0}]_{\Lambda_i}||$ by the triangle inequality, the Cauchy-Schwarz inequality, the fact that for $a \in \mathbb{R}^\Gamma$, $\sum_{i\in\Gamma} \norm{a_{\Lambda_i}}^2 \leq 2d \norm{a}^2$, where $d=|\Lambda| = (2k+1)^p$, and the following application of Lemma \ref{lem:squaredsums}:
\ben
\begin{split}
&\sum_{i\in\Gamma} (1 + ||[\hat{\Delta}_{1,0}]_{\Lambda_i}|| + 2||\tau_0 [\hat{Z}_{0}]_{\Lambda_i}||)^2 \\
&\leq 3(|\Gamma| + 2d ||\hat{\Delta}_{1,0}||^2 + 4 \tau_0^2 2d ||\hat{Z}_{0}||^2).
\end{split}
\een
From \eqref{eq:B1func1eq1},  we have
\ben
\begin{split}
T_3 &\leq  P\Big( \frac{\lvert \lvert \hat{Z}_{0} \lvert \lvert }{\sqrt{|\Gamma|}}\geq 2 \Big) + P\Big( \frac{ \lvert \lvert\hat{\Delta}_{1,0}\lvert \lvert}{\sqrt{|\Gamma|}} \geq \frac{\frac{\e}{\sqrt{2d}}  \min\{1, \frac{1}{3L \sqrt{3}}\}}{2 + 4 \tau_0 \sqrt{2d}} \Big) \\ &\overset{(a)}{\leq} e^{-\delta n} + K e^{-\kappa n \e^2},
\end{split}
\een
where we use Lemma \ref{subexp} and  $\mc{H}_1 (a)$ to obtain step $(a)$.

\textbf{(c)} We first show concentration for $\frac{1}{n}(h^{1})^* \mc{V}(\beta)=\frac{1}{n}\sum_{i\in\Gamma}\hat{h}_i^1\beta_i$. Let the function $\phi_1:\mathbb{R}^{2|\Lambda|}\to  \mathbb{R}$ be defined as $\phi_1(x,y):=\mc{S}(x)\mc{S}(y)$ for any $(x,y)\in\mathbb{R}^{\Lambda\times\Lambda}$, where the operator $\mc{S}$ is defined in \eqref{eq:operator_S}. Then, using the fact that $\phi_1(\hat{h}^1_{\Lambda_i},\beta_{\Lambda_i})=\hat{h}_i^1\beta_i$ and $\mathbb{E}[\phi_1(\tau_0[\tilde{\mathbf{Z}}_0]_{\Lambda_i},\beta_{\Lambda_i})]=\mathbb{E}[[\tau_0\tilde{\mathbf{Z}}_0]_i]\mathbb{E}[\beta_i]=0$ for all $i\in\Gamma$, since $[\tilde{\mathbf{Z}}_0]_i$ has zero-valued mean and is independent of $\beta_i$, we find
\begin{align*}
&P\Big( \Big \lvert \frac{(h^{1})^* \mc{V}(\beta)}{n} \Big \lvert  \geq \e\Big)\\
&= P\Big( \Big \lvert  \frac{1}{|\Gamma|}\sum_{i\in\Gamma}\Big( \phi_1(\hat{h}^1_{\Lambda_i},\beta_{\Lambda_i}) - \mathbb{E}[\phi_1(\tau_0[\tilde{\mathbf{Z}}_0]_{\Lambda_i},\beta_{\Lambda_i})]\Big) \Big \lvert  \geq \delta\e \Big). 
\end{align*}
Finally, note that $\phi_1$ is PL(2) since $\mc{S}$ is Lipschitz by Lemma \ref{lem:Lprods}, hence, we can apply $\mc{H}_1 (b)$ to give the desired upper bound. 

Next, we show concentration for $\frac{1}{n}(h^{1})^* q^0 = \frac{1}{n}\sum_{i\in\Gamma}\hat{h}_i^1 \hat{q}^0_i$. Recall, $\hat{q}^0_i=f_0(\mathbf{0},\beta_{\Lambda_i})$ for all $i\in\Gamma$. The function $\phi_2:\mathbb{R}^{2|\Lambda|}\to\mathbb{R}$ defined as $\phi_2(x,y):=\mc{S}(x)f_0(\mathbf{0},y)$ is PL(2) by Lemma \ref{lem:Lprods} since $\mc{S}$ and $f_0$ are both Lipschitz. Notice that $\phi_2(\hat{h}^1_{\Lambda_i},\beta_{\Lambda_i})=\hat{h}_i^1\hat{q}_i^0$ and $\mathbb{E}[\phi_2(\tau_0[\tilde{\mathbf{Z}}_0]_{\Lambda_i},\beta_{\Lambda_i})]=\mathbb{E}[\tau_0[\tilde{\mathbf{Z}}_0]_i]\mathbb{E}[f_0(\mathbf{0},\beta_{\Lambda_i})]=0$ for all $i \in \Gamma$ since $[\tilde{\mathbf{Z}}_0]_i$ has zero-valued mean and is independent of $\beta$.  
Therefore, using $\mc{H}_1 (b)$, 
\begin{align*}
&P\Big(\Big \lvert \frac{(h^{1})^* q^0}{n} \Big \lvert  \geq \e\Big)\\
&= P\Big(\Big \lvert  \frac{1}{|\Gamma|}\sum_{i\in\Gamma}\Big( \phi_2(\hat{h}^1_{\Lambda_i},\beta_{\Lambda_i}) - \mathbb{E}[\phi_2(\tau_0[\tilde{\mathbf{Z}}_0]_{\Lambda_i},\beta_{\Lambda_i})]\Big) \Big \lvert  \geq \delta\e\Big)\\
&\leq Ke^{-\kappa n \e^2}.
\end{align*}

\textbf{(d)}  The function $\phi_3:\mathbb{R}^{2|\Lambda|}\to\mathbb{R}$ defined as $\phi_3(x,y):=(\mc{S}(x))^2$ is PL(2) by Lemma \ref{lem:Lprods} since the operator $\mc{S}$ defined in \eqref{eq:operator_S} is Lipschitz. Notice that $\frac{1}{|\Gamma|}(h^1)^*h^1=\frac{1}{|\Gamma|}\sum_{i\in\Gamma}\phi_3(\hat{h}^1_{\Lambda_i},\beta_{\Lambda_i})$ and $\mathbb{E}[\phi_3(\tau_0[\tilde{\mathbf{Z}}_0]_{\Lambda_i},\beta_{\Lambda_i})]=\tau_0^2\mathbb{E}[([\tilde{\mathbf{Z}}_0]_i)^2]=\tau_0^2$ for all $i\in\Gamma$, which follows from the definition of $\tilde{\mathbf{Z}}_0$ in \eqref{eq:tildeZcov}. Therefore, the result follows using $\mc{H}_1 (b)$, since
\begin{align*}
&P\Big(\Big \lvert  \frac{1}{|\Gamma|}\norm{h^1}^2 - \tau_0^2 \Big \lvert  \geq  \e\Big)\\
& = P\Big(\Big \lvert  \frac{1}{|\Gamma|} \sum_{i\in\Gamma} \Big(\phi_3(\hat{h}^1_{\Lambda_i},\beta_{\Lambda_i}) - \mathbb{E}[\phi_3(\tau_0[\tilde{\mathbf{Z}}_0]_{\Lambda_i},\beta_{\Lambda_i})]\Big) \Big \lvert  \geq  \e\Big).
\end{align*}

\textbf{(e)} We prove concentration for $\frac{1}{n}(q^0)^*q^1$, and the result for $\frac{1}{n}(q^1)^*q^1$ follows similarly. The function $\phi_4:\mathbb{R}^{2|\Lambda|}\to\mathbb{R}$ defined as $\phi_4(x,y):=f_0(\mathbf{0},y)f_1(x,y)$ is PL(2) by Lemma \ref{lem:Lprods}, since $f_0$ and $f_1$ are Lipschitz. Notice that $\frac{1}{|\Gamma|}\sum_{i\in\Gamma}\mathbb{E}[f_0(\mathbf{0}, \beta_{\Lambda_i}) f_1(\tau_0 [\tilde{\mathbf{Z}}_0]_{\Lambda_i}, \beta_{\Lambda_i})] = \delta \tilde{E}_{0,1}$ by \eqref{eq:Edef} and $(q^0)^*q^1 = \sum_{i\in\Gamma}\phi_4(\hat{h}^1_{\Lambda_i},\beta_{\Lambda_i})$. Hence, we have the desired upper bound using $\mc{H}_1(b)$, since 
\begin{align*}
&P\Big (\Big \lvert \frac{1}{n}(q^0)^*q^1 - \tilde{E}_{0,1} \Big \lvert   \geq \e \Big)\\
& = P\Big(\Big \lvert \frac{1}{|\Gamma|}\sum_{i\in\Gamma} \Big(\phi_4(\hat{h}^1_{\Lambda_i}, \beta_{\Lambda_i}) - \mathbb{E}[\phi_4(\tau_0 [\tilde{\mathbf{Z}}_0]_{\Lambda_i}, \beta_{\Lambda_i})] \Big) \Big \lvert   \geq \delta\e \Big). 
\end{align*}

\textbf{(f)} The concentration of $\lambda_0$ to $\hat{\lambda}_0$ follows from $\mathcal{H}_1 (b)$ applied to the function $\phi_h([h^1]_{\Lambda_i}, \beta_{\Lambda_i}) := f_0'([h^1]_{\Lambda_i}, \beta_{\Lambda_i})$, since $f_0'$ is assumed to be Lipschitz, hence PL(2).  

The only other result to prove is concentration for $\frac{1}{n}(h^1)^*q^1=\frac{1}{n}\sum_{i\in\Gamma}\hat{h}_i^{1}\hat{q}_i^{1}$. The function $\phi_5:\mathbb{R}^{2|\Lambda|}\to\mathbb{R}$ defined as $\phi_5(x,y)=\mc{S}(x)f_1(x,y)$ is PL(2) by Lemma \ref{lem:Lprods}. Notice that $\phi_5(\hat{h}^1_{\Lambda_i},\beta_{\Lambda_i})=\hat{h}_i^{1}\hat{q}_i^{1}$. Moreover, let the function $\tilde{f}_i:\mathbb{R}\to\mathbb{R}$ be defined as
$\tilde{f}_i(x):=\mathbb{E}_{[\tilde{\mathbf{Z}}_0]_{\Lambda_i\setminus \{i\}},\beta_{\Lambda_i}}[f_1( \mc{R}(x,[\tau_0\tilde{\mathbf{Z}}_0]_{\Lambda_i}),\beta_{\Lambda_i})]$, where the function $\mc{R}:\mathbb{R}^{1\times \Lambda}\to\mathbb{R}^\Lambda$ replaces the center coordinate of the second argument, which is in $\mathbb{R}^\Lambda$, with the first argument, which is in $\mathbb{R}$, a scalar. For example, $\tilde{f}_i([\tau_0\tilde{\mathbf{Z}}_0]_i) = \mathbb{E}_{[\tilde{\mathbf{Z}}_0]_{\Lambda_i\setminus \{i\}},\beta_{\Lambda_i}}[f_1( [\tau_0\tilde{\mathbf{Z}}_0]_{\Lambda_i},\beta_{\Lambda_i})]$.  Then we have
\begin{align*}
&\sum_{i\in\Gamma}\mathbb{E}[\phi_5(\tau_0[\tilde{\mathbf{Z}}_0]_{\Lambda_i},\beta_{\Lambda_i})]=\sum_{i\in\Gamma}\mathbb{E}[[\tau_0\tilde{\mathbf{Z}}_{0}]_i f_1(\tau_0 [\tilde{\mathbf{Z}}_0]_{\Lambda_i}, \beta_{\Lambda_i})] \\
&= \sum_{i\in\Gamma} \mathbb{E}_{[\tilde{\mathbf{Z}}_0]_i}[[\tau_0\tilde{\mathbf{Z}}_0]_i\tilde{f}_i([\tau_0\tilde{\mathbf{Z}}_0]_i)]\\
& \overset{(a)}{=} \sum_{i\in\Gamma} \mathbb{E}_{[\tilde{\mathbf{Z}}_0]_i}[([\tau_0\tilde{\mathbf{Z}}_{0}]_i)^2] \mathbb{E}_{[\tilde{\mathbf{Z}}_0]_i}[\tilde{f}_i'([\tau_0\tilde{\mathbf{Z}}_0]_i)]\\
&\overset{(b)}{=} \tau_0^2 \sum_{i\in\Gamma} \mathbb{E}_{[\tilde{\mathbf{Z}}_0]_{\Lambda_i},\beta_{\Lambda_i}}[f_1'([\tau_0\tilde{\mathbf{Z}}_0]_{\Lambda_i},\beta_{\Lambda_i})]\overset{(c)}{=} \delta |\Gamma|  \hat{\lambda}_{1} \breve{E}_{0,0} .
\end{align*}
In the above, step $(a)$ follows from Stein's Method, Lemma \ref{fact:stein}, step $(b)$ follows from the definition of $\tilde{\mathbf{Z}}_0$ in \eqref{eq:tildeZcov} and the definition of $f_1'$, which is the partial derivative w.r.t. the center coordinate of the first arguments, and step $(c)$ follows from the definition of $\hat{\lambda}_1$ in \eqref{eq:hatlambda_hatxi} and the definition of $\breve{E}_{0,0}$ in \eqref{eq:Edef}.  Therefore, using $\mc{H}_1 (b)$, we have the desired upper bound, since
\begin{align*}
&P\Big (\Big \lvert  \frac{1}{n}(h^1)^*q^1 - \hat{\lambda}_{1} \breve{E}_{0,0} \Big \lvert  \geq \e \Big)\\
&= P\Big( \Big \lvert \frac{1}{|\Gamma|}\sum_{i\in\Gamma} \Big ( \phi_5(\hat{h}^1_{\Lambda_i},\beta_{\Lambda_i}) - \mathbb{E}[ \phi_5(\tau_0 [\tilde{\mathbf{Z}}_0]_{\Lambda_i}, \beta_{\Lambda_i})]\Big) \Big \lvert  \geq \delta\e\Big). 
\end{align*}

\textbf{(g)}
Note that $\mathbf{Q}_{1} = \frac{1}{n}\norm{q^0}^2$ and $\tilde{C}^{1} = \tilde{E}_{0,0} = \sigma_0^2 > 0$.  By Lemma \ref{inverses} and \eqref{eq:qassumption},
\be
P\Big(\Big\lvert n\norm{q^0}^{-2}- \sigma_0^{-2}\Big \lvert \geq \e\Big) \leq 2 K e^{-\kappa n \e^2 \sigma_0^{2} \min(1, \sigma_0^{2})}.  
\label{eq:H1g1}
\ee
By the definitions in Section \ref{subsec:defs}, $\gamma^1_0 = \frac{1}{n}\mathbf{Q}_{1}^{-1} (q^0)^*q^1$ and $\hat{\gamma}_0^1= (\tilde{C}^{1})^{-1} \tilde{E}_1 = \tilde{E}_{0,1} \sigma_0^{-2}.$  Therefore, 
\begin{align*}
& P\Big( \lvert \gamma^1_0 - \hat{\gamma}^1_0 \lvert \geq \epsilon \Big) = P\Big(\Big \lvert  \frac{1}{n}\mathbf{Q}_{1}^{-1}(q^0)^*q^1 - \tilde{E}_{0,1} \sigma_0^{-2} \Big \lvert \geq \epsilon \Big) \\
&\overset{(a)}{\leq} P\Big( \lvert  \mathbf{Q}_{1}^{-1} - \sigma_0^{-2} \lvert \geq \tilde{\e} \Big) + P\Big(\Big \lvert  \frac{1}{n}(q^0)^*q^1 - \tilde{E}_{0,1}  \Big \lvert \geq \tilde{\e} \Big)\\
& \overset{(b)}{\leq} K e^{-\kappa  n \e^2} + K e^{-\kappa  n \e^2}.
\end{align*}
where $(a)$ follows from Lemma \ref{products} with
$\tilde{\e}:= \min\Big\{ \sqrt{\frac{\e}{3}}, \ \frac{\e}{3 \tilde{E}_{0,1} }, \ \frac{\e \sigma_0^2}{3} \Big\}$ 
and $(b)$ from \eqref{eq:H1g1} and $\mc{H}_1(e)$.

\textbf{(h)}
 From the definitions in Section \ref{subsec:defs},  we have  $\norm{q^1_{\perp}}^2 = \norm{q^1}^2 -\norm{q^1_{\parallel}}^2 =   \norm{q^1}^2- (\gamma_0^1)^2 \norm{q^0}^2$, and $(\sigma_{1}^{\perp})^2 = \tilde{E}_{1,1} - \tilde{E}^*_{1} (\tilde{C}^{1})^{-1}  \tilde{E}_{1} = \sigma_1^2 - (\tilde{E}_{0,1})^2 \tilde{E}_{0,0}^{-1} = \sigma_1^2 - (\hat{\gamma}^1_0)^2 \sigma_0^2$. We therefore have
\begin{align*}
&P\Big(\Big \lvert \frac{\norm{q^1_{\perp}}^2}{n} - (\sigma_{1}^{\perp})^2 \Big \lvert \geq \epsilon\Big) \leq P\Big(\Big \lvert \frac{\norm{q^1}^2}{n} - \sigma_1^2 \Big \lvert \geq \frac{\epsilon}{2} \Big)\\
& + P\Big(\Big \lvert (\gamma_0^1)^2 \frac{\norm{q^0}^2}{n} - (\hat{\gamma}^1_0)^2 \sigma_0^2\Big \lvert \geq \frac{\epsilon}{2}\Big) \leq K e^{-\kappa n \e^2} + K e^{-\kappa n \e^2}.
\end{align*}
where the last inequality is obtained using $\mc{H}_1(e)$ for bounding the first term and by applying Lemma \ref{products} to the second term along with the concentration of $\norm{q^0}$ in \eqref{eq:qassumption}, $\mc{H}_1(g)$, and Lemma \ref{powers} (for concentration of the square).


\subsection{Step 4: Showing that $\mc{H}_{t+1}$ holds}
The probability statements in the lemma and the other parts of $\mc{H}_{t+1}$ are conditioned on the event that the matrices $\mathbf{Q}_{1}, \ldots, \mathbf{Q}_{t+1}$ are invertible, but for the sake of brevity, we do not explicitly state the conditioning in the probabilities.  The following lemma will be used to prove $\mc{H}_{t+1}$.

\begin{lem}
\label{lem:Qv_conc}
Let $v := \frac{1}{n}B^*_{t+1} m_{\perp}^t - \frac{1}{n}Q_{t+1}^*(\xi_t q^t - \sum_{i=0}^{t-1} \alpha^t_i \xi_i q^i)$ and $\mathbf{Q}_{t+1} := \frac{1}{n}Q_{t+1}^* Q_{t+1}$.  Then for $j \in [t+1]$,
\ben
P\Big(\Big \lvert [\mathbf{Q}_{t+1}^{-1} v]_{j} \Big \lvert  \geq \e \Big) \leq e^{-\kappa n \e^2}.
\een
\end{lem}
\begin{proof} 
We can infer from the proof of \cite[Lemma 6]{RushV18} that the proof of Lemma \ref{lem:Qv_conc} involves induction hypotheses $\mathcal{H}_{1}(g)$-$\mathcal{H}_t(g)$, $\mathcal{H}_{1}(h)$-$\mathcal{H}_t(h)$, $\mathcal{H}_t(e)$, and $\mathcal{B}_{t}(c), \mathcal{B}_{t}(f), \mathcal{B}_{t}(g)$. Notice that the statements of these hypotheses have the same form as the corresponding statements in \cite{RushV18}. While the definition of the concentration constants in these hypotheses may be different for separable and non-separable denoisers, the proof uses the result that the quantities concentrate with desired rate rather than what the concentration constants are. Therefore, the proof for \cite[Lemma 12]{RushV18}, which is similar to the proof of \cite[Lemma 6]{RushV18}, is directly applicable here.
\end{proof}

We are ready to prove $\mc{H}_{t+1} (a)- \mc{H}_{t+1} (h)$.

\textbf{(a)} 
Recall the definition of $\Delta_{t+1,t}$ from Lemma \ref{lem:hb_cond} \eqref{eq:Dt1t}.  Using Lemma \ref{fact:gauss_p0}, $\frac{1}{\sqrt{n}} \norm{m^t_{\perp}} \mathsf{P}^{\parallel}_{Q_{t+1}} Z_t \overset{d}{=} \frac{1}{\sqrt{n}}\norm{m^t_{\perp}}  \frac{1}{\sqrt{|\Gamma|}}\hat{Q}_{t+1} \bar{Z}_{t+1},$ where columns of the matrix $\hat{Q}_{t+1} \in \mathbb{R}^{|\Gamma| \times (t+1)}$ form an orthogonal basis for the column space of $Q_{t+1}$, which are normalized such that $\hat{Q}_{t+1}^* \hat{Q}_{t+1} = |\Gamma|\mathsf{I}_{t+1}$, and $\bar{Z}_{t+1} \in \mathbb{R}^{t+1}$ is an independent random vector with i.i.d.\ $\mc{N}(0,1)$ entries.  We can then write
\begin{align*}
\Delta_{t+1,t} &\overset{d}{=}  \sum_{r=0}^{t-1} (\alpha^t_r - \hat{\alpha}^{t}_r) h^{r+1} + Z_t \Big(\frac{\norm{m^t_{\perp}}}{\sqrt{n}} - \tau_{t}^{\perp}\Big)\\
& - \frac{\norm{m^t_{\perp}}}{\sqrt{n}} \frac{\hat{Q}_{t+1} \bar{Z}_{t+1}}{\sqrt{|\Gamma|}} + Q_{t+1} \mathbf{Q}_{t+1}^{-1}v,
\end{align*}
where $ \mathbf{Q}_{t+1} \in \mathbb{R}^{(t+1) \times (t+1)}$ and $v \in \mathbb{R}^{t+1}$ are defined in Lemma \ref{lem:Qv_conc}. By Lemma \ref{lem:squaredsums},
\begin{align*}
&\norm{\Delta_{t+1,t}}^2  \leq (2t+3)\Big[\sum_{r=0}^{t-1} (\alpha^t_r - \hat{\alpha}^t_r)^2 \norm{h^{r+1}}^2\\
& +\norm{Z_t}^2 \Big(\frac{\norm{m^t_{\perp}}}{\sqrt{n}} - \tau_{t}^{\perp} \Big)^2 + \frac{\norm{m^t_{\perp}}^2}{n}\frac{||\hat{Q}_{t+1} \bar{Z}_{t+1}||^2 }{|\Gamma|}\\
& + \sum_{j=0}^{t} \norm{q^j}^2 [\textbf{Q}_{t+1}^{-1} v]_{j+1}^2\Big], 
\end{align*}
where we have used $Q_{t+1}\textbf{Q}_{t+1}^{-1}v = \sum_{j=0}^{t} q^j [\textbf{Q}_{t+1}^{-1} v]_{j+1}$. Applying Lemma, with $\e_t = \frac{\e}{(2t+3)^2}$,
\ref{sums},
\begin{align}
P\Big(\frac{\norm{\Delta_{t+1,t}}^2}{|\Gamma|} \geq \epsilon \Big) & \leq \sum_{r=0}^{t-1} P\Big( \lvert\alpha^t_r - \hat{\alpha}^t_r  \lvert \frac{||h^{r+1}||}{\sqrt{|\Gamma|}} \geq  \sqrt{\e_t} \Big)\nonumber\\
& +  P\Big(\Big \lvert \frac{\norm{m^t_{\perp}}}{\sqrt{n}} - \tau_{t}^{\perp}\Big \lvert \frac{\norm{Z_{t}}}{\sqrt{|\Gamma|}} \geq  \sqrt{\e_t} \Big)\nonumber \\
&  + P\Big(\frac{\norm{m^t_{\perp}}}{\sqrt{n}}  \frac{||\hat{Q}_{t+1} \bar{Z}_{t+1}||}{|\Gamma|} \geq  \sqrt{\e_t}\Big) \nonumber\\
& + \sum_{j=0}^{t} P\Big(\Big \lvert [\textbf{Q}_{t+1}^{-1} v]_{j+1} \Big \lvert \frac{\norm{q^j}}{\sqrt{n}} \geq  \sqrt{\e_t} \Big).
\label{eq:deltt1t_sq_conc}
\end{align}
 We now show each of the terms in \eqref{eq:deltt1t_sq_conc} has the desired upper bound.  For  $0 \leq r \leq t$,
\begin{align*}
& P\Big( \lvert\alpha^t_r - \hat{\alpha}^t_r\lvert \frac{||h^{r+1}||}{\sqrt{|\Gamma|}} \geq  \sqrt{\e_t} \Big) \\ 
&\leq P\Big(\lvert\alpha^t_r - \hat{\alpha}^t_r  \lvert \Big(  \Big\lvert\frac{||h^{r+1}||}{\sqrt{|\Gamma|}}  -\tau_r \Big \lvert+\tau_r\Big) \geq \sqrt{\e_t} \Big)  \\ 
& \overset{(a)}{\leq} P\Big( \lvert \alpha^t_r - \hat{\alpha}^t_r  \lvert \geq \frac{\sqrt{\e_t}}{2}  \min\{1 , \tau_r^{-1} \} \Big)\\
&+ P\Big(\Big \lvert \frac{||h^{r+1}||}{\sqrt{|\Gamma|}} - \tau_r \Big \lvert \geq \sqrt{\e_t}  \Big) \overset{(b)}{\leq} K e^{-\kappa n \e} + K e^{- \kappa n \e} ,
\end{align*}
where step $(a)$ follows from similar justification as that for step $(a)$ in \eqref{eq:T1_H1a} and step $(b)$ follows from induction hypotheses $\mathcal{B}_t (g)$, $\mathcal{H}_{1} (d)-\mathcal{H}_{t} (d)$, and Lemma \ref{sqroots}.  Next, the second term in \eqref{eq:deltt1t_sq_conc} is bounded as 
\ben
\begin{split}
 &P\Big(\Big \lvert \frac{\norm{m^t_{\perp}}}{\sqrt{n}} - \tau_{t}^{\perp} \Big \lvert \frac{\norm{Z_t}}{\sqrt{|\Gamma|}} \geq \sqrt{\tilde{\e}_t} \Big) \leq P\Big(\Big \lvert \frac{\norm{m^t_{\perp}}}{\sqrt{n}} - \tau_{t}^{\perp}\Big \lvert  \geq \e_t^{\frac{1}{4}} \Big) \\
 &+  P\Big( \frac{\norm{Z_t}}{\sqrt{|\Gamma|}} \geq \e_t^{\frac{1}{4}} \Big) \overset{(c)}{\leq}  Ke^{-\kappa  n \sqrt{\e}} + Ke^{- \kappa n \sqrt{\e}},
\end{split}
\een
where step $(c)$ is obtained using induction hypothesis $\mathcal{B}_t (h)$, Lemma \ref{subexp}, and Lemma \ref{sqroots}. Since $\frac{1}{\sqrt{n}} \norm{m^t_{\perp}}$ concentrates on $\tau_{t}^{\perp}$ by $\mc{B}_t (h)$, the third term in \eqref{eq:deltt1t_sq_conc} can be bounded as 
\begin{equation}
\begin{split}
&P\Big(\frac{\norm{m^t_{\perp}}}{\sqrt{n}} \cdot  \frac{||\hat{Q}_{t+1} \bar{Z}_{t+1}|| }{|\Gamma|}\geq  \sqrt{\e_t}  \Big)\\
& \leq P\Big( \Big( \Big \lvert \frac{\|m^t_\perp\|}{\sqrt{n}} - \tau_t^\perp \Big \lvert + \tau_t^\perp\Big)\cdot \frac{\|\hat{Q}_{t+1}\bar{Z}_{t+1}\|}{|\Gamma|}  \geq \sqrt{\e_t}\Big) \\
 &\leq P\Big(  \Big \lvert \frac{\norm{m^t_{\perp}}}{\sqrt{n}} - \tau_{t}^{\perp} \Big \lvert  \geq  \sqrt{\e} \Big)\\
& + P\Big(\frac{1}{|\Gamma|}||\hat{Q}_{t+1}\bar{Z}_{t+1}|| \geq 
 \frac{ \sqrt{\e_t}}{2}  \min\{ 1 , (\tau_t^{\perp})^{-1}\} \Big).
 \label{eq:T3split}
 \end{split}
\end{equation}
For the second term in \eqref{eq:T3split}, first  bound  the norm of $\hat{Q}_{t+1} \bar{Z}_{t+1}$ as follows. Letting $\hat{q}_i$ denote the $i^{th}$ column of $\hat{Q}_{t+1}$, we have
\ben
\begin{split}
\norm{\hat{Q}_{t+1} \bar{Z}_{t+1}}^2 =  \norm{\sum_{i=0}^{t} \hat{q}_i \bar{Z}_{t+1,i} }^2 & \overset{(d)}{\leq} (t+1)\sum_{i=0}^{t} \norm{\hat{q}_i}^2 (\bar{Z}_{t+1,i})^2\\
& \overset{(e)}{=} |\Gamma|(t+1) \sum_{i=0}^{t} (\bar{Z}_{t+1,i})^2,
\end{split}
\een
where step (d) follows from Lemma \ref{lem:squaredsums} and step (e) uses $\norm{\hat{q}_i}^2 = |\Gamma|$ for all $0 \leq i \leq t$.  Therefore,
\begin{align}
P\Big(\frac{1}{|\Gamma|^2} \norm{\hat{Q}_{t+1}\bar{Z}_{t+1}}^2 \geq \e' \Big) &\overset{(f)}{\leq} \sum_{i=0}^{t}P\Big(\abs{\bar{Z}_{t+1,i}}^2 \geq \frac{|\Gamma|\e'}{(t+1)^2} \Big)\nonumber\\
& \overset{(g)}{\leq} 2 e^{-\frac{1}{2(t+1)^2} |\Gamma| \e'}.\label{eq:Qt1Zt1conc}
\end{align}
Step $(f)$ is obtained from Lemma \ref{sums} and step $(g)$ from Lemma \ref{lem:normalconc}. Using \eqref{eq:Qt1Zt1conc}, the RHS of \eqref{eq:T3split} is bounded by $Ke^{-\kappa n\e }$.  Finally, for $0 \leq j \leq t$, the last term in \eqref{eq:deltt1t_sq_conc} can be bounded by 
\ben
\begin{split}
& P\Big(\Big \lvert [\textbf{Q}_{t+1}^{-1} v]_{j+1} \Big \lvert \frac{||q^j||}{\sqrt{n}} \geq \sqrt{\e_t} \Big)\\
&\leq P\Big(\Big \lvert [\textbf{Q}_{t+1}^{-1} v]_{j+1} \Big \lvert \Big( \Big \lvert \frac{||q^j||}{\sqrt{n}} - \sigma_j \Big \lvert + \sigma_j \Big) \geq \sqrt{\e_t} \Big) 
\\  & \leq  P\Big(\Big \lvert \frac{\norm{q^j}}{\sqrt{n}}- \sigma_j \Big \lvert \geq  \sqrt{\e} \Big)+ P\Big(\Big \lvert [\textbf{Q}_{t+1}^{-1} v ]_{j+1} \Big \lvert \geq \frac{\sqrt{\e_t}}{2  \max\{1, \sigma_j\}} \Big)\\
&\overset{(g)}\leq K e^{-\kappa n \e^2} + K e^{-\kappa n \e^2},
\end{split}
\een
where step $(g)$ follows from Lemma \ref{lem:Qv_conc}, the induction hypothesis $\mc{H}_{t}(e)$, and Lemma \ref{sqroots}. Thus we have bounded each term of \eqref{eq:deltt1t_sq_conc} as desired.

\textbf{(b)} For brevity, we use the notation $\mathbb{E}_{\phi_h}:=\frac{1}{|\Gamma|}\sum_{i\in\Gamma}\mathbb{E}[\phi_h(\tau_0[\tilde{\mathbf{Z}}_0]_{\Lambda_i},...,\tau_t[\tilde{\mathbf{Z}}_t]_{\Lambda_i},\beta_{\Lambda_i})]$ and
\begin{align}
\hat{a}_i &:= ( [\hat{h}^1_{\pure}]_i + [\hat{\Delta}_{1,0}]_i,...,[\hat{h}^{t+1}_{\pure}]_i+\sum_{r=0}^{t}\mathsf{d}_r^{t}[\hat{\Delta}_{r+1,r}]_i, \beta_i), \nonumber \\
\hat{c}_i &:= ([\hat{h}_{\pure}^1]_i,...,[\hat{h}_{\pure}^{t+1}]_i,\beta_i),
\label{eq:Htaici}
\end{align}
for $i\in \Gamma$.  Hence $\hat{a}$ and $\hat{c}$ are arrays in $\mathbb{R}^{\Gamma}$ with entries $\hat{a}_i$, $\hat{c}_i\in\mathbb{R}^{(t+2)}$. We note that by $\hat{a}_{\Lambda_i}$ we mean for the $p$-dimensional cube $\Lambda_i$ to be applied to each of the $(t+2)$ elements of $\hat{a}$ and we define $\norm{\hat{a}_{\Lambda_i}}^2 := \sum_{j \in \Lambda_i} \norm{\hat{a}_j}^2$. Moreover, define $\hat{\Delta}_{r+1,r}=\mc{V}^{-1}(\Delta_{r+1,r})$, hence $\hat{\Delta}_{r+1,r}\in\mathbb{R}^\Gamma$, for all $r=0,\ldots, t$.
Then, using the conditional distribution of $h^{t+1}$ from Lemma~\ref{lem:h_cond_pure} and Lemma~\ref{sums}, 
\begin{align}
&P\Big( \Big\vert \frac{1}{|\Gamma|}\sum_{i\in\Gamma}\phi_h(\hat{h}^1_{\Lambda_i},...,\hat{h}^{t+1}_{\Lambda_i},\beta_{\Lambda_i}) -\mathbb{E}_{\phi_h}\Big\vert \geq \epsilon  \Big) \nonumber\\
&\leq P\Big(\Big\vert \frac{1}{|\Gamma|}\sum_{i\in\Gamma} (\phi_h(\hat{a}_{\Lambda_i})-\phi_h(\hat{c}_{\Lambda_i})) \Big\vert \geq \frac{\epsilon}{2}\Big)\nonumber\\
&+ P\Big(\Big\vert \frac{1}{|\Gamma|}\sum_{i\in\Gamma}\phi_h(\hat{c}_{\Lambda_i})-\mathbb{E}_{\phi_h}\Big\vert \geq \frac{\epsilon}{2}\Big).
\label{eq:Htbi}
\end{align}
Label the terms of \eqref{eq:Htbi} as $T_1$ and $T_2$. We next show that both terms are bounded by
$Ke^{-\kappa n \epsilon^2}$. 

First consider term $T_1$. Let $d=|\Lambda|$. Notice that
\begin{align*}
&\|\hat{a} - \hat{c}\|^2= \sum_{i\in\Gamma} \Big[ [\hat{\Delta}_{1,0}]_i^2 + (\sum_{r=0}^1 \mathsf{d}^1_r [\hat{\Delta}_{r+1,r}]_{i})^2 +\ldots\\
&\hspace*{1.5in} +  (\sum_{r=0}^t \mathsf{d}^t_r [\hat{\Delta}_{r+1,r}]_i)^2\Big]\\
&\overset{(a)}{\leq}  \|\hat{\Delta}_{1,0}\|^2 + \sum_{r=0}^1 (\mathsf{d}^1_r)^2  \sum_{r'=0}^1\|\hat{\Delta}_{r'+1,r'} \|^2 +\ldots\\
&\hspace*{1.5in} + \sum_{r=0}^t( \mathsf{d}^t_r)^2 \sum_{r'=0}^t\|\hat{\Delta}_{r'+1,r'}\|^2\\
&\overset{(b)}{=} \sum_{r'=0}^t\|\hat{\Delta}_{r'+1,r'}\|^2 \sum_{k=r'}^t \sum_{r=0}^k (\mathsf{d}^k_r)^2.
\end{align*}
In the above, $(a)$ follows from Cauchy-Schwartz and  $(b)$ by collecting the terms in the sums. Hence,
\begin{align}
&\frac{\|\hat{a} - \hat{c}\|^2}{|\Gamma|} \leq \Big(\frac{ \|\hat{\Delta}_{1,0}\|}{\sqrt{|\Gamma|}} \sum_{k=0}^t \sum_{r=0}^k |\mathsf{d}^k_r| + \frac{\|\hat{\Delta}_{2,1}\|}{\sqrt{|\Gamma|}} \sum_{k=1}^t \sum_{r=0}^k |\mathsf{d}^k_r|\nonumber\\
& \hspace*{1.2in}+ \ldots + \frac{\|\hat{\Delta}_{t+1,t}\|}{\sqrt{|\Gamma||}} \sum_{r=0}^t |\mathsf{d}^t_r|\Big)^2.
\label{eq:Delta_tot_def}
\end{align}
Denote the RHS of \eqref{eq:Delta_tot_def} by $\Delta_\tot^2$, then using Lemma \ref{sums} and $\mc{H}_1(a) - \mc{H}_{t+1} (a)$, we have
\begin{equation}
P\Big(\frac{\|\hat{a} - \hat{c}\|}{\sqrt{|\Gamma|}}  \geq \e \Big) \leq P(\Delta_\tot \geq \e) \leq K e^{-\kappa |\Gamma| \e^2}.
\label{eq:bound_Del_tot}
\end{equation}
Now, using the pseudo-Lipschitz property of $\phi_h$, we have
\be
\begin{split}
T_1&\leq P\Big(\frac{1}{|\Gamma|}\sum_{i\in\Gamma} L (1 + \|\hat{a}_{\Lambda_i}\| + \|\hat{c}_{\Lambda_i}\|) \|[\hat{a}-\hat{c}]_{\Lambda_i}\| \geq \frac{\epsilon}{2}\Big)\\
&\overset{(a)}{\leq}P \Big( \Big( \frac{1}{|\Gamma |}\sum_{i\in\Gamma} ( 1+\|\hat{a}_{\Lambda_i}\|+\|\hat{c}_{\Lambda_i}\|)^2 \Big)^{1/2} \\
&\hspace*{1in}\cdot\Big( \frac{1}{|\Gamma |} \sum_{i\in\Gamma}\|[\hat{a}-\hat{c}]_{\Lambda_i}\|^2\Big)^{1/2}  \geq \frac{\e}{2L}\Big)\\
&\overset{(b)}{\leq} P\Big(\Big( 1 + \sqrt{2d}\frac{\|\hat{a}\|}{\sqrt{|\Gamma |}} + \sqrt{2d}\frac{\|\hat{c}\|}{\sqrt{|\Gamma |}}\Big)\Delta_\tot  \geq \frac{\e}{2\sqrt{6d}L} \Big)\\
&\overset{(c)}{\leq} P\Big(\Big( 1 + 2 \sqrt{2d}\frac{\|\hat{c}\|}{\sqrt{|\Gamma |}} + \sqrt{2d}\Delta_\tot \Big)
\Delta_\tot  \geq \frac{\e}{2\sqrt{6d}L} \Big).
\label{eq:T1_bound_1}
\end{split}
\ee
Above, step $(a)$ follows from Cauchy-Schwartz and step $(b)$ from an application of Lemma \ref{lem:squaredsums}:
\begin{align*}
\sum_{i\in\Gamma} ( 1+\|\hat{a}_{\Lambda_i}\|+\|\hat{c}_{\Lambda_i}\|)^2&\leq 3\sum_{i\in\Gamma}( 1+\|\hat{a}_{\Lambda_i}\|^2+\|\hat{c}_{\Lambda_i}\|^2) \\
&\leq 3(\sqrt{|\Gamma|} + \sqrt{2d} \|\hat{a}\| + \sqrt{2d}\|\hat{c}\|)^2,
\end{align*}
and $\sum_{i\in\Gamma} \|[\hat{a}-\hat{c}]_{\Lambda_i}\|^2 \leq 2d \|\hat{a}-\hat{c}\|^2$ along with \eqref{eq:Delta_tot_def}.
Step $(c)$ follows from $\|\hat{a}\| \leq \|\hat{a}-\hat{c}\|+\|\hat{c}\| \leq \sqrt{|\Gamma|} \Delta_\tot + \|\hat{c}\|$.  Notice that 
\begin{equation*}
\|\hat{c}\|^2 =  \sum_{r=0}^t \|\hat{h}_\pure^{r+1}\|^2 + \|\beta\|^2 \overset{d}{=} \sum_{r=0}^t \tau_r^2\|\tilde{\mbf{Z}}_{r}\|^2 + \|\beta\|^2,
\end{equation*}
where the last step follows from Lemma \ref{lem:h_cond_pure}.
Define $\mathbb{E}_{c}:=\sum_{r=0}^t \tau_r^2 + \sigma_{\beta}^2$. Then
\begin{align}
&P\Big( \Big\lvert\frac{\|\hat{c}\|^2}{|\Gamma|}  -  \mathbb{E}_{c} \Big\lvert \geq \e \Big) \leq \sum_{r=0}^t P\Big(\Big\lvert\frac{\|\tilde{\mbf{Z}}_{r}\|^2}{|\Gamma|}- 1\Big\lvert \geq \frac{\e}{(t+2)\tau_r^2}\Big)\nonumber\\
&\qquad\qquad + P\Big(\Big\lvert\frac{\|\beta\|^2}{|\Gamma|}-\sigma_{\beta}^2\Big\lvert \geq \frac{\e}{t+2}\Big) \leq Ke^{-\kappa |\Gamma| \e^2},\label{eq:bound_hatc}
\end{align}
where the last step follows from Lemma \ref{subexp} and \eqref{eq:beta0_assumption2}.
Therefore, using the bound in \eqref{eq:T1_bound_1},
\begin{align*}
T_1 &\leq  P\Big(\Big( 1 + 2 \sqrt{2d}\Big(\frac{\|\hat{c}\|}{\sqrt{|\Gamma |}}-\mathbb{E}_c^{1/2}\Big)\\
&\qquad\qquad + 2 \sqrt{2d}\mathbb{E}_c^{1/2} + \sqrt{2d}\Delta_\tot \Big)\Delta_\tot  \geq \frac{\e}{2\sqrt{6d}L} \Big)\\
&\leq P\Big(\Big\lvert\frac{\|\hat{c}\|}{\sqrt{|\Gamma |}}-\mathbb{E}_c^{1/2}\Big\lvert \geq \frac{\e}{\sqrt{2d}} \Big) \\
&\qquad\quad+ P\Big( \Delta_\tot \geq \frac{\frac{\e}{\sqrt{2d}}\min\{1,\frac{1}{2\sqrt{3}\, L}\}}{4+2\sqrt{2d} \,\mathbb{E}_c^{1/2}} \Big) \leq K e^{-\kappa |\Gamma| \e^2},
\end{align*}
where the last step follows from \eqref{eq:bound_Del_tot} and \eqref{eq:bound_hatc}.

Next, consider term $T_2$ of \eqref{eq:Htbi}. 
\begin{align*}
&P\Big(\Big\vert \frac{1}{|\Gamma|}\sum_{i\in\Gamma} \Big[\phi_h([h_{\textsf{pure}}^1]_{\Lambda_i},...,[h_{\textsf{pure}}^{t+1}]_{\Lambda_i},\beta_{\Lambda_i})\\
&\qquad\qquad- \mathbb{E}[\phi_h(\tau_0[\tilde{\mathbf{Z}}_0]_{\Lambda_i},...,\tau_t[\tilde{\mathbf{Z}}_t]_{\Lambda_i},\beta_{\Lambda_i})] \Big]\Big\vert \geq \frac{\epsilon}{2}\Big)\\
&\leq P\Big(\Big\vert \frac{1}{|\Gamma|}\sum_{i\in\Gamma} \Big[\phi_h([h_{\textsf{pure}}^1]_{\Lambda_i},...,[h_{\textsf{pure}}^{t+1}]_{\Lambda_i},\beta_{\Lambda_i})\\
&\qquad\qquad-\mathbb{E}_{\tilde{\mbf{Z}}_1,\ldots,\tilde{\mbf{Z}}_t}[\phi_h(\tau_0[\tilde{\mathbf{Z}}_0]_{\Lambda_i},...,\tau_t[\tilde{\mathbf{Z}}_t]_{\Lambda_i},\beta_{\Lambda_i})] \Big]\Big\vert \geq \frac{\epsilon}{4}\Big)\\
&+ P\Big(\Big\vert \frac{1}{|\Gamma|}\sum_{i\in\Gamma} \Big[\mathbb{E}_{\tilde{\mbf{Z}}_1,\ldots,\tilde{\mbf{Z}}_t}[\phi_h(\tau_0[\tilde{\mathbf{Z}}_0]_{\Lambda_i},...,\tau_t[\tilde{\mathbf{Z}}_t]_{\Lambda_i},\beta_{\Lambda_i})]\\
&\qquad\qquad-\mathbb{E}[\phi_h(\tau_0[\tilde{\mathbf{Z}}_0]_{\Lambda_i},...,\tau_t[\tilde{\mathbf{Z}}_t]_{\Lambda_i},\beta_{\Lambda_i})] \Big]\Big\vert \geq \frac{\epsilon}{4}\Big).
\end{align*}
Label the two terms on the RHS as $T_{2a}$ and $T_{2b}$. $T_{2a}$ can be bounded in a similar way as $T_2$ in \eqref{eq:phih_0} and $T_{2b}$ has the desired bound by Lemma \ref{lem:PL_MRF_conc}, since the function $\tilde{\phi}_h:\mathbb{R}^\Lambda\to\mathbb{R}$ defined as 
\begin{equation*}
\begin{split}
\tilde{\phi}(s)&:=\mathbb{E}_{\tilde{\mbf{Z}}_1,\ldots,\tilde{\mbf{Z}}_t}[\phi_h(\tau_0[\tilde{\mathbf{Z}}_0]_{\Lambda_i},...,\tau_t[\tilde{\mathbf{Z}}_t]_{\Lambda_i},s)]\\
&=\mathbb{E}_{\tilde{\mbf{Z}}_1,\ldots,\tilde{\mbf{Z}}_t}[\phi_h(\tau_0\mc{T}_i([\tilde{\mathbf{Z}}_0]_{\Lambda_i\cap\Gamma}),...,\tau_t\mc{T}_i([\tilde{\mathbf{Z}}_t]_{\Lambda_i\cap\Gamma}),s)]
\end{split}
\end{equation*}
is PL(2) by Lemmas \ref{lem:expZ} and \ref{lem:PLwithAvg}.

\textbf{(c)} We first show the concentration of $\frac{1}{n}(h^{t+1})^*\mc{V}(\beta)=\frac{1}{n}\sum_{i\in\Gamma}\hat{h}^{t+1}_i\beta_i$. Using the PL(2) function $\phi_1$ defined in $\mc{H}_1 (c)$, we have that $\phi_1(\hat{h}^{t+1}_{\Lambda_i},\beta_{\Lambda_i})=\hat{h}^{t+1}_i \beta_i$ and $\mathbb{E}[\phi_1(\tau_t[\tilde{\mathbf{Z}}_t]_{\Lambda_i},\beta_{\Lambda_i})]=\mathbb{E}[[\tau_t\tilde{\mathbf{Z}}_t]_i]\mathbb{E}[\beta_i]=0$ for all $i\in\Gamma$, since $[\tau_t\tilde{\mathbf{Z}}_t]_i$ has zero-valued mean and is independent of $\beta_i$. Therefore, $\mc{H}_{t+1} (b)$ gives the desired upper bound, since
\begin{align*}
&P\Big(\Big\vert \frac{(h^{t+1})^*\mc{V}(\beta)}{n}\Big\vert \geq \epsilon \Big)\\
&=P\Big(\Big\vert \frac{1}{|\Gamma|}\sum_{i\in\Gamma} \Big(\phi_1(\hat{h}^{t+1}_{\Lambda_i},\beta_{\Lambda_i}) -\mathbb{E}[\phi_1(\tau_t[\tilde{\mathbf{Z}}_t]_{\Lambda_i},\beta_{\Lambda_i})]\Big) \Big\vert \geq \delta\epsilon \Big).  
\end{align*}

We now show the concentration of $\frac{1}{n}(h^{t+1})^*q^0=\frac{1}{n}\sum_{i\in\Gamma} \hat{h}^{t+1}_i\hat{q}^0_i$. Using the PL(2) function $\phi_2$ defined in $\mc{H}_1 (c)$, we have that $\phi_2(\hat{h}^{t+1}_{\Lambda_i},\beta_{\Lambda_i})=\hat{h}^{t+1}_i \hat{q}^0_i$ and $\mathbb{E}[\phi_2(\tau_t[\tilde{\mathbf{Z}}_t]_{\Lambda_i},\beta_{\Lambda_i})]=\mathbb{E}[\tau_t[\tilde{\mathbf{Z}}_t]_i]\mathbb{E}[f_0(\mathbf{0},\beta_{\Lambda_i})]=0$, since $[\tilde{\mathbf{Z}}_t]_i$ has zero-valued mean and is independent of $\beta_{\Lambda_i}$ for all $i\in\Gamma$. Therefore, using $\mc{H}_{t+1} (b)$, we have the desired upper bound, since
\begin{align*}
&P\Big(\Big\vert \frac{(h^{t+1})^*q^0}{n}\Big\vert \geq \epsilon \Big)\\
& = P\Big(\Big\vert \frac{1}{|\Gamma|}\sum_{i\in\Gamma} \Big(\phi_2(\hat{h}^{t+1}_{\Lambda_i},\beta_{\Lambda_i})-\mathbb{E}[\phi_2(\tau_t[\tilde{\mathbf{Z}}_t]_{\Lambda_i},\beta_{\Lambda_i})] \Big)\Big\vert \geq \delta\e\Big). 
\end{align*}

\textbf{(d)} Let a function $\tilde{\phi}_3:\mathbb{R}^{3|\Lambda|}\to\mathbb{R}$ be defined as $\tilde{\phi}_3(x,y,z)=\mc{S}(x)\mc{S}(y)$. Since the operator $\mc{S}$ defined in \eqref{eq:operator_S} is Lipschitz, $\tilde{\phi}_3$ is PL(2) by Lemma \ref{lem:Lprods}. Note that 
$\tilde{\phi}_3([\hat{h}^{r+1}]_{\Lambda_i},\hat{h}^{t+1}_{\Lambda_i},\beta_{\Lambda_i})=\hat{h}^{r+1}_i\hat{h}^{t+1}_i$ and $\mathbb{E}[\tilde{\phi}_3(\tau_r[\tilde{\mathbf{Z}}_r]_{\Lambda_i},\tau_t[\tilde{\mathbf{Z}}_t]_{\Lambda_i},\beta_{\Lambda_i})]=\tau_r\tau_t\mathbb{E}[[\tilde{\mathbf{Z}}_r]_i[\tilde{\mathbf{Z}}_t]_i]=\breve{E}_{r,t}$, where the last equality follows from the definition in \eqref{eq:tildeZcov}. Therefore, the result follows from $\mc{H}_{t+1} (b)$, since
\begin{align*}
&P\Big(\Big\vert \frac{1}{|\Gamma|} (h^{r+1})^* h^{t+1}- \breve{E}_{r,t} \Big\vert \geq \epsilon \Big)\\ 
&=P\Big(\Big\vert \frac{1}{|\Gamma|}\sum_{i\in\Gamma}\tilde{\phi}_3([\hat{h}^{r+1}]_{\Lambda_i},\hat{h}^{t+1}_{\Lambda_i},\beta_{\Lambda_i})\\
&\hspace*{1in} - \mathbb{E}[\tilde{\phi}_3(\tau_r [\tilde{\mathbf{Z}}_r]_{\Lambda_i},\tau_t[\tilde{\mathbf{Z}}_t]_{\Lambda_i},\beta_{\Lambda_i}] \Big\vert \geq \epsilon \Big). 
\end{align*}

\textbf{(e)} We will show the concentration of $\frac{1}{n}(q^0)^*q^{t+1}=\frac{1}{n}\sum_{i\in\Gamma}\hat{q}_i^0\hat{q}_i^{t+1}$; the concentration of $\frac{1}{n}(q^{r+1})^*q^{t+1}$ follows similarly. The function $\tilde{\phi}_4(x,y): \mathbb{R}^{2|\Lambda|} \rightarrow \mathbb{R}$ defined as  $\tilde{\phi}_4(x,y):=f_0(\mathbf{0},y)f_{t+1}(x,y)$ is PL(2) by Lemma \ref{lem:Lprods} and $\tilde{\phi_4}(\hat{h}^{t+1}_{\Lambda_i},\beta_{\Lambda_i})=\hat{q}_i^0 \hat{q}_i^{t+1}$.  Moreover,
\begin{equation*}
\begin{split}
&\sum_{i\in\Gamma}\mathbb{E}[\tilde{\phi_4}(\tau_t[\tilde{\mathbf{Z}}_t]_{\Lambda_i},\beta_{\Lambda_i})]\\
&=\sum_{i\in\Gamma}\mathbb{E}[f_0(\mathbf{0},\beta_{\Lambda_i})f_{t+1}(\tau_t [\tilde{\mathbf{Z}}_t]_{\Lambda_i},\beta_{\Lambda_i})]=\delta |\Gamma| \tilde{E}_{0,t+1},
\end{split}
\end{equation*}
by definition \eqref{eq:Edef}. Therefore, using $\mc{H}_{t+1} (b)$, we have the desired result, since
\begin{equation*}
\begin{split}
&P\Big(\Big\vert \frac{(q^0)^*q^{t+1}}{n} - \tilde{E}_{0,t+1} \Big\vert \geq \epsilon\Big)\\
&=P\Big(\Big\vert \frac{1}{|\Gamma|}\sum_{i\in\Gamma}\Big[\tilde{\phi}_4(\hat{h}^{t+1}_{\Lambda_i},\beta_{\Lambda_i}) - \mathbb{E}[\tilde{\phi}_4(\tau_t [\tilde{\mathbf{Z}}_t]_{\Lambda_i},\beta_{\Lambda_i})]\Big] \Big\vert \geq \delta\epsilon\Big). 
\end{split}
\end{equation*}

\textbf{(f)} The concentration of $\lambda_t$ around $\hat{\lambda}_t$ follows $\mathcal{H}_{t+1}(b)$ applied to
the function $\phi_h(h^{t+1}_{\Lambda_i},\beta_{\Lambda_i}):=f_{t+1}'(h^{t+1}_{\Lambda_i},\beta_{\Lambda_i})$, since $f_{t+1}'$ is assumed to be Lipschitz, hence PL(2).
Next, we show concentration for $\frac{1}{n}(h^{t+1})^*q^{r+1}=\frac{1}{n}\sum_{i\in\Gamma} \hat{h}^{t+1}_i \hat{q}^{r+1}_i$. Let $\tilde{\phi}_5:\mathbb{R}^{3|\Lambda|}\to\mathbb{R}$ be defined as $\tilde{\phi}_5(x,y,z):=\mc{S}(y)f_{r+1}(x,z)$, which is PL(2) by Lemma \ref{lem:Lprods}. Note, $\tilde{\phi}_5([\hat{h}^{r+1}]_{\Lambda_i},\hat{h}^{t+1}_{\Lambda_i},\beta_{\Lambda_i})=\hat{h}^{t+1}_i \hat{q}^{r+1}_i$ and
\begin{align}
&\sum_{i\in\Gamma}\mathbb{E}[\tilde{\phi}_5(\tau_r[\tilde{\mathbf{Z}}_r]_{\Lambda_i},\tau_t [\tilde{\mathbf{Z}}_t]_{\Lambda_i},\beta_{\Lambda_i})]\nonumber\\
&= \sum_{i\in\Gamma} \mathbb{E}[[\tau_t\tilde{\mathbf{Z}}_t]_i f_{r+1}(\tau_r[\tilde{\mathbf{Z}}_r]_{\Lambda_i},\beta_{\Lambda_i})]= |\Gamma| \hat{\lambda}_{r+1} \breve{E}_{r,t},
\end{align}
where the last equality follows using Stein's Method, Lemma \ref{fact:stein}, as in $\mc{H}_1 (f)$. Therefore, $\mc{H}_{t+1} (b)$ gives the desired result, since
\begin{align*}
&P\Big(\Big\vert \frac{(h^{t+1})^*q^{r+1}}{n} - \hat{\lambda}_{r+1}\breve{E}_{r,t} \Big\vert\geq \epsilon\Big)\\
& =P\Big( \Big\vert \frac{1}{|\Gamma|}\sum_{i\in\Gamma} \Big( \tilde{\phi}_5(\hat{h}^{r+1}_{\Lambda_i},\hat{h}^{t+1}_{\Lambda_i},\beta_{\Lambda_i})\\
&\hspace*{1in} - \mathbb{E}[\tilde{\phi}_5(\tau_r[\tilde{\mathbf{Z}}_r]_{\Lambda_i},\tau_t[\tilde{\mathbf{Z}}_t]_{\Lambda_i},\beta_{\Lambda_i})]\Big)\Big\vert \geq \delta\e\Big). 
\end{align*}

\textbf{(g)} 
We can represent $\textbf{Q}_{t+1}$ as follows.
\ben
\mathbf{Q}_{t+1} = \frac{1}{n}\Big( \begin{array}{cc}
n\mathbf{Q}_{t} & Q_{t}^* q^{t} \\
 (Q_{t}^* q^{t})^* & \norm{q^{t}}^2 \end{array} \Big),
 \een
 Then, using $\frac{1}{n}\mathbf{Q}_{t}^{-1} Q_{t}^* q^{t} = \gamma^{t}$ and $(Q_{t}^* q^{t})^* \gamma^{t} = (q^{t})^* q^{t}_{\parallel}$, it follows by the block inversion formula that
\be
 \mathbf{Q}_{t+1}^{-1} = \Big( \begin{array}{cc}
\mathbf{Q}_{t}^{-1} + n \norm{q^{t}_{\perp}}^{-2} \gamma^{t}(\gamma^{t})^* & -n \norm{q^{t}_{\perp}}^{-2}  \gamma^{t} \\
-n \norm{q^{t}_{\perp}}^{-2}  (\gamma^{t})^* & n \norm{q^{t}_{\perp}}^{-2}  \end{array} \Big).
\label{eq:Qt1_inverse}
\ee
Using definitions \eqref{eq:hatalph_hatgam_def} and \eqref{eq:sigperp_defs}, block inversion can be similarly used to invert $\tilde{C}^{t+1}$:
\begin{align}
& (\tilde{C}^{t+1})^{-1}= \Big( \begin{array}{cc}
\tilde{C}^t & \tilde{E}_t \nonumber\\
\tilde{E}_t^* & \sigma_{t}^2  \end{array} \Big)^{-1}\\
 & = \Big( \begin{array}{cc}
(\tilde{C}^{t})^{-1} + (\sigma_t^{\perp})^{-2} \hat{\gamma}^{t}(\hat{\gamma}^{t})^* & -(\sigma_t^{\perp})^{-2}  \hat{\gamma}^{t} \\
-(\sigma_t^{\perp})^{-2}  (\hat{\gamma}^{t})^* & (\sigma_{t}^{\perp})^{-2}  \end{array} \Big).
\label{eq:Qt1_inverse_conc}
\end{align}
 In what follows, we show concentration for each of the elements in \eqref{eq:Qt1_inverse} to the corresponding elements in \eqref{eq:Qt1_inverse_conc}.

First, $n\norm{q^{t}_{\perp}}^{-2}$ concentrates to $(\sigma_{t}^{\perp})^{-2}$ at rate $K e^{-\kappa n \e^2}$ by $\mc{H}_{t}(h)$ and Lemma \ref{inverses}.  Next, consider the $i^{th}$ element of $-n \norm{q^{t}_{\perp}}^{-2} \gamma^{t}$. For $i \in [t]$, using  Lemma \ref{products} and $\mc{H}_{t} (g), (h)$ as discussed in the previous paragraph,
\begin{align}
P &\Big(\Big \lvert n \norm{q^{t}_{\perp}}^{-2} \gamma^{t}_{i-1} - (\sigma_{t}^{\perp})^{-2} \hat{\gamma}^{t}_{i-1} \Big \lvert \geq \e\Big) \leq K e^{-\kappa n \e^2}. \label{eq:matlem_bound1}
\end{align}
Consider element $(i,j)$ of $\mathbf{Q}_{t}^{-1} + n \norm{q^{t}_{\perp}}^{-2} \gamma^{t}(\gamma^{t})^*$ for $i,j \in [t]$.
\begin{align*}
P &\Big(\Big \lvert [\mathbf{Q}_{t}^{-1}]_{i,j} + n \norm{q^{t}_{\perp}}^{-2} \gamma^{t}_{i-1} \gamma^{t}_{j-1} - [(\tilde{C}^{t})^{-1}]_{i,j}\\
&\hspace*{2in}-(\sigma_{t}^{\perp})^{-2} \hat{\gamma}^{t}_{i-1} \hat{\gamma}^{t}_{j-1} \Big \lvert \geq \e\Big) \\
&\overset{(a)}{\leq} P \Big(\Big \lvert [\mathbf{Q}_{t}^{-1}]_{i,j}  - [(\tilde{C}^{t})^{-1}]_{i,j} \Big \lvert \geq \frac{\e}{2}\Big)  + P \Big( \lvert \gamma^{t}_{j-1} -  \hat{\gamma}^{t}_{j-1}  \lvert \geq \frac{\e'}{2}\Big) \\
&\quad  + P \Big(\Big \lvert n \norm{q^{t}_{\perp}}^{-2} \gamma^{t}_{i-1} - (\sigma_{t}^{\perp})^{-2} \hat{\gamma}^{t}_{i-1}  \Big \lvert \geq \frac{\e'}{2}\Big)  \\
&\overset{(b)}{\leq} K e^{- \kappa n \e^2} + K e^{-\kappa n \e^2} + K e^{-\kappa n \e^2}.
\end{align*}
Step $(a)$ follows from Lemma \ref{sums} and Lemma \ref{products} with
$
\e' = \min \{\sqrt{\frac{\e}{3}}, \frac{\e (\sigma_{t}^{\perp})^{2}}{3 \hat{\gamma}^{t}_{i-1}}, \frac{\e}{3  \hat{\gamma}^{t}_{j-1}}\}.
$
Step $(b)$ follows from the inductive hypothesis $\mc{H}_t(g)$, together with \eqref{eq:matlem_bound1}.

We now prove $\gamma^{t+1} \doteq \hat{\gamma}^{t+1}$. Recall, $\gamma^{t+1}= \frac{1}{n}\mathbf{Q}_{t+1}^{-1}Q_{t+1}^*q^{t+1}$ where $\mathbf{Q}_{t+1}:=\frac{1}{n}Q_{t+1}^*Q_{t+1}$. 
Thus, $\gamma_{r-1}^{t+1}=\frac{1}{n} \sum_{i=1}^{t+1} [\mathbf{Q}_{t+1}^{-1}]_{r,i}(q^{i-1})^*q^{t+1}$, for $1 \leq r \leq t+1$.
Then by the definition of $\hat{\gamma}^{t+1}$, for $1 \leq r \leq t+1$,
\begin{align*}
&P\Big(\vert \gamma^{t+1}_{r-1} - \hat{\gamma}^{t+1}_{r-1}\vert\geq \epsilon\Big)= P\Big(\Big\vert \sum_{i=1}^{t+1}\Big(\frac{1}{n} [\mathbf{Q}_{t+1}^{-1}]_{r,i} (q^{i-1})^*q^{t+1} \\
&\qquad\qquad-[(\tilde{C}^{t+1})^{-1}]_{r,i}\tilde{E}_{i-1,t+1} \Big)\Big\vert \geq \e \Big)\\
&\overset{(a)}{\leq} \sum_{i=1}^{t+1} P\Big(\Big\vert \frac{1}{n} [\mathbf{Q}_{t+1}^{-1}]_{r,i}(q^{i-1})^*q^{t+1} \\
&\qquad\qquad-[(\tilde{C}^{t+1})^{-1}]_{r,i}\tilde{E}_{i-1,t+1} \Big\vert \geq \frac{\epsilon}{t+1}\Big)\\
&\overset{(b)}{\leq} \sum_{i=1}^{t+1} P \Big(\Big\vert  \frac{1}{n} (q^{i-1})^*q^{t+1} - \tilde{E}_{i-1,t+1} \Big\vert \geq \tilde{\epsilon}_i \Big)\\
& + P\Big(\Big\vert [\mathbf{Q}_{t+1}^{-1}]_{r,i} -   [(\tilde{C}^{t+1})^{-1}]_{r,i} \Big\vert \geq \tilde{\epsilon}_i \Big)\\
&\overset{(c)}{\leq} K e^{-\kappa n \epsilon^2} + K e^{-\kappa n \epsilon^2}.
\end{align*}
Step (a) follows from Lemma \ref{sums} and step (b) from Lemma~\ref{products}, with $\tilde{\epsilon}_i:=\min\{\sqrt{\frac{\epsilon}{3(t+1)}},\frac{\epsilon}{3(t+1)\tilde{E}_{i-1,t+1}},\frac{\epsilon}{3(t+1)[(\tilde{C}^{t+1})^{-1}]_{r,i}}\}$.  Step (c) uses $\mathcal{H}_{t+1}(e)$ and what we have just demonstrated in the previous paragraphs.

\textbf{(h)}  
First, note that 
$||q^{t+1}_{\perp}||^2 = ||q^{t+1}||^2 - ||q^{t+1}_{\parallel}||^2 = ||q^{t+1}||^2 - ||Q_{t+1} \gamma^{t+1}||^2$. 
Using the definition of $\sigma_{t+1}^{\perp}$ in \eqref{eq:sigperp_defs}, we then have
\begin{align}
& P\Big(\Big \lvert \frac{\norm{q^{t+1}_{\perp}}^2}{n} - (\sigma_{t+1}^{\perp})^2 \Big \lvert \geq \epsilon\Big) \nonumber\\
&= P\Big(\Big \lvert \frac{\norm{q^{t+1}}^2}{n} - \frac{\norm{Q_{t+1} \gamma^{t+1}}^2}{n}- \sigma_{t+1}^2 + (\hat{\gamma}^{t+1})^* \tilde{E}_{t+1}\Big \lvert \geq \epsilon\Big) \nonumber \\
&\leq P\Big(\Big \lvert \frac{1}{n}\norm{q^{t+1}}^2 - \sigma_{t+1}^2 \Big \lvert \geq \frac{\epsilon}{2} \Big)\nonumber\\
& \qquad+ P\Big(\Big \lvert \frac{1}{n}\norm{Q_{t+1} \gamma^{t+1}}^2 - (\hat{\gamma}^{t+1})^* \tilde{E}_{t+1} \Big \lvert \geq \frac{\epsilon}{2}\Big).
\label{eq:Hgt1}
\end{align}
By $\mathcal{H}_{t+1} (e)$, the first term on the LHS of  \eqref{eq:Hgt1} is  bounded by $Ke^{-\kappa n\epsilon^2}$. For the second term, using $\gamma^{t+1} =  \frac{1}{n} \textbf{Q}_{t+1}^{-1}Q_{t+1}^* q^{t+1}$,
\ben
\begin{split}
&\norm{Q_{t+1} \gamma^{t+1}}^2 = n (\gamma^{t+1})^* \textbf{Q}_{t+1} \gamma^{t+1}\\
& = (\gamma^{t+1})^*\textbf{Q}_{t+1} \textbf{Q}_{t+1}^{-1} {Q_{t+1}^* q^{t+1}} = (\gamma^{t+1})^* {Q_{t+1}^* q^{t+1}}\\
& = \sum_{i=0}^{t} \gamma^{t+1}_i (q^i)^* q^{t+1}.
\end{split}
\een 
Hence 
\begin{align*}
& P\Big(\Big \lvert \frac{1}{n}\norm{Q_{t+1} \gamma^{t+1}}^2 - (\hat{\gamma}^{t+1})^* \tilde{E}_{t+1} \Big \lvert \geq \frac{1}{2} \epsilon \Big)\\
& = P\Big(\Big \lvert \sum_{i=0}^{t} \Big(\frac{1}{n}\gamma^{t+1}_i (q^i)^* q^{t+1} - \hat{\gamma}^{t+1}_i \tilde{E}_{i,t+1} \Big) \Big \lvert \geq \frac{1}{2}\epsilon\Big) \\
& \leq \sum_{i=0}^{t} P\Big( \lvert \frac{1}{n} \gamma^{t+1}_i (q^i)^* q^{t+1} - \hat{\gamma}^{t+1}_i \tilde{E}_{i,t+1} \lvert \geq \frac{1}{2} \epsilon (t+1)^{-1}\Big) \\
&\overset{(a)}{\leq} \sum_{i=0}^{t} P( \lvert \gamma^{t+1}_i - \hat{\gamma}^{t+1}_i \lvert \geq \tilde{\epsilon}_i )+ P\Big( \lvert \frac{(q^i)^* q^{t+1}}{n} - \tilde{E}_{i,t+1}  \lvert \geq \tilde{\epsilon}_i \Big) \\
&\overset{(b)}{\leq} K e^{-\kappa n \epsilon^2} + K e^{-\kappa n \epsilon^2}.
\end{align*}
Step (a) follows from the concentration of products, Lemma \ref{products}, using
$\tilde{\epsilon}_i := \min \Big \{\sqrt{\frac{\epsilon}{6(t+1)}} , \frac{\epsilon}{6(t+1) \tilde{E}_{i,t+1}} , \frac{\epsilon}{6(t+1) \hat{\gamma}^{t+1}_i}  \Big\}$,
and step (b) using $\mathcal{H}_{t+1} (e)$ and $\mathcal{H}_{t+1} (g)$.

\appendices

\section{Concentration Lemmas}
\label{app:conc_lemma}
In the following, $\e >0$ is assumed to be a generic constant, with additional conditions specified whenever needed.  The proof of the Lemmas in this section can be found in \cite{RushV18}.

\begin{applem}
\label{sums}
(Concentration of Sums.)
If random variables $X_1, \ldots, X_M$ satisfy $P(\abs{X_i} \geq \e) \leq e^{-n\kappa_i \e^2}$ for $1 \leq i \leq M$, then 
\ben
P\Big( \lvert \sum_{i=1}^M X_i  \lvert \geq \e\Big) \leq \sum_{i=1}^M P\Big(|X_i| \geq \frac{\e}{M}\Big) \leq M e^{-n (\min_i \kappa_i) \e^2/M^2}.
\een
\end{applem}

\begin{applem}[Concentration of Products]
\label{products} 
For random  variables $X,Y$ and non-zero constants $c_X, c_Y $, if
\ben
\begin{split}
&P\Big( | X- c_X |  \geq \e \Big) \leq K e^{-\kappa n \e^2},\\
&P\Big( | Y- c_Y |  \geq \e \Big) \leq K e^{-\kappa n \e^2},
\end{split}
\een
then the probability  $P\Big( | XY - c_Xc_Y|  \geq \e \Big)$ is bounded by 
\begin{align*}  
&  P\Big( | X- c_X |  \geq \min\Big( \sqrt{\frac{\e}{3}}, \frac{\e}{3 c_Y} \Big) \Big)\\
&  +  P\Big(  | Y- c_Y |  \geq \min\Big( \sqrt{\frac{\e}{3}}, \frac{\e}{3 c_X} \Big) \Big) \\
&\leq 2K \exp\Big\{\frac{-\kappa n \e^2}{9\max(1, c_X^2, c_Y^2)}\Big\}.
\end{align*}
\end{applem}


\begin{applem}
\label{sqroots}
(Concentration of Square Roots.)
Let $c \neq 0$. If
\ben
P\Big( \lvert X_n^2 - c^2 \lvert \geq \epsilon \Big) \leq e^{-\kappa n \epsilon^2},
\een
then
\ben
P \Big( \Big \lvert \abs{X_n} - \abs{c}  \Big\lvert \geq \epsilon \Big) \leq e^{-\kappa n \abs{c}^2 \epsilon^2}.
\een
\end{applem}


\begin{applem}[Concentration of Powers]
\label{powers}
Assume $c \neq 0$ and $\e \in(0, 1]$.  Then for any integer $k \geq 2$, if
\ben
P\Big(\lvert X_n - c  \lvert \geq \epsilon \Big) \leq e^{-\kappa n \epsilon^2},
\een
then
\ben
P\Big( \lvert X_n^k - c^k \lvert \geq \epsilon \Big) \leq e^{ {-\kappa n \e^2}/[{(1+\abs{c})^k -\abs{c}^k}]^2}.
\een
\end{applem}

\begin{applem}[Concentration of Scalar Inverses]
\label{inverses} Assume $c \neq 0$ and $\e \in(0, 1)$. If
\ben
P\Big( \lvert X_n - c  \lvert \geq \epsilon \Big) \leq e^{-\kappa n \epsilon^2},
\een
then
\ben
P\Big( \lvert X_n^{-1} - c^{-1} \lvert \geq \epsilon \Big) \leq 2 e^{-n \kappa \e^2 c^2 \min\{c^2, 1\}/4}.
\een
\end{applem}

\begin{applem}
\label{lem:normalconc}
For a standard Gaussian random variable $Z$ and  $\e > 0$,
$P( \abs{Z} \geq \e) \leq 2e^{-\frac{1}{2}\e^2}$.
\end{applem}

\begin{applem}
($\chi^2$-concentration.)
For  $Z_i$, $i \in [n]$ that are i.i.d. $\sim \mc{N}(0,1)$, and  $0 \leq \e \leq 1$,
\[P\Big(\Big \lvert \frac{1}{n}\sum_{i=1}^n Z_i^2 - 1\Big \lvert \geq \e \Big) \leq 2e^{-n \e^2/8}.\]
\label{subexp}
\end{applem}


\begin{applem}
\cite{BLMConc} Let $X$ be a centered sub-Gaussian random variable with variance factor $\nu$, i.e., $\ln \expec[e^{tX}] \leq \frac{t^2 \nu}{2}$, $\forall t \in \mathbb{R}$. Then $X$ satisfies:
\begin{enumerate}
\item For all $x> 0$, $P(X > x)  \vee P(X <-x) \leq e^{-\frac{x^2}{2\nu}}$, for all $x >0$.
\item For every integer $k \geq 1$,
$\expec[X^{2k}] \leq 2 (k !)(2 \nu)^k \leq (k !)(4 \nu)^k.$
\end{enumerate}
\label{lem:subgauss}
\end{applem}

\section{Other Useful Lemmas}
\label{app:other}

In this section, when the results are standard, they are presented without proof.

\begin{applem}\cite[Fact 7]{RushV18}
Let $u \in \mathbb{R}^N$ be a deterministic vector and let $\tilde{A} \in \mathbb{R}^{n \times N}$ be a matrix with independent  $\mc{N}(0, 1/n)$ entries. Moreover, 
let $\mc{W}$  be a $d$-dimensional subspace of $\mathbb{R}^n$ for $d \leq n$. Let $(w_1, ..., w_d)$ be an orthogonal basis of $\mc{W}$ with $\norm{w_\ell}^2 = n$ for $\ell \in [d]$, and let  $\mathsf{P}^{\parallel}_\mc{W}$ denote  the orthogonal projection operator onto $\mc{W}$.  Then for $D = [w_1\mid \ldots \mid w_d]$, we have 
$\mathsf{P}^{\parallel}_{\mc{W}} \tilde{A} u \overset{d}{=}   \frac{\norm{u}}{\sqrt{n}} \mathsf{P}^{\parallel}_{\mc{W}} Z_u\overset{d}{=}  \frac{\norm{u}}{\sqrt{n}} Dx$ where $x \in \mathbb{R}^d$ is a random vector with i.i.d.\ $\mc{N}(0, 1/n)$ entries. 
\label{fact:gauss_p0}
\end{applem}

\begin{applem}
(Stein's lemma.)
For zero-mean jointly Gaussian random variables $Z_1, Z_2$, and any function $f:\mathbb{R} \to \mathbb{R}$ for which $\expec[Z_1 f(Z_2)]$ and $\expec[f'(Z_2)]$  both exist, we have $\expec[Z_1 f(Z_2)] = \expec[Z_1Z_2] \expec[f'(Z_2)]$.
\label{fact:stein}
\end{applem}

\begin{applem}
(Products of Lipschitz Functions are PL(2).) 
Let $f,g:\mathbb{R}^p \to \mathbb{R}$ be Lipschitz continuous.  Then the product function $h:\mathbb{R}^p \to \mathbb{R}$ defined as $h(x) := f(x) g(x)$ is PL(2).
\label{lem:Lprods}
\end{applem}

\begin{applem} 
Let $\Lambda$ be defined in \eqref{eq:def_Lambda}. For each $r=1,\ldots,t$, let $\tau_r>0$ be a constant and let $Z_r\in\mathbb{R}^\Lambda$ have i.i.d.\ standard normal entries. Suppose  $f:\mathbb{R}^{|\Lambda|(t+1)}\to\mathbb{R}$ is PL(2) with PL constant $L$, then the function $\tilde{f}:\mathbb{R}^{\Lambda}\to\mathbb{R}$ defined as $\tilde{f}(s):=\mathbb{E}_{Z_1,\ldots, Z_t}[f( \tau_1 Z_1,\ldots, \tau_t Z_t,s)]$ is PL(2).  
\label{lem:expZ}
\end{applem}
\begin{proof}
Take arbitrary $x,y\in\mathbb{R}^\Lambda$, 
\begin{align*}
&|\tilde{f}(x) - \tilde{f}(y)|\\
& = \Big \lvert\mathbb{E}\Big[f(\tau_1 Z_1,\ldots, \tau_t Z_t,x) - f(\tau_1 Z_1,\ldots, \tau_t Z_t,y)\Big] \Big \lvert \\
&\overset{(a)}{\leq}\mathbb{E}\Big[ \Big \lvert f(\tau_1 Z_1,\ldots, \tau_t Z_t,x) - f(\tau_1 Z_1,\ldots,\tau_t Z_t,y) \Big \lvert \Big]\\
&\overset{(b)}{\leq}\mathbb{E}[L(1+ 2\sum_{r=1}^t \tau_r\|Z_r\|  + \|x\|+ \|y\|)\|x-y\|] \\
&\overset{(c)}{\leq} L(1+ 2|\Lambda|\sqrt{\frac{2}{\pi}}\sum_{r=1}^t \tau_r + \|x\|+ \|y\|)\|x-y\|\\
& \leq L\Big(1+2|\Lambda|\sqrt{\frac{2}{\pi}}\sum_{r=1}^t \tau_r\Big)(1+ \|x\| + \|y\|)\|x-y\|.
\end{align*}
In the above, step $(a)$ follows from Jensen's inequality, step $(b)$ holds since $f$ is PL(2) and using the triangle inequality, and step $(c)$ follows from $\mathbb{E}\|Z_r\|\leq \sum_{i\in\Lambda} \mathbb{E} \abs{[Z_r]_i} = |\Lambda|\sqrt{\frac{2}{\pi}}$.
\end{proof}

\begin{applem}
Let $\Gamma$ and $\Lambda$ be  as defined in \eqref{eq:def_Gamma} and \eqref{eq:def_Lambda}, and let $\Lambda_i$ be $\Lambda$ translated to be centered at $i$ for each $i\in\Gamma$. Let $f:\mathbb{R}^{\Lambda}\to \mathbb{R}$ be a PL(2) function with constant $L$ and define $\tilde{f}_i:\mathbb{R}^{\Lambda_i\cap\Gamma}\to\mathbb{R}$ as $\tilde{f}_i(v):=f(\mc{T}_i(v))$, where $\mc{T}_i:\mathbb{R}^{\Lambda_i\cap\Gamma}\to\mathbb{R}^{\Lambda}$ is defined in \eqref{eq:def_T}. Then $\tilde{f}_i$ is PL(2) for all $i\in\Gamma$.
\label{lem:PLwithAvg}
\end{applem}

\begin{proof}
Let $i$ be an arbitrary but fixed index in $\Gamma$. Let $d=|\Lambda|$, $a_i:=|\Lambda_i\cap\Gamma|$, and $b_i=d-a_i$, so that $b_i$ counts the number of ``missing" entries in $\Lambda_i$. For any $x,y\in\mathbb{R}^{\Lambda_i\cap\Gamma}$, we have that 
\begin{align*}
&|\tilde{f}(x) - \tilde{f}(y)|^2= |f(\mc{T}_{i}(x)) - f(\mc{T}_{i}(y))|^2\\
& \overset{(a)}{\leq} 3 L^2 (1 + \|\mc{T}_{i}(x)\|^2 + \|\mc{T}_{i}(y)\|^2)\|\mc{T}_{i}(x) - \mc{T}_{i}(y)\|^2\\
&\overset{(b)}{=} 3 L^2  \Big[b_i\Big(\frac{1}{a_i}\sum_{j\in\Lambda_i\cap\Gamma} x_j-y_j\Big)^2 + \norm{x-y}^2\Big] \times\\
&\,\, \Big[1 + b_i \Big[\frac{1}{a_i}\sum_{j\in\Lambda_i\cap\Gamma} x_j \Big]^2 + \norm{x}^2 +  b_i \Big[\frac{1}{a_i}\sum_{j\in\Lambda_i\cap\Gamma} y_j\Big]^2 + \norm{y}^2\Big] \\
&\overset{(c)}{\leq} 3 L^2  \Big[\frac{b_i + a_i}{a_i}\norm{x-y}^2 \Big]  \Big[ 1+ \frac{b_i + a_i}{a_i}\norm{x}^2  + \frac{b_i + a_i}{a_i}\norm{y}^2 \Big]\\
&= \frac{3 d L^2 }{a_i}\Big(\frac{a_i}{d}+\|x\|^2 + \|y\|^2\Big)\|x-y\|^2\\
& \leq \frac{3dL^2}{a_i}(1+\|x\| + \|y\|)^2\|x-y\|^2,
\end{align*}
where step $(a)$ follows from the pseudo-Lipschitz property of $f$ and Lemma \ref{lem:squaredsums}, step $(b)$ from our definition of $\mc{T}_{i}$ in \eqref{eq:def_T}, and step $(c)$ from Lemma \ref{lem:squaredsums}. 
\end{proof}

\begin{applem}
For any scalars $a_1, ..., a_t$ and positive integer $m$, we have  $(\abs{a _1} + \ldots + \abs{a_t})^m \leq t^{m-1} \sum_{i=1}^t \abs{a_i}^m$.
Consequently, for any vectors $\un{u}_1, \ldots, \un{u}_t \in \mathbb{R}^N$, $\norm{\sum_{k=1}^t \un{u}_k}^2 \leq t \sum_{k=1}^t \norm{\un{u}_k}^2$.
\label{lem:squaredsums}
\end{applem}

\section{Concentration with Dependencies} 
\label{app:conc_dependent}

We first state a concentration result, existing in the literature, for functions acting on random fields that satisfy the Dobrushin uniqueness condition in Lemma \ref{lem:kulske}. Then we use Lemma \ref{lem:kulske} to obtain Lemma \ref{lem:PL_MRF_conc}, which is needed to prove $\mc{H}_t (b)$.

\begin{applem}\cite[Theorem 1]{kulske2003}
\label{lem:kulske}
Suppose that the random field $X=(X_i)_{i\in\Gamma}$ taking values in $E^{\Gamma}$ is distributed according to a Gibbs measure $\mu$ that obeys the Dobrushin uniqueness condition with Dobrishin constant $c$, and the transposed Dobrishin uniqueness condition with constant $c^*$. Suppose that $F$ is a real function on $E^{\Gamma}$ with $\mathbb{E}[\exp(tF(X))]<\infty$ for all real $t$. Then we have for all $r\geq 0$,
\begin{equation}
P\Big(F(X) - \mathbb{E}[F(X)] > r\Big) \leq \exp\Big(-\frac{r^2}{2}\frac{(1-c)(1-c^*)}{\|\underline{\delta}(F)\|_{\ell^2}^2}\Big).
\end{equation}
Here $\underline{\delta}(F):=(\delta_i(F))_{i\in\Gamma}$ is the variation vector of $F$, where $\delta_i(F):=\sup_{\xi,\xi';\xi_{i^c}=\xi'_{i^c}} \abs{F(\xi)-F(\xi')}$ denotes the variation of $F$ at the site $i$. Its $\ell^2$-norm is defined as $\|\underline{\delta}(F)\|_{\ell^2}^2:=\sum_{i\in\Gamma}(\delta_i(F))^2$. If this norm is infinite, then the statement is empty (and thus correct).
\end{applem}

\begin{applem}
\label{lem:PL_MRF_conc}
Let $\Gamma$ and $\Lambda$ be defined in \eqref{eq:def_Gamma} and \eqref{eq:def_Lambda}, respectively,
and let $X=(X_i)_{i\in\Gamma}$ be a stationary Markov random field with a unique Gibbs distribution measure $\mu$ on $E^\Gamma\subset\mathbb{R}^{\Gamma}$. Assume that $\mu$ satisfies the Dobrushin uniqueness condition and the transposed Dobrushin uniqueness condition with constants $c$ and $c^*$, respectively. 
Suppose that the state space $E$ is bounded, meaning that there exists an $M$ such that $|x|\leq M$, for all $x \in E$.  
Let $f_i:\mathbb{R}^{\Lambda_i\cap\Gamma}\rightarrow\mathbb{R}$, where $\Lambda_i$ is $\Lambda$ being translated to be centered at location $i\in\Gamma$, be a PL(2) function with pseudo-Lipschitz constant $L_i$, for all $i\in\Gamma$. Then for all $\e\in(0,1)$ there exist $K,\kappa>0$, such that
\begin{equation}
P\Big(\Big \lvert \frac{1}{|\Gamma|}\sum_{i\in\Gamma} \Big( f_i(X_{\Lambda_i\cap\Gamma}) - \mathbb{E}[f_i(X_{\Lambda_i\cap\Gamma})]\Big) \Big \lvert \geq \e\Big) \leq K e^{-\kappa |\Gamma| \e^2}.
\end{equation}
\end{applem}

\begin{proof}
Let the function $F$ in Lemma \ref{lem:kulske} be defined as $F(X):=\sum_{i\in\Gamma}f_i(X_{\Lambda_i\cap\Gamma})$. 
In order to apply Lemma \ref{lem:kulske}, we need to calculate $\|\underline{\delta}(F)\|_{\ell^2}^2$. Let $d:=|\Lambda|$ and $L:=\max_{i\in\Gamma} L_i$, then we have that 
\begin{align*}
&\delta_i(F) = \sup_{\substack{\xi,\xi'\in E^\Gamma\\ \xi_{i^c}=\xi_{i^c}'}} \abs{F(\xi) - F(\xi')}\\
& = \sup_{\substack{\xi,\xi'\in E^\Gamma\\ \xi_{i^c}=\xi_{i^c}'}} \Big \lvert\sum_{j:i\in\Lambda_j\cap\Gamma} f_j(\xi_{\Lambda_j\cap\Gamma}) - f_j(\xi'_{\Lambda_j\cap\Gamma}) \Big \lvert\\
&\overset{(a)}{\leq}\sup_{\substack{\xi,\xi'\in E^\Gamma\\ \xi_{i^c}=\xi_{i^c}'}} \sum_{j:i\in\Lambda_j\cap\Gamma}\Big[L_j(1+\|\xi_{\Lambda_j\cap\Gamma}\|+\|\xi'_{\Lambda_j\cap\Gamma}\|) \times \\
&\hspace{5.3cm}\|\xi_{\Lambda_j\cap\Gamma} - \xi'_{\Lambda_j\cap\Gamma}\|\Big]\\
&\overset{(b)}{\leq}dL(1+2\sqrt{d}M)2\sqrt{d}M.
\end{align*}
In the above, step $(a)$ uses the triangle inequality and the pseudo-Lipschitz property of $f$. Step $(b)$ follows from the fact that $|x|\leq M,$ for all $ x\in E$ and that $L_j\leq L$ for all $j\in\Gamma$. 
Therefore,
\begin{align*}
\|\underline{\delta}(F)\|_{\ell^2}^2 \leq |\Gamma|(dL(1+2\sqrt{d}M)2\sqrt{d}M)^2.
\end{align*}
Now applying Lemma \ref{lem:kulske}, we have 
\begin{align*}
&P\Big( \Big \lvert \frac{1}{|\Gamma|}\sum_{i\in\Gamma} \Big( f_i(X_{\Lambda_i\cap\Gamma}) - \mathbb{E}\Big[f_i(X_{\Lambda_i\cap\Gamma}) \Big]\Big) \Big \lvert \geq \e\Big)\\
& \leq 2 \exp\Big(-\frac{|\Gamma|\e^2(1-c)(1-c^*)}{(dL(1+2\sqrt{d}M)2\sqrt{d}M)^2}\Big).
\end{align*}
\end{proof}

Lemma \ref{lem:exp_PL_subgauss_vector_conc} provides a technical result about pseudo-Lipschitz functions with sub-Gaussian inputs, which will be used to prove Lemma \ref{lem:PL_overlap_gauss_conc_ext}. 

\begin{applem}\cite[Lemma D.2]{MaRushBaron17}
\label{lem:exp_PL_subgauss_vector_conc}
Let $X \in\mathbb{R}^d$ be a random vector whose entries have a sub-Gaussian marginal distribution with variance factor $\nu$ as in Lemma \ref{lem:subgauss}.  Let $\tilde{X}$ be an independent copy of $X$.  If $f:\mathbb{R}^d\rightarrow \mathbb{R}$ is a PL(2) function with pseudo-Lipschitz constant $L$, then the expectation $\mathbb{E}[\exp(rf(X))]$ satisfies the following for $0<r<[5L(2d\nu+24d^2\nu^2)^{1/2}]^{-1}$,
\begin{align}
&\mathbb{E}[e^{rf(X)}] \leq \mathbb{E}[e^{r(f(X)-f(\tilde{X})}]\nonumber\\
& \leq [1-25 r^2 L^2(d\nu+12d^2\nu^2)]^{-1}\leq e^{50r^2 L^2(d \nu+12d^2\nu^2)}.
\label{eq:threebounds}
\end{align}
\end{applem}

Lemma \ref{lem:PL_overlap_gauss_conc_ext} provides a concentration inequality for sums of pseudo-Lipschitz functions acting on overlapping subsets of jointly Gaussian random variables.

\begin{applem}
\label{lem:PL_overlap_gauss_conc_ext}
Let $\Gamma$ and $\Lambda$ be defined as in \eqref{eq:def_Gamma} and \eqref{eq:def_Lambda}.
For each $r=1,\ldots,t$, let $(Z^r_i)_{i\in\Gamma}$ have i.i.d.\ $\mc{N}(0,1)$ entries, and for all $r,s=1,\ldots,t$ and $i\neq j$, $Z^r_i$ is independent of $Z^s_j$. Moreover, for each $i\in\Gamma$, let $(Z_i^1,\ldots, Z_i^t)$ be jointly Gaussian with covariance matrix $K\in\mathbb{R}^{t\times t}$.

For each $i\in\Gamma$, define $Y_i:=(Z^1_{\Lambda_i\cap\Gamma},\ldots,Z^t_{\Lambda_i\cap\Gamma})$, where $\Lambda_i$ is $\Lambda$ translated to be centered at location $i\in\Gamma$. Let $f_i:\mathbb{R}^{|\Lambda_i\cap\Gamma|t}\rightarrow \mathbb{R}$ be a PL(2) function for all $i\in\Gamma$. Then for all $\e\in(0,1)$, there exist $K,\kappa>0$ such that
\begin{equation}
P\Big(\Big \lvert \frac{1}{|\Gamma|}\sum_{i\in\Gamma} \Big(f_i(Y_i) - \mathbb{E}[f_i(Y_i)]\Big) \Big \lvert \geq \e \Big)\leq K e^{-\kappa |\Gamma|\e^2}.
\end{equation}
\end{applem}

\begin{proof} 
In the following, we prove the case for $p=2$ and the proof for other dimensions follows similarly. Without loss of generality, let $i=(i_1,i_2)$ and $\Gamma:=\{(i_1,i_2)\}_{1\leq i_1,i_2\leq n}$, hence $|\Gamma|=n^2$. Further, assume without loss of generality that $\mathbb{E}[f_i(Y_i)]=0,$ for all $i \in \Gamma$. In what follows, we demonstrate
the upper-tail bound:
\begin{equation}
P\Big(\frac{1}{|\Gamma|}\sum_{i\in\Gamma}f_i(Y_i)\geq \e \Big)\leq Ke^{-\kappa |\Gamma|\e^2},
\label{eq:upper.tail}
\end{equation}
and the lower-tail bound follows similarly.  Together they provide the desired result.

Using the Cram\'{e}r-Chernoff method, for all $r>0$,
\begin{align}
P\Big(\frac{1}{|\Gamma|}\sum_{i\in\Gamma}f_i(Y_i)\geq \e \Big) &=P\Big(e^{r\sum_{i\in\Gamma}f_i(Y_i)}\geq e^{r|\Gamma| \e}\Big)\nonumber \\
&\leq e^{-r |\Gamma| \e}\mathbb{E}[e^{r\sum_{i\in \Gamma} f_i(Y_i)}].
\label{eq:CC1}
\end{align}
Let $d^2=|\Lambda|$ and $L_i$ be the pseudo-Lipschitz parameters associated with functions $f_i$  for $i\in \Gamma$ and define $L := \max_{i\in\Gamma} L_i$.  In the following, we will show that for $0<r<(10Ld^2\sqrt{2td^2+24t^2d^4})^{-1}$,
\begin{equation}
\mathbb{E}[e^{r\sum_{i\in\Gamma}f_i(Y_i)}] \leq \exp(\kappa' |\Gamma| r^2), 
\label{eq:lem1_0}
\end{equation}
where $\kappa'$ is any constant that satisfies $\kappa'\geq 450L^2d^2(d^2+12d^4)$. Then plugging \eqref{eq:lem1_0} into \eqref{eq:CC1}, we can obtain the desired result in \eqref{eq:upper.tail}: 
\begin{equation*}
P\Big(\frac{1}{|\Gamma|}\sum_{i\in\Gamma}f_i(Y_i)\geq \e \Big) \leq \exp(-|\Gamma|(r\e - \kappa'  r^2)).
\end{equation*}
Set $r = \e/(2\kappa')$, which is the choice that maximizes the term in the exponent in the above, i.e.\ it maximizes $(r\e - \kappa'  r^2)$  over $r$.  We can ensure that $\forall \e \in (0,1)$, $r$ falls within the region required in \eqref{eq:lem1_0} by choosing $\kappa'$ large enough.

We now show \eqref{eq:lem1_0}.  Define index sets 
\begin{align*}
I_{j_1,j_2}:=\Big\{(j_1+k_1 d,j_2+k_2 d) \, \Big\vert \, &k_1=0,...,\lfloor \frac{n-j_1}{d} \rfloor,\\
& k_2=0,...,\lfloor \frac{n-j_2}{d} \rfloor \Big.\Big\}
\end{align*}
for $j_1,j_2=1,...,d$, and let $C_{j_1,j_2}$ denote the cardinality of $I_{j_1,j_2}$. 
We notice that for any fixed $(j_1,j_2)$, the $Y_{i_1,i_2}$'s are i.i.d.\ for all $(i_1,i_2)\in I_{j_1,j_2}$.  Also, we have $\Gamma = \cup_{j_1,j_2=1}^{d} I_{j_1,j_2}$, and $I_{j_1,j_2} \cap I_{s_1,s_2}=\emptyset$, for $(j_1,j_2)\neq (s_1,s_2)$, making the collection $I_{1,1}, I_{1,2}, \ldots, I_{d,d}$ a partition of $\Gamma$. Therefore,
\begin{align*}
\sum_{i\in\Gamma} f_i(Y_i) &= \sum_{j_1,j_2=1}^{d} \sum_{(i_1,i_2)\in I_{j_1,j_2}} f_{i_1,i_2}(Y_{i_1,i_2}) \\
&= \sum_{j_1,j_2=1}^{d} p_{j_1,j_2}\cdot\frac{1}{p_{j_1,j_2}}\sum_{(i_1,i_2)\in I_{j_1,j_2}} f_{i_1,i_2}(Y_{i_1,i_2}),
\end{align*}
where $0 < p_{j_1,j_2} <1$ are probabilities satisfying $\sum_{j_1,j_2=1}^{d} p_{j_1,j_2}=1$.  Using the above,
\begin{align}
&\mathbb{E}[\exp(r\sum_{i\in\Gamma} f_i(Y_i))]\nonumber\\
&=\mathbb{E}\Big[\exp\Big(\sum_{j_1,j_2=1}^{d} p_{j_1,j_2} \cdot \frac{r}{p_{j_1,j_2}}\sum_{(i_1,i_2)\in I_{j_1,j_2}} f_{i_1,i_2}(Y_{i_1,i_2})\Big)\Big]\nonumber\\
& \overset{(a)}{\leq} \sum_{j_1,j_2=1}^{d} p_{j_1,j_2} \mathbb{E}\Big[\exp\Big(\frac{r}{p_{j_1,j_2}}\sum_{(i_1,i_2)\in I_{j_1,j_2}}f_{i_1,i_2}(Y_{i_1,i_2})\Big)\Big]  \nonumber \\
&\overset{(b)}{=}\sum_{j_1,j_2=1}^{d} p_{j_1,j_2}\prod_{(i_1,i_2)\in I_{j_1,j_2}}\mathbb{E}\Big[\exp\Big(\frac{r}{p_{j_1,j_2}}f_{i_1,i_2}(Y_{i_1,i_2})\Big)\Big]\nonumber\\
& \overset{(c)}{\leq} \sum_{j_1,j_2=1}^{d}\! p_{j_1,j_2} \exp\Big(\frac{50C_{j_1,j_2}L^2r^2 (td^2+12t^2d^4)}{p_{j_1,j_2}^2}\Big),
\label{eq:lem1_1}
\end{align}
where step $(a)$ follows from Jensen's inequality, step $(b)$ from the independence of $Y_{i_1,i_2}$'s for $(i_1,i_2)\in I_{j_1,j_2}$,
and step $(c)$ from Lemma~\ref{lem:exp_PL_subgauss_vector_conc} with variance factor $\nu=1$ and restriction 
\begin{equation}
0 < r < \Big(5L\sqrt{2td^2 + 24t^2d^4}\Big)^{-1} \min_{(j_1,j_2)} p_{j_1,j_2}.
\label{eq:r_region}
\end{equation}

Let $p_{j_1,j_2}=\sqrt{C_{j_1,j_2}}/C$, where $C=\sum_{j_1,j_2=1}^{d} \sqrt{C_{j_1,j_2}}$, ensuring  $\sum_{j_1,j_2=1}^{d} p_{j_1,j_2} = 1$. Then,
\begin{align*}
&\sum_{j_1,j_2=1}^{d} p_{j_1,j_2} \exp\Big(\frac{50 C_{j_1,j_2}L^2r^2 (td^2+12t^2d^4)}{p_{j_1,j_2}^2}\Big)\\
&= e^{50C^2L^2r^2 (td^2+12t^2d^4)} \overset{(a)}{\leq} e^{450d^2n^2L^2(td^2+12t^2d^4)r^2} \leq e^{\kappa' |\Gamma| r^2},
\end{align*}
whenever $\kappa' \geq 450d^2L^2(td^2+12t^2d^4)$. In the above, step $(a)$ follows from:
\begin{align*}
C^2&=\Big(\sum_{j_1,j_2=1}^{d}\sqrt{C_{j_1,j_2}}\Big)^2\\
& =\sum_{j_1,j_2=1}^{d} C_{j_1,j_2} +\sum_{j_1,j_2=1}^{d}\sum_{(k_1,k_2)\neq (j_1,j_2)}\sqrt{C_{j_1,j_2}C_{k_1,k_2}}\\
&\overset{(b)}{\leq} n^2 +d^2(d^2-1) C_{1,1}\\
&  \overset{(c)}{\leq} d^2n^2 + 4(d^2-1)(d^2+nd) <9d^2n^2, 
\end{align*} 
where step $(b)$ holds because $C_{1,1}= \max_{(j_1,j_2)} C_{j_1,j_2}$ and step $(c)$ holds because $C_{1,1}=(\lfloor\frac{n-1}{d}\rfloor+1)^2 \leq (\frac{n}{d}+2)^2$.  
Finally, we consider the effective region for $r$ as required in \eqref{eq:r_region}.
Notice that 
\begin{align*}
\min_{(j_1,j_2)} p_{j_1,j_2}& = \frac{\sqrt{C_{d,d}}}{C} = \frac{\sqrt{C_{d,d}}}{\sum_{j_1,j_2=1}^{d} \sqrt{C_{j_1,j_2}}} \geq  \frac{\sqrt{C_{d,d}}}{d^2\sqrt{C_{1,1}}} \\
&= \frac{1}{d^2}\frac{\lfloor\frac{n-d}{d}\rfloor+1}{\lfloor\frac{n-1}{d}\rfloor+1}\\
&=\frac{1}{d^2}\frac{\lfloor\frac{n}{d}\rfloor}{\lfloor\frac{n-1}{d}\rfloor+1}\geq  \frac{1}{d^2}\frac{\lfloor\frac{n-1}{d}\rfloor}{\lfloor\frac{n-1}{d}\rfloor+1}\geq \frac{1}{2d^2}.
\end{align*}
Hence, if we require $0<r< (10Ld^2\sqrt{2td^2+24t^2d^4})^{-1}$, then \eqref{eq:r_region} is satisfied.
\end{proof}

\bibliographystyle{IEEEtran}
\bibliography{../NSdenoisers}

%
%
%

%


\begin{IEEEbiographynophoto}{Yanting Ma}
 (S’13–M’18) received the Ph.D. degree
in electrical engineering from North Carolina State
University, Raleigh, NC, USA, in 2017.
She is currently a Research Scientist with Mitsubishi Electric Research Laboratories (MERL),
Cambridge, MA, USA. Prior to joining MERL, she
was a Postdoctoral Associate with Boston University,
Boston, MA, USA. Her research interests include algorithm design and analysis for inverse problems using statistical inference and optimization techniques. Additionally, she
is generally interested in applied probability and convex analysis.
\end{IEEEbiographynophoto}

\begin{IEEEbiographynophoto}{Cynthia Rush}
 (S'14-M’16) is currently an Assistant Professor with the Department of Statistics, Columbia University, New York, NY, USA.  She received the M.A. and Ph.D. degrees both in Statistics from Yale University, New Haven, CT, USA, in 2016 and a B.A. degree in mathematics from the University of North Carolina, Chapel Hill, NC, USA, in 2010. Her research interests include statistics, information theory, machine learning, and applied probability.
\end{IEEEbiographynophoto}

\begin{IEEEbiographynophoto}{Dror Baron}
(S'99--M'03--SM'10) received the B.Sc. (summa cum laude) and
M.Sc. degrees from the Technion -- Israel Institute of Technology, Haifa, Israel,
in 1997 and 1999, and the Ph.D. degree from the University of Illinois at Urbana-Champaign
in 2003, all in electrical engineering.
From 1997 to 1999, Dr. Baron worked at Witcom Ltd. in modem design.
From 1999 to 2003, he was a research assistant at the University of Illinois at
Urbana-Champaign, where he was also a Visiting Assistant Professor in 2003.
From 2003 to 2006, he was a Postdoctoral Research Associate in the Department
of Electrical and Computer Engineering at Rice University, Houston, TX. From
2007 to 2008, he was a quantitative financial analyst with Menta Capital, San
Francisco, CA, and from 2008 to 2010 a Visiting Scientist in the Department
of Electrical Engineering at the Technion -- Israel Institute of Technology,
Haifa. Since 2010,  Dr. Baron has been with the Electrical and Computer Engineering
Department at North Carolina State University, where he is currently an
Associate Professor.

Dr. Baron's research interests combine information theory, signal processing,
and fast algorithms; in recent years, he has focused on information theoretic
aspects of compressed sensing.
Dr. Baron was a recipient of the 2002 M. E. Van Valkenburg Graduate Research
Award, and received honorable mention at the Robert Bohrer Memorial Student
Workshop in April 2002, both at the University of Illinois. He also participated
from 1994 to 1997 in the Program for Outstanding Students, comprising the top
0.5\% of undergraduates at the Technion.
\end{IEEEbiographynophoto}






\end{document}